%% file: si-sat.tex
\documentclass[sigconf, nonacm]{acmart}
\input{preamble}
\input{newcommands}

\newcommand\vldbdoi{10.14778/3583140.3583145}
\newcommand\vldbpages{1264 - 1276}
\newcommand\vldbvolume{16}
\newcommand\vldbissue{6}
\newcommand\vldbyear{2023}
\newcommand\vldbauthors{\authors}
\newcommand\vldbtitle{\shorttitle}
\newcommand\vldbavailabilityurl{https://github.com/hengxin/PolySI-PVLDB2023-Artifacts}
\newcommand\vldbpagestyle{empty}

\begin{document}
\title{Efficient Black-box Checking of Snapshot Isolation in Databases}

\thanks{$^\dagger$Joint first authors.}

 \author{Kaile Huang$^\dagger$}
 \affiliation{%
   \institution{State Key Laboratory for Novel Software Technology\\
    Nanjing University}
 }

  \author{Si Liu$^\dagger$}
  \affiliation{%
   \institution{ETH Zurich}
  }

 \author{Zhenge Chen}
 \author{Hengfeng Wei}
 \authornote{Corresponding author.}
 \affiliation{%
  \institution{State Key Laboratory for Novel Software Technology\\
  Software Institute\\Nanjing University}
}

  \author{David Basin}
\affiliation{%
	\institution{ETH Zurich}
}

 
 
 
   \author{Haixiang Li}
   \author{Anqun Pan}
 \affiliation{%
   \institution{Tencent Inc.}
 }
 





\input{sections/abstract}

\maketitle

\pagestyle{\vldbpagestyle}
\begingroup\small\noindent\raggedright\textbf{PVLDB Reference Format:}\\
\vldbauthors. \vldbtitle. PVLDB, \vldbvolume(\vldbissue): \vldbpages, \vldbyear.\\
\href{https://doi.org/\vldbdoi}{doi:\vldbdoi}
\endgroup
\begingroup
\renewcommand\thefootnote{}\footnote{\noindent
This work is licensed under the Creative Commons BY-NC-ND 4.0 International License. Visit \url{https://creativecommons.org/licenses/by-nc-nd/4.0/} to view a copy of this license. For any use beyond those covered by this license, obtain permission by emailing \href{mailto:info@vldb.org}{info@vldb.org}. Copyright is held by the owner/author(s). Publication rights licensed to the VLDB Endowment. \\
\raggedright Proceedings of the VLDB Endowment, Vol. \vldbvolume, No. \vldbissue\ %
ISSN 2150-8097. \\
\href{https://doi.org/\vldbdoi}{doi:\vldbdoi} \\
}\addtocounter{footnote}{-1}\endgroup

\ifdefempty{\vldbavailabilityurl}{https://github.com/hengxin/PolySI-PVLDB2023-Artifacts}{
\vspace{.3cm}
\begingroup\small\noindent\raggedright\textbf{PVLDB Artifact Availability:}\\
The source code, data, and/or other artifacts have been made available at \url{\vldbavailabilityurl}.
\endgroup
}

\input{sections/intro-nobi}
\input{sections/problem-definition}
\input{sections/g-polygraph}
\input{sections/si-sat}
\input{sections/experiments}
\input{sections/discussion}
\input{sections/related-work-nobi}

\input{sections/conclusion}
\input{sections/ack}

\clearpage
\balance
\bibliographystyle{ACM-Reference-Format}
\bibliography{si-sat}

\clearpage
\appendix
\input{sections/appendix}

\end{document}

%% file: preamble.tex
\usepackage{todonotes}

\usepackage{xcolor}
\usepackage{wrapfig}
\usepackage{graphicx}
\usepackage{array}
\usepackage{hyperref}

\usepackage{multirow}
\usepackage{multicol}

\usepackage{caption}
\usepackage{subcaption}

\usepackage{enumerate}
\usepackage{mathtools}

\usepackage{graphicx}
\usepackage[export]{adjustbox}

\usepackage{bbding} 

\usepackage{balance}

\usepackage{pgfplots}
\usetikzlibrary{patterns}

\newtheorem{theorem}{Theorem}

\usepackage{varwidth}
\usepackage{algorithm}
\usepackage[noend]{algpseudocode}
\algnewcommand{\Break}{\textbf{Break}}

\makeatletter
\newsavebox{\measure@tikzpicture}
\NewEnviron{scaletikzpicturetowidth}[1]{%
  \def\tikz@width{#1}%
  \def\tikzscale{1}\begin{lrbox}{\measure@tikzpicture}%
  \BODY
  \end{lrbox}%
  \pgfmathparse{#1/\wd\measure@tikzpicture}%
  \edef\tikzscale{\pgfmathresult}%
  \BODY
}
\makeatother

\makeatletter
\newcommand\resetstackedplots{
\makeatletter
\pgfplots@stacked@isfirstplottrue
\makeatother
\addplot [forget plot,draw=none] coordinates{(1,0) (2,0) (3,0) (4,0) (5,0) (6,0)};
}
\makeatother

%% file: newcommands.tex
\newcommand{\set}[1]{\{#1\}}
\newcommand{\bset}[1]{\big\{#1\big\}}
\newcommand{\Bset}[1]{\Big\{#1\Big\}}
\newcommand{\bbset}[1]{\bigg\{#1\bigg\}}
\newcommand{\tuple}[1]{\langle#1\rangle} 
\newcommand{\btuple}[1]{\big\langle#1\big\rangle}

\newcommand{\lbl}{\mathcal{L}}

\newcommand{\interpret}{\textsc{Interpret}}

\newcommand{\bfit}[1]{\emph{\textbf{#1}}} 
\newcommand{\True}{\textsf{True}}
\newcommand{\False}{\textsf{False}}

\newcommand{\code}[2]{#1:#2}
\newcommand{\linecode}[2]{#2}
\newcommand{\hStatex}[0]{\vspace{4pt}}

\newcommand{\npc}{\textsf{NP-complete}}

\newcommand{\pruning}{\mathcal{P}}

\newcommand{\false}{{\sf false}}

\newcommand{\Key}{{\sf Key}}
\newcommand{\keyxvar}{\mathit{x}}
\newcommand{\keyyvar}{\mathit{y}}
\newcommand{\keyzvar}{\mathit{z}}
\newcommand{\Val}{{\sf Val}}
\newcommand{\valvar}{\mathit{v}}

\newcommand{\h}{\mathcal{H}}
\newcommand{\uniqueval}{\textsf{UniqueValue}}

\newcommand{\opset}{\mathit{O}}

\newcommand{\Op}{{\sf Op}}

\newcommand{\readevent}{{\sf R}}

\newcommand{\writeevent}{{\sf W}}

\newcommand{\WriteTx}{{\sf WriteTx}}

\newcommand{\rel}[1]{\xrightarrow{#1}}
\newcommand{\comp}{\;;\;}
\newcommand{\po}{{\sf po}}

\newcommand{\intaxiom}{\textsc{Int}}

\newcommand{\abortedreads}{\textsc{AbortedReads}}

\newcommand{\intermediatereads}{\textsc{IntermediateReads}}

\newcommand{\si}{\textsc{SI}}

\newcommand{\ser}{\textsc{SER}}

\newcommand{\checksi}{\textsc{CheckSI}}
\newcommand{\createknowngraph}{\textsc{CreateKnownGraph}}
\newcommand{\generateconstraints}{\textsc{GenerateConstraints}}
\newcommand{\pruneconstraints}{\textsc{PruneConstraints}}
\newcommand{\encodeconstraints}{\textsc{SAT-Encode}}
\newcommand{\solveconstraints}{\textsc{MonoSAT-Solve}}
\newcommand{\solve}{\textsc{Solve}}

\newcommand{\g}{\mathit{G}}
\newcommand{\inducedgraph}{\mathit{I}}
\newcommand{\knowninducedgraph}{\mathit{KI}}
\newcommand{\eithervar}{\mathit{either}}
\newcommand{\orvar}{\mathit{or}}
\newcommand{\reachability}{\textsc{Reachability}}
\newcommand{\reachabilityvar}{\mathit{reachability}}

\newcommand{\vertex}{\mathit{V}}
\newcommand{\edges}{\mathit{E}}
\newcommand{\cons}{\mathit{C}}

\newcommand{\vvar}{\mathit{v}}
\newcommand{\precvar}{\mathit{prec}}

\newcommand{\consvar}{\mathit{cons}}

\newcommand{\inducedrule}{\mathcal{R}}
\newcommand{\BV}{\mathsf{BV}}
\newcommand{\Clause}{\mathsf{CL}}

\newcommand{\solver}{\mathit{solver}}
\newcommand{\instance}{\textsc{MonoSAT-Solver}}
\newcommand{\opid}{\iota}
\newcommand{\OpId}{\mathsf{OpId}}

\newcommand{\fromvar}{\mathit{from}}
\newcommand{\tovar}{\mathit{to}}
\newcommand{\typevar}{\mathit{type}}

\newcommand{\graphA}{\mathit{Dep}}
\newcommand{\graphB}{\mathit{AntiDep}}

\newcommand{\cycle}{\mathcal{C}}

\newcommand{\T}{\mathcal{T}}
\providecommand{\G}{}
\renewcommand{\G}{\mathcal{G}}
\renewcommand{\H}{\mathcal{H}}

\newcommand{\SO}{\textsf{SO}}
\newcommand{\WR}{\textsf{WR}}
\newcommand{\WW}{\textsf{WW}}
\newcommand{\RW}{\textsf{RW}}

\newcommand{\name}{PolySI}

\algblockdefx[upon]{CheckSER}{EndCheckSER}%
[1][Unknown]{\textbf{CheckSER} #1}%
{\textbf{end}}

\algblockdefx[upon]{CheckSI}{EndCheckSI}%
[1][Unknown]{\textbf{CheckSI} #1}%
{\textbf{end}}

\algblockdefx[upon]{ConstructEncoding}{EndConstructEncoding}%
[1][Unknown]{\textbf{ConstructEncoding} #1}%
{\textbf{end}}

\algblockdefx[upon]{GenConstraints}{EndGenConstraints}%
[1][Unknown]{\textbf{GenConstraints} #1}%
{\textbf{end}}

\algblockdefx[upon]{Prune}{EndPrune}%
[1][Unknown]{\textbf{Prune} #1}%
{\textbf{end}}

\algblockdefx[upon]{InferRWEdges}{EndInferRWEdges}%
[1][Unknown]{\textbf{InferRWEdges} #1}%
{\textbf{end}}

\algblockdefx[upon]{Coalesce}{EndCoalesce}%
[1][Unknown]{\textbf{Coalesce} #1}%
{\textbf{end}}

\algblockdefx[upon]{GenChainToChainEdges}{EndGenChainToChainEdges}%
[1][Unknown]{\textbf{GenChainToChainEdges} #1}%
{\textbf{end}}

\newcommand{\ct}{\mathit{c}}

\newcommand{\dependency}{\mathit{dep}}

\newcommand{\recoveredgraph}{\mathit{Graph_{recovered}}}
\newcommand{\finalizedgraph}{\mathit{Graph_{finalized}}}
\newcommand{\taggedgraph}{\mathit{Graph_{tagged}}}
\newcommand{\acs}{\mathit{Graph_{acs}}}
\newcommand{\testgraph}{\mathit{Graph_{tested}}}
\newcommand{\graphsize}{\mathit{graph\_size}}
\newcommand{\minimalExtendDeps}{\mathit{minimal\_extend\_deps}}
\newcommand{\minimalAcs}{\mathit{minimal\_acs}}
\newcommand{\extendDeps}{\mathit{extendDeps}}
\newcommand{\extendgraph}{\mathit{graph_{extended}}}

\newcommand{\Restore}{\textsc{Restore}}
\newcommand{\Finalize}{\textsc{Finalize}}
\newcommand{\Findacs}{\textsc{Find\_ACS}}
\newcommand{\Resolve}{\textsc{Resolve}}

%% file: sections/abstract.tex
\begin{abstract}

Snapshot isolation (SI) is a prevalent weak isolation level that avoids the performance penalty imposed by serializability and 
simultaneously prevents various undesired data anomalies.
Nevertheless, SI anomalies have recently been found in production  cloud databases that claim to provide  the SI guarantee.
Given the complex and often unavailable internals of such databases, a black-box SI checker is highly desirable.

\looseness=-1
In this paper we present \name{},  a novel black-box
 checker that efficiently checks SI and provides understandable counterexamples upon detecting violations.  
 \name{}  builds on a novel characterization of SI using generalized polygraphs (GPs),  for which we establish its 
 soundness and completeness.   \name{}  employs an SMT solver and also
 accelerates SMT solving by utilizing the compact constraint encoding of GPs
 and domain-specific optimizations for pruning constraints.
As demonstrated by our  extensive assessment,  \name{} successfully reproduces all of 2477 known SI anomalies,  detects novel SI violations in three production cloud  databases,  identifies their causes,
   outperforms the state-of-the-art black-box checkers under a wide range of workloads, and can scale up to large-sized workloads.

\end{abstract}

%% file: sections/intro-nobi.tex
\section{Introduction}
\label{section:intro}

Database systems are an essential building block of  many software systems and applications.
Transactional access to databases simplifies concurrent programming by providing an abstraction for executing concurrent computations on shared data in isolation~\cite{Bernstein:Book1986}.
The gold-standard isolation level,  {\it serializability} (\ser)~\cite{SER:JACM1979},
ensures that all transactions appear to execute serially,  one after another.
However,  providing \ser, especially in geo-replicated environments like modern cloud databases,  is computationally expensive~\cite{HAT:VLDB2013,DBLP:conf/osdi/LuSL20}.

Many databases  provide weaker guarantees for transactions to balance the trade-off  between  data consistency and system performance.  \emph{Snapshot isolation} (SI)~\cite{CritiqueANSI:SIGMOD1995} is one of the prevalent weaker isolation levels used in practice, which  avoids the performance penalty imposed by SER and
simultaneously prevents  undesired data anomalies such as fractured reads,  causality violations, and
lost updates 
\cite{AnalysingSI:JACM2018}.
In addition to classic centralized databases such as Microsoft SQL Server~\cite{SQLServer} and Oracle Database~\cite{Oracle},
SI is  supported by  numerous   \emph{production cloud} database systems like
Google's Percolator~\cite{Percolator:OSDI2010},
MongoDB~\cite{MongoDB}, TiDB~\cite{TiDB}, YugabyteDB~\cite{YugabyteDB}, Galera~\cite{maria-galera}, and Dgraph~\cite{dgraph}.

Unfortunately, as recently reported in \cite{Elle:VLDB2020, MongoDB-Jepsen, TiDB-Jepsen},
data anomalies have been found in several production cloud databases that claim to provide SI.\footnote{These anomalies, which we are also concerned with in this paper, are isolation violations purely in  \emph{database engines}. They may be tolerated by end users or higher-level applications,  depending on their business  logic~\cite{ACIDRain:SIGMOD2017,IsoDiff:VLDB2020}. }
This raises the question of whether such databases actually deliver the promised SI guarantee in practice.
Given that their internals (e.g., source code) are often unavailable to the outsiders 
 or hard to digest,
  a black-box SI checker is highly desirable.

A natural question then to ask is ``What  should an ideal black-box SI checker look like?''  The SIEGE principle \cite{Elle:VLDB2020} has already provided a strong baseline: an ideal checker would be \emph{sound}
(return no false positives),
\emph{informative} (report understandable counterexamples),
\emph{effective} (detect violations in real-world databases),
\emph{general} (compatible with different patterns of transactions),
and \emph{efficient} (add modest checking time even for workloads of high concurrency).
Additionally,  (i) we expect an ideal checker to be \emph{complete}, thus missing no violations;
and (ii) we also augment the \emph{generality} criterion by requiring the checker to be compatible
 not only with general (read-only, write-only, and read-write) transaction workloads 
    but also with standard key-value/SQL APIs.   We call this extended principle SIEGE+.

None of the existing SI checkers,
to the best of our knowledge,  satisfies SIEGE+ (see Section \ref{section:related-work} for the detailed comparison).  For example,
dbcop \cite{Complexity:OOPSLA2019} is incomplete,  incurs exponentially increasing overhead under higher concurrency (Section \ref{ss:efficiency}),  and returns no counterexamples upon finding a violation;
Elle \cite{Elle:VLDB2020} relies on specific database APIs such as lists and the (internal) timestamps of transactions to infer isolation anomalies, thus not conforming
to our black-box setting.



\vspace{1ex}
\noindent \textbf{The \name{} Checker.}
We present \name{},  a novel, black-box SI checker designed to achieve
all the  SIEGE+ criteria.
\name{} builds on three key ideas in response to three major challenges.

First,  despite previous attempts to characterize SI~\cite{CritiqueANSI:SIGMOD1995, Adya:PhDThesis1999,CentralisedSemantics:ECOOP2020},
its semantics is usually explained in terms of low-level implementation choices invisible to the database outsiders.
Consequently, one must \emph{guess} the dependencies (aka uncertain/unknown dependencies) between  client-observable data,
for example,   which of the two writes 
was first recorded in the database.

We introduce a novel dependency graph, called \emph{generalized polygraph} (GP),
based on which we present a new \emph{sound} and \emph{complete}  characterization of SI.
There are  two main advantages of a GP:  (i)  it naturally  models the guesses  by capturing \emph{all} possible dependencies between transactions in a single compacted data structure; and (ii)  it enables
 the acceleration of SMT solving by compacting constraints (see below) as demonstrated by our experiments.



Second,   there have been recent advances
 in SAT and SMT solving 
 for checking \emph{graph properties} such as
 the MonoSAT solver~\cite{MonoSAT:AAAI2015}  
and its  successful application to the black-box checking of SER~\cite{Cobra:OSDI2020}.  The idea is to \emph{search} for
an acyclic graph where the nodes are transactions in the history\footnote{A history collected from dynamically executing a system records the transactional requests to and  responses from the database.  See Section~\ref{ss:si-formal} for its formal definition.} and the edges meet certain constraints.
We show that
SMT techniques can also be applied to build an effective SI checker.
This application is nontrivial as
a brute-force approach
  would be inefficient due to
   the high computational complexity of checking SI \cite{Complexity:OOPSLA2019}: the problem is
\npc{} in general and $O(n^{c})$ with $c$ (resp.  $n$) a fixed,  yet in practice large,  number of clients (resp.  transactions),
even for a single transaction history.
In fact,  checking SI is known to be asymptomatically
more complex than checking SER~\cite{Complexity:OOPSLA2019}.
In the context of SMT solving over graphs,
SI leads to  much larger search space due to its specific anomaly patterns \cite{AnalysingSI:JACM2018} while checking SER  simply requires finding a cycle. 

Thanks to our GP-based characterization of SI,
we leverage its compact encoding of constraints on transaction dependencies to accelerate MonoSAT solving.
Moreover, we develop domain-specific optimizations that further prune constraints,
thereby reducing the search space.  For example,  \name{} prunes a constraint if an associated uncertain dependency would result in an SI violation with known dependencies.

Finally,  although MonoSAT outputs cycles upon detecting  a violation,  they are still   \emph{uninformative} with respect to  understanding how the violation actually occurred.
Locating the actual causes of violations would  facilitate  debugging and repairing  the defective implementations.  For example,  if an SI checker were to identify a \emph{lost update} anomaly 
from  the returned counterexample,  developers could then focus on investigating the write-write conflict resolution mechanism.
Hence,  we design and integrate into \name{} a novel interpretation algorithm that explains the  counterexamples returned by MonoSAT.  More specifically,   \name{}
(i) recovers the violating scenario
 by bringing back any potentially involved transactions and  dependencies eliminated during pruning and solving and (ii)
 finalizes the core participants to  highlight the violation cause.

\vspace{1ex}
\noindent \textbf{Main Contributions.}
In summary, we provide:
 \begin{enumerate}
 \item a new GP-based characterization of SI that both facilitates the modeling of uncertain transaction dependencies  inherent to black-box testing  and also enables  the acceleration of constraint solving (Section \ref{section:polygraph-si});
 \item  a sound and complete  GP-based checking algorithm for SI with domain-specific optimizations for pruning constraints (Section \ref{section:si-sat});
 \item  the  \name{} tool comprising both our new checking algorithm and the interpretation algorithm for debugging; and
 \item  an extensive assessment of \name{} that demonstrates its fulfilment of  SIEGE+  (Section~\ref{section:experiments}).  In particular,  \name{} successfully reproduces all of 2477 known SI anomalies,  detects novel SI violations in three production cloud  databases,  identifies their causes,
   outperforms the state-of-the-art black-box checkers under a wide range of workloads, and can scale up to large-sized workloads.
 \end{enumerate}



%% file: sections/problem-definition.tex

\section{Preliminaries}  \label{section:problem-definition}

\input{sections/si-informal}
\input{sections/si-formal}
\input{sections/si-checking-problem}
\input{sections/polygraph}

%% file: sections/si-informal.tex

\subsection{Snapshot Isolation in a Nutshell} \label{ss:si-informal}

\emph{Snapshot isolation} (SI)~\cite{CritiqueANSI:SIGMOD1995} is one of the most prominent weaker isolation levels that modern (cloud) databases usually provide to avoid  the performance penalty imposed by \emph{serializability} (SER).
Figure~\ref{fig:hierarchy} shows a hierarchy of  isolation levels where SI sits inbetween \emph{transactional causal consistency}~\cite{Eiger:NSDI2013} and SER,
and is not comparable to \emph{repeatable read} \cite{Adya:PhDThesis1999}.

\input{figs/tcm-less}

A transaction with SI always reads
from a snapshot that reflects a single commit ordering of transactions
and is allowed to commit if no concurrent transaction has updated the data that it intends to write.
SI prevents various undesired data anomalies such as fractured reads,  causality violations,
lost updates, and long fork~\cite{CritiqueANSI:SIGMOD1995,AnalysingSI:JACM2018}.
The following examples illustrate two kinds of anomalies disallowed by SI.  
 As we will see in Section~\ref{section:experiments},  both anomalies have been detected by our \name{} checker in production cloud databases.




\begin{example}[Causality Violation] \label{ex:social}
 Alice posts a photo of her birthday party.
Bob writes a comment to her post.
Later,  Carol sees Bob's comment but not Alice's post.
\end{example}

\begin{example}[Lost Update] \label{ex:banking}
Dan and Emma share a banking account with 10 dollars.
Both  simultaneously  deposit  50 dollars.  The resulting balance is 60, instead of 110,  as one of the deposits is lost.
\end{example}



In this paper we focus on the prevalent \emph{strong session} variant of \si{}~\cite{LazyReplSI:VLDB2006, AnalysingSI:JACM2018}, which
 additionally requires a transaction to observe all the effects of the preceding transactions in the same \emph{session} \cite{TerryDPSTW94}.
Many production databases, including DGraph \cite{dgraph}, Galera \cite{maria-galera},  and CockroachDB \cite{cockroach},  provide this isolation level in practice.


%% file: figs/tcm-less.tex

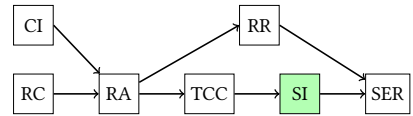
\begin{figure}[h!]
  \centering
  \resizebox{0.30\textwidth}{!}{
  \begin{tikzpicture}[model/.style = {draw, minimum size = 20pt, font = \large},
    node distance = 0.5cm and 0.8cm]

    \node[model] (rc) {\textsc{RC}};
      \node[model, above = of rc] (ci) {\textsc{CI}};
    \node[model, right = of rc] (ra) {\textsc{RA}};
    \node[model, right = of ra] (tcc) {\textsc{TCC}};
    \node[model, right = of tcc, fill = white!70!green] (si) {\textsc{SI}};
    \node[model, above left = 0.50cm and 0.00cm of si] (rr) {\textsc{RR}};
    \node[model, right = of si] (ser) {\textsc{SER}};

    \path[every edge, ->, thick]
      (rc) edge (ra)
      (ra) edge (tcc)
      (ra) edge (rr.west)
      (tcc) edge (si)
      (si) edge (ser);
    \draw[->, thick] (ci.east) to (ra);
    \draw[->, thick] (rr.east) to (ser);
  \end{tikzpicture}}
  \caption{A hierarchy of  isolation levels.
    $A \rightarrow B$: $A$ is strictly weaker than $B$.
    CI: cut isolation~\cite{HAT:VLDB2013};
    RC: read committed~\cite{CritiqueANSI:SIGMOD1995};
    RA: read atomicity~\cite{RAMP:TODS2016};
    RR: repeatable read~\cite{Adya:PhDThesis1999};
    TCC: transactional causal consistency~\cite{Eiger:NSDI2013};
    SI: snapshot isolation~\cite{AnalysingSI:JACM2018}; 
    SER: serializability~\cite{CritiqueANSI:SIGMOD1995}.}
  \label{fig:hierarchy}
\end{figure}

%% file: sections/si-formal.tex

\subsection{Snapshot Isolation: Formal Definition} \label{ss:si-formal}

We recall the formalization of SI over dependency graphs,
which serves as the theoretical foundation of \name.
The following account is standard, see for example~\cite{AnalysingSI:JACM2018},
and Table~\ref{table:notations} summarizes the notation used throughout the paper.

We consider a distributed key-value store
managing a set of keys $\Key = \set{\keyxvar, \keyyvar, \keyzvar, \dots}$,
which are associated with values from a set $\Val$.\footnote{
  We discuss how to support SQL queries in Section~\ref{section:discussion}.
  However, we do not support predicates in this work.}
We denote by $\Op$ the set of possible read or write operations on keys:
$\Op = \set{\readevent_{\opid}(\keyxvar, \valvar), \writeevent_{\opid}(\keyxvar, \valvar)
  \mid \opid \in \OpId, \keyxvar \in \Key, \valvar \in \Val}$,
where $\OpId$ is the set of operation identifiers.
We omit operation identifiers when they are unimportant.

\input{tables/table-notations}
\subsubsection{Relations, Orderings, Graphs, and Logics} \label{ss:relations}

A binary relation $R$ over a given set $A$ is a subset of $A \times A$, i.e., $R \subseteq A \times A$.
For $a, b \in A$, we use $(a, b) \in R$ and $a \rel{R} b$ interchangeably.
We use $R{?}$ and $R^{+}$ to denote the reflexive closure and the transitive closure of $R$, respectively.
A relation $R \subseteq A \times A$ is \emph{acyclic} if $R^{+} \cap I_{A} = \emptyset$,
where $I_{A} \triangleq \set{(a, a) \mid a \in A}$ is the identity relation on $A$.
Given two binary relations $R$ and $S$ over set $A$, we define their composition as
$R \comp S = \{ (a, c) \mid \exists b \in A: a \rel{R} b \rel{S} c\}$.
A strict partial order is an irreflexive and transitive relation.
A strict total order is a relation that is a strict partial order and total.

For a directed labeled graph $G = (V, E)$,
we use $V_{G}$ and $E_{G}$ to denote the set of vertices and edges in $G$, respectively.
For a set $F$ of edges, $G|_{F}$ denotes
the directed labeled graph that has the set $V$ of vertices and the set $F$ of edges.

In logical formulas, we write $\_$ for irrelevant parts that are  implicitly existentially quantified.
We use $\exists!$ to mean ``unique existence.''
\subsubsection{Transactions and Histories} \label{sss:transactions-histories}

\begin{definition} \label{def:transaction}
  A \bfit{transaction} is a pair $(\opset, \po)$,
  where $\opset \subseteq \Op$ is a finite, non-empty set of operations
  and $\po \subseteq \opset \times \opset$ is a strict total order called the \bfit{program order}.
\end{definition}

For a transaction $T$, we let $T \vdash \writeevent(\keyxvar, \valvar)$
if $T$ writes to $\keyxvar$ and the last value written is $\valvar$,
and $T \vdash \readevent(\keyxvar, \valvar)$
if $T$ reads from $\keyxvar$ before writing to it
and $\valvar$ is the value returned by the first such read.
We also use $\WriteTx_{\keyxvar} = \set{T \mid T \vdash \writeevent(\keyxvar, \_)}$.

Clients interact with the store by issuing transactions during \emph{sessions}.
We use a \emph{history} to record the client-visible results of such interactions.
For conciseness, we consider only committed transactions in the formalism~\cite{AnalysingSI:JACM2018};
see further discussions in Section~\ref{ss:completing-si-checking}.

\begin{definition} \label{def:history}
  A \bfit{history} is a pair $\h = (\T, \SO)$,
  where $\T$ is a set of transactions with disjoint sets of operations
  and the \bfit{session order} $\SO \subseteq \T \times \T$
  is the union of strict total orders on disjoint sets of $\T$,
  which correspond to transactions in different sessions.
\end{definition}


\input{sections/depgraph}

%% file: tables/table-notations.tex

\begin{table}[t]
\small
	\centering
	\caption{Notation}
	\label{table:notations}
	\resizebox{\columnwidth}{!}{%
	\begin{tabular}{|c|c|c|}
	\hline
	\textbf{Category}                   & \textbf{Notation}          & \textbf{Meaning}                  \\ \hline
	\multirow{3}{*}{\textit{KV Store}}  & $\Key$                     & set of keys                   \\ \cline{2-3}
										& $\Val$                     & set of values                 \\ \cline{2-3}
										& $\Op$                      & set of operations             \\ \hline
	\multirow{3}{*}{\textit{Relations}} & $R?$                       & reflexive closure of $R$          \\ \cline{2-3}
										& $R^{+}$                    & transitive closure of $R$         \\ \cline{2-3}
										& $R \comp S$                & composition of $R$ with $S$        \\ \hline
	\textit{Dependency}                 & $\SO$, $\WR$, $\WW$, $\RW$ & dependency relations/edges        \\ \hline
	\multirow{3}{*}{\textit{Graph}}     & $G = (V, E, C)$            & (generalized) polygraph           \\ \cline{2-3}
										& $V_{G}$, $E_{G}$, $C_{G}$  & components of $G$ \\ \cline{2-3}
										& $G|_{F}$                   & digraph with set $F$ of edges \\ \hline
	\multirow{4}{*}{\textit{Algorithm}} & $\H = (\T, \SO)$           & history to check                  \\ \cline{2-3}
										& $\inducedgraph$            & \si{} induced graph                     \\ \cline{2-3}
										& $\BV$                      & set of Boolean variables      \\ \cline{2-3}
										& $\Clause$                  & set of clauses                \\ \hline
	\end{tabular}%
	}
\normalsize
	\end{table}

%% file: sections/depgraph.tex

\subsubsection{Dependency Graph-based Characterization of \si} \label{sss:depgraph-si}

A dependency graph extends a history with three relations
(or typed edges, in terms of graphs): $\WR$, $\WW$, and $\RW$,
representing three possibility of dependencies between transactions in this history~\cite{AnalysingSI:JACM2018}.
The $\WR$ relation associates a transaction that reads some value with the one that writes this value.
The $\WW$ relation stipulates a strict total order (aka the version order~\cite{Adya:PhDThesis1999}) among the transactions on the same key.
The $\RW$ relation is derived from $\WR$ and $\WW$,
relating a transaction that reads some value to the one that overwrites this value,
in terms of the version orders specified by the $\WW$ relation.

\begin{definition} \label{def:depgraph}
  A \bfit{dependency graph} is a tuple $\G = (\T, \SO, \WR, \\ \WW, \RW)$,
  where $(\T, \SO)$ is a history and
  \begin{enumerate}[(1)]
	\item $\WR: \Key \to 2^{\T \times \T}$ is such that
    \begin{itemize}
      \item $\forall \keyxvar \in \Key.\; \forall S \in \T.\;
        S \vdash \readevent(\keyxvar, \_) \!\!\implies\!\! \exists! T \in \T.\; T \rel{\WR(\keyxvar)} S$
      \item $\forall \keyxvar \in \Key.\; \forall T, S \in \T.\;
        T \rel{\WR(x)} S \implies T \neq S \land
          \exists \valvar \in \Val.\; T \vdash \writeevent(\keyxvar, \valvar) \land S \vdash \readevent(\keyxvar, \valvar)$.
    \end{itemize}
	\item $\WW: \Key \to 2^{\T \times \T}$ is such that
    for every $\keyxvar \in \Key$, $\WW(\keyxvar)$ is a strict total order on the set $\WriteTx_{\keyxvar}$;
	\item $\RW: \Key \to 2^{\T \times \T}$ is such that
    $\forall T, S \in \T.\; \forall \keyxvar \in \Key.\;
      T \rel{\RW(\keyxvar)} S \iff T \neq S \land \exists T' \in \T.\; T' \rel{\WR(\keyxvar)} T \land T' \rel{\WW(\keyxvar)} S$.
  \end{enumerate}
\end{definition}

We denote a component of $\G$, such as $\WW$, by $\WW_{\G}$.
We write $T \rel{\WR/\WW/\RW} S$ when the key $\keyxvar$ in $T \rel{\WR(\keyxvar)/\WW(\keyxvar)/\RW(\keyxvar)} S$ is irrelevant or the context is clear.

Intuitively, a history satisfies \si{} if and only if
it can be extended to a dependency graph that contains only cycles (if any) with at least two adjacent $\RW$ edges.
Formally,

\begin{theorem}[Dependency Graph-based Characterization of \si{}
    (Theorem 4.1 of~\cite{AnalysingSI:JACM2018})] \label{thm:depgraph-si}
  For a history \emph{$\H = (\T, \SO)$},
  \emph{\begin{align*}
    \H \models \si & \iff \H \models \intaxiom \;\land \\
      &\exists \WR, \WW, \RW.\; \G = (\H, \WR, \WW, \RW) \;\land \\
      &\quad (((\SO_{\G} \cup \WR_{\G} \cup \WW_{\G}) \comp \RW_{\G}?) \text{\it\; is acyclic}).
  \end{align*}}
\end{theorem}

The \emph{internal consistency axiom} \intaxiom{} ensures that,
within a transaction, a read from a key returns the same value
as the last write to or read from this key in the transaction.


%% file: sections/si-checking-problem.tex

\subsection{The SI Checking Problem} \label{ss:problem-definition}

\begin{definition} \label{def:si-checking}
  The \bfit{\si{} checking problem} is the decision problem of
  determining whether a given history $\H$ satisfies \si{}, i.e., is $\H \models \si$?
\end{definition}

We take the common ``\uniqueval'' assumption on histories~\cite{Adya:PhDThesis1999, Complexity:OOPSLA2019, ClientCentric:PODC2017, VCC:POPL2017, Cobra:OSDI2020}:
for each key, every write to the key assigns a unique value.
For  database testing, we can produce such histories
by ensuring the uniqueness of the values written on the client side (or workload generator) using, e.g., the client identifier and local counter.
Under this assumption, each read can be associated with the transaction
that issues the corresponding write~\cite{AnalysingSI:JACM2018}.

Theorem~\ref{thm:depgraph-si} provides a brute-force approach to the SI checking problem:
enumerate all possible $\WW$ relations and check whether any of them results in a dependency graph
that contains only cycles with at least two adjacent $\RW$ edges.
This approach is, however, prohibitively expensive.

%% file: sections/polygraph.tex

\subsection{Polygraphs} \label{ss:polygraphs}

A dependency graph extending a history represents \emph{one} possibility of dependencies between transactions in this history.
To capture \emph{all} possible dependencies between transactions in a single structure,
we rely on polygraphs~\cite{SER:JACM1979}.
Intuitively, a polygraph can be viewed as a family of dependency graphs.

\input{figs/polygraph}

\begin{definition} \label{def:polygraph}
  A \bfit{polygraph} $G = (V, E, C)$ associated with a history $\H = (\T, \SO)$
  is a directed labeled graph $(V, E)$ called the \emph{known graph},
  together with a set $C$ of \emph{constraints} such that
  \begin{itemize}
    \item $V$ corresponds to all transactions in the history $\H$;
    \item $E = \set{(T, S, \SO) \mid T \rel{\SO} S} \cup \set{(T, S, \WR) \mid T \rel{\WR} S}$,
      where $\SO$ and $\WR$, when used as the third component of an edge, are edge labels (i.e., types); and
    \item $C = \set{\tuple{(T_{k}, T_{i}, \WW), (T_{j}, T_{k}, \RW)}
      \mid (T_{i} \rel{\WR(\keyxvar)} T_{j})
      \land T_{k} \in \WriteTx_{\keyxvar} \land T_{k} \neq T_{i} \land T_{k} \neq T_{j}}$.
  \end{itemize}
\end{definition}

As shown in Figure~\ref{fig:polygraph},
for a pair of transactions $T_{i}$ and $T_{j}$ such that $T_{i} \rel{\WR(\keyxvar)} T_{j}$
and a transaction $T_{k}$ that writes $\keyxvar$,
the constraint $\tuple{(T_{k}, T_{i}, \WW), (T_{j}, T_{k}, \RW)}$ captures the unknown dependencies
that ``either $T_{k}$ happened before $T_{i}$ or $T_{k}$ happened after $T_{j}$.''

%% file: figs/polygraph.tex

\begin{figure}[t]
  \centering
  \begin{subfigure}[b]{0.23\textwidth}
    \centering
    \includegraphics[width = 1.00\textwidth]{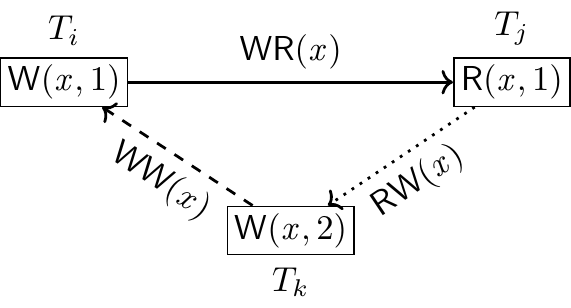}
    \caption{A polygraph}
    \label{fig:polygraph}
  \end{subfigure}
  \hfill
  \begin{subfigure}[b]{0.23\textwidth}
    \centering
    \includegraphics[width = 0.85\textwidth]{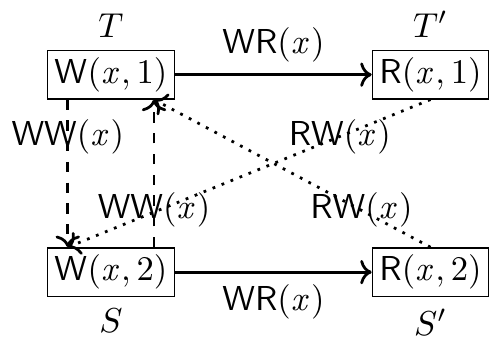}
    \caption{A generalized polygraph}
    \label{fig:g-polygraph}
  \end{subfigure}
  \caption{
  Examples of polygraphs and generalized polygraphs.
     $\WR$, $\WW$, and $\RW$ relations are represented by solid,
    dashed, and dotted arrows, respectively.}
  \label{fig:polygraph-g-polygraph}
\end{figure}

%% file: sections/g-polygraph.tex

\section{Characterizing SI using Generalized Polygraphs} \label{section:polygraph-si}

In this section we introduce generalized polygraphs with generalized constraints
and use them to characterize \si.
By compacting several constraints together,
using generalized constraints leads to a compact encoding which in turn helps accelerate the solving process
(see Section \ref{sss:differential}).

\subsection{Generalized Polygraphs} \label{ss:g-polygraph}

In a polygraph, a constraint involves only a single pair of transactions related by $\WR$,
like $T_{i}$ and $T_{j}$ on $\keyxvar$ in Figure~\ref{fig:polygraph}.
Thus, several constraints are needed when there are multiple transactions reading the value of $\keyxvar$ from $T_{i}$.
To \emph{compact} these constraints, we introduce \emph{generalized polygraphs} with generalized constraints.

\begin{definition} \label{def:polygraph-generalized}
  A \bfit{generalized polygraph} $G = (V, E, C)$ associated with a history $\H = (\T, \SO)$
  is a directed labeled graph $(V, E)$ called the \emph{known graph},
  together with a set $C$ of \emph{generalized constraints} such that
  \begin{itemize}
    \item $V$ corresponds to all transactions in the history $\H$;
    \item $E \subseteq V \times V \times \lbl$ is a set of edges with labels (i.e., types)
      from $\lbl = \set{\SO, \WR, \WW, \RW}$; and
    \item $C = \bset{\btuple{\eithervar \triangleq \set{(T, S, \WW)} \;\cup \bigcup\limits_{T' \in \WR(\keyxvar)(T)}\!\!\set{(T', S, \RW)}, \\
      \orvar \triangleq \set{(S, T, \WW)} \;\cup \bigcup\limits_{S' \in \WR(\keyxvar)(S)}\!\!\set{(S', T, \RW)}}
      \mid T \in \WriteTx_{\keyxvar} \land S \in \WriteTx_{\keyxvar} \land T \neq S}$.
  \end{itemize}
\end{definition}

A generalized constraint is a pair of sets of edges of the form $\tuple{\eithervar, \orvar}$.
The $\eithervar$ part handles the possibility of $T$ being ordered before $S$ via a $\WW$ edge.
This forces each transaction $T'$ that reads the value of $\keyxvar$ from $T$ to be ordered before $S$ via an $\RW$ edge.
Symmetrically, the $\orvar$ part handles the possibility of $S$ being ordered before $T$ via a $\WW$ edge.
This forces each transaction $S'$ that reads the value of $\keyxvar$ from $S$ to be ordered before $T$ via an $\RW$ edge.


\begin{example}[Generalized Polygraphs vs. Polygraphs] \label{ex:polygraph}
  In Figure~\ref{fig:g-polygraph}, both transactions $T$ and $S$ write to $\keyxvar$,
  and $T'$ and $S'$ read the values of $\keyxvar$ from $T$ and $S$, respectively.
  The possible dependencies between these transactions can be compactly expressed as
  a single generalized constraint $\tuple{\set{(T, S, \WW), (T', S, \RW)}, \set{(S, T, \WW), (S', T, \RW)}}$,
  which corresponds to two constraints:
  $\tuple{(T, S, \WW), (S', T, \RW)}$ and $\tuple{(S, T, \WW), (T', S, \RW)}$.
\end{example}

Note that a generalized polygraph may contain edges of any type in $E$ such that
a ``pruned'' generalized polygraph (Section~\ref{ss:prune-constraints}) is still a generalized polygraph.
For a generalized polygraph $G = (V, E, C)$ and a label $L \in \lbl$,
we use $V_{G}$, $E_{G}$, $C_{G}$, and $L_{G}$ to denote the set $V$ of vertices,
the set $E$ of known edges, the set $C$ of constraints, and the set of known edges with label $L$ in $E$, respectively.
For a generalized polygraph $G = (V, E, C)$ and a set $F$ of edges,
we use $G|_{F}$ to denote the directed labeled graph that has the set $V$ of vertices and the set $F$ of edges.
For any two subsets $R, S \subseteq E$, we define their composition as
$R \comp S = \{ (a, c, L_{1} \comp L_{2}) \mid
  \exists b \in V: (a, b, L_{1}) \in R \land (b, c, L_{2}) \in S \}$,
where $L_{1} \comp L_{2}$ is a newly introduced label used when composing edges.
In the sequel, we use generalized polygraphs, but sometimes we still refer to them as polygraphs.
\subsection{Characterizing \si} \label{ss:g-poly-si}

According to Theorem~\ref{thm:depgraph-si}, we are interested in the \emph{induced} graph
of a generalized polygraph $\g$, obtained by composing the edges of $\g$ according to the rule
$((\SO \cup \WR \cup \WW) \comp \RW?)$.

\begin{definition} \label{def:induced-graph}
  The \bfit{induced \si{} graph} of a polygraph $\g = (V, E, C)$ is the graph
  $G' = (V, E, C, \inducedrule)$,
  where $\inducedrule = (\SO \cup \WR \cup \WW) \comp \RW?$ is the induce rule.
\end{definition}

The concept of \emph{compatible} graphs gives a meaning to polygraphs
and their induced \si{} graphs.
A graph is compatible with a polygraph when it is a resolution of the constraints of the polygraph.
Thus, a polygraph corresponds to a family of its compatible graphs.

\begin{definition} \label{def:compatible-graphs-with-a-polygraph}
  A directed labeled graph $G' = (V', E')$ is \bfit{compatible with
  a generalized polygraph} $G = (V, E, C)$ if
  \begin{itemize}
    \item $V' = V$;
    \item $E' \supseteq E$; and
    \item $\forall \tuple{\eithervar, \orvar} \in C.\;
      (\eithervar \subseteq E' \land \orvar \cap E' = \emptyset)
      \lor (\orvar \subseteq E' \land \eithervar \cap E' = \emptyset)$.
  \end{itemize}
\end{definition}

By applying the induce rule $\inducedrule$ to a compatible graph of a polygraph,
we obtain a compatible graph with the induced \si{} graph of this polygraph.

\begin{definition} \label{def:compatible-graphs-with-an-induced-si-graph}
  Let $G' = (V', E')$ be a compatible graph with a polygraph $G$.
  Then $G'|_{(\SO_{G'} \;\cup\; \WR_{G'} \;\cup\; \WW_{G'}) \comp \RW_{G'}?}$ is a \bfit{compatible graph with the induced \si{} graph} of $G$.
\end{definition}

\begin{example}[Compatible Graphs] \label{ex:compatible-graphs-with-a-polygraph}
  There are two compatible graphs with the generalized polygraph of Figure~\ref{fig:g-polygraph}:
  one is with the edge set $\set{(T, T', \WR), (S, S', \WR), (T, S, \WW), (T', S, \RW)}$,
  and the other is with $\set{(T, T', \WR), (S, S', \WR), (S, T, \WW), (S', T, \RW)}$.

  Accordingly, there are also two compatible graphs with the induced \si{} graph
  of the polygraph of Figure~\ref{fig:g-polygraph}:
  one is with the edge set $\set{(T, T', \WR), (S, S', \WR), (T, S, \WW), (T, S, \WR \comp \RW)}$.
  The edge $(T, S, \WR \comp \RW)$ is obtained from $(T, T', \WR) \comp (T', S, \RW)$.
  It is identical to $(T, S, \WW)$ if the edge types are ignored.
  The other is with $\set{(T, T', \WR), (S, S', \WR), (S, T, \WW), (S, T, \WR \comp \RW)}$.
  Similarly, $(S, T, \WR \comp \RW)$ is identical to $(S, T, \WW)$
  if the edge types are ignored.
\end{example}

We are concerned with the acyclicity of polygraphs and their induced \si{} graphs.

\begin{definition} \label{def:si-acyclicity}
  An induced \si{} graph is \bfit{acyclic} if there exists an acyclic compatible graph with it,
  \emph{when the edge types are ignored}.
  A polygraph is \bfit{\si-acyclic} if its induced \si{} graph is acyclic.
\end{definition}

Finally, we present the generalized polygraph-based characterization of \si.
Its proof can be found in Appendix~\ref{section:appendix-proofs}.
The key lies in the correspondence between compatible graphs of polygraphs and dependency graphs.
\begin{theorem}[Generalized Polygraph-based Characterization of SI] \label{thm:polygraph-si}
  A history $\H$ satisfies \emph{\si{}} if and only if
  \emph{$\H \models \intaxiom$} and the generalized polygraph of $\H$ is \si-acyclic.
\end{theorem}


%% file: sections/si-sat.tex

\section{The Checking Algorithm for \si}  \label{section:si-sat}

Given a history $\H$, \name{} encodes the induced \si{} graph
of the generalized polygraph of $\H$ into an SAT formula
and utilizes the MonoSAT solver~\cite{MonoSAT:AAAI2015} to test its acyclicity.
We choose MonoSAT mainly because, compared to conventional SMT solvers such as Z3,  it is more efficient in checking graph properties~\cite{MonoSAT:AAAI2015}.

The main challenge 
 is that the size (measured as the number of variables and clauses)
of the resulting SAT formula may be too large for MonoSAT to solve in reasonable time.
Hence, \name{} prunes constraints of the generalized polygraph of $\H$ before encoding.
As we will see in Section~\ref{ss:efficiency},  this pruning process is crucial to \name's high performance.
Additionally,  solving is accelerated by utilizing,
instead of original polygraphs, generalized polygraphs with compact generalized constraints.

\input{figs/long-fork-example}
\subsection{Overview} \label{ss:overview}

The procedure $\checksi$ (line~\linecode{\ref{alg:checksi}}{\ref{line:func-checksi}} of Algorithm~\ref{alg:checksi})
outlines the checking algorithm.
First, if $\H$ does not satisfy the $\intaxiom$ axiom,
the checking algorithm terminates and returns \false{}
(line~\linecode{\ref{alg:checksi}}{\ref{line:checksi-intaxiom}};
see Section~\ref{ss:completing-si-checking} for 
the predicates $\abortedreads$ and $\intermediatereads$).
The algorithm proceeds otherwise in the following three steps:

\begin{itemize}
  \item
    construct the generalized polygraph $\g$ of $\H$
    (lines~\linecode{\ref{alg:checksi}}{\ref{line:checksi-call-createknowngraph}}
    and~\linecode{\ref{alg:checksi}}{\ref{line:checksi-call-generateconstraints}});
  \item
    prune constraints in the polygraph $\g$
    (line~\linecode{\ref{alg:checksi}}{\ref{line:checksi-call-pruneconstraints}}); and
  \item 
    encode the induced \si{} graph, denoted $\inducedgraph$, of $\g$ after pruning
    into an SAT formula (line~\linecode{\ref{alg:checksi}}{\ref{line:checksi-call-encodeconstraints}}),
    and call MonoSAT to test whether $\inducedgraph$ is acyclic (line~\linecode{\ref{alg:checksi}}{\ref{line:checksi-call-solveconstraints}}).
\end{itemize}
Algorithm~\ref{alg:checksi} depicts the core procedures of pruning and encoding.
The remaining procedures are given in Appendix~\ref{section:appendix-polysi}.

Before diving into details, we illustrate our algorithm using the example history in Figure~\ref{fig:long-fork}, which exemplifies the well-known ``long fork'' anomaly in \si~\cite{PSI:SOSP2011,AnalysingSI:JACM2018}.
Specifically, transaction $T_{0}$ writes to both keys  $\keyxvar$ and $\keyyvar$.
Transactions $T_{1}$ and $T_{2}$ concurrently write to $\keyxvar$ and $\keyyvar$, respectively.
Transaction $T_{3}$ sees the write by $T_{1}$, but not the write by $T_{2}$,
while $T_{4}$ sees the write by $T_{2}$, but not the write by $T_{1}$.
The session committing $T_{0}$ then issues $T_{5}$ to update $\keyxvar$.

\vspace{1ex}
\noindent \textbf{Pruning Constraints.}
To check whether this history satisfies \si,
we must determine the order between $T_{0}$, $T_{1}$, and $T_{5}$ (on $\keyxvar$)
and the order between $T_{0}$ and $T_{2}$ (on $\keyyvar$).
Consider first the constraint on the order between $T_{0}$ and $T_{5}$ shown in
 Figure~\ref{fig:long-fork-2more-soww}.
Due to $T_{0} \rel{\SO} T_{5}$, the $T_{5} \rel{\WW(\keyxvar)} T_{0}$ case
would introduce an undesired cycle.
Therefore, this case can be safely pruned
and the other case of $T_{0} \rel{\WW(\keyxvar)} T_{5}$,
along with the edge $T_{4} \rel{\RW(\keyxvar)} T_{5}$, become known.

Figure~\ref{fig:long-fork-keyx} shows the constraint on the order
between $T_{0}$ and $T_{1}$:
$\tuple{\eithervar = \set{(T_{1}, T_{0}, \WW), (T_{3}, T_{0}, \RW},
  \orvar = \set{(T_{0}, T_{1}, \WW), (T_{4}, T_{1}, \RW)}}$.
Consider first the $\eithervar$ case.
Note that the edge $T_{3} \rel{\RW(\keyxvar)} T_{0}$ is in an undesired cycle
$T_{3} \rel{\RW(\keyxvar)} T_{0} \rel{\WR(\keyyvar)} T_{3}$,
which contains only a single $\RW$ edges.
Hence, the $\eithervar$ case could be safely pruned,
without SAT encoding and solving.
Conversely, the $\orvar$ case does not introduce undesired cycles.
Thus, the edges in the $\orvar$ case become known,
before SAT encoding and solving.

Similarly, 
the $\eithervar$ case of the constraint
on the order between $T_{0}$ and $T_{2}$, namely
$\tuple{\eithervar = \set{(T_{2}, T_{0}, \WW), (T_{4}, T_{0}, \RW},
  \orvar = \\ \set{(T_{0}, T_{2}, \WW), (T_{4}, T_{2}, \RW)}}$,
could be safely pruned (not shown in Figure~\ref{fig:long-fork-keyy}),
and the edges in the $\orvar$ case become known.

\vspace{1ex}
\noindent \textbf{SAT Encoding.}
The order between $T_{1}$ and $T_{5}$ is still uncertain after pruning
in Figure~\ref{fig:long-fork-keyy}.
We encode the constraint $\tuple{
  \eithervar = \set{(T_{1}, T_{5}, \WW), (T_{3}, T_{5}, \RW)},
  \orvar = \set{(T_{5}, T_{1}, \WW)}}$
on the order as a SAT formula
\[
  (\BV_{1,5} \land \BV_{3,5} \land \lnot\BV_{5,1}) \lor
  (\BV_{5,1} \land \lnot\BV_{1,5} \land \lnot\BV_{3,5}),
\]
where $\BV_{i,j}$ is a Boolean variable indicating the existence of
the edge from $T_{i}$ to $T_{j}$ in the pruned polygraph.
We then encode the induced \si{} graph, denoted $\inducedgraph$.
Since $T_{2} \rel{\WR(\keyyvar)} T_{4} \rel{\RW(\keyxvar)} T_{5}$,
we have $\BV^{\inducedgraph}_{2,5} = \BV_{2,4} \land \BV_{4,5}$,
where $\BV^{\inducedgraph}_{i,j}$ is a Boolean variable indicating the existence
of the edge from $T_{i}$ to $T_{j}$ in $\inducedgraph$.
Similarly, we have $\BV^{\inducedgraph}_{1,2} = \BV_{1,3} \land \BV_{3,2}$
and $\BV^{\inducedgraph}_{2,1} = \BV_{2,4} \land \BV_{4,1}$.
In contrast, since it is possible that $T_{3} \rel{\RW(\keyxvar)} T_{5}$,
we have $\BV^{\inducedgraph}_{1,5} = \BV_{1,3} \land \BV_{3,5}$.

\vspace{1ex}
\noindent \textbf{MonoSAT Solving.}
Finally, we feed the SAT formula to MonoSAT for an acyclicity
test of the graph $\inducedgraph$.
MonoSAT successfully finds an undesired cycle
$T_{1} \rel{\WR(\keyxvar)} T_{3} \rel{\RW(\keyyvar)} T_{2}
    \rel{\WR(\keyyvar)} T_{4} \rel{\RW(\keyxvar)} T_{1}$,
which contains two \emph{non-adjacent} $\RW$ edges;
see Figure~\ref{fig:long-fork-cycle}.
Therefore, this history violates \si.
\input{algs/checksi}

\input{sections/constructing}
\input{sections/pruning}
\input{sections/sat-encoding}
\input{sections/monosat-solving}
\input{sections/completing}

%% file: figs/long-fork-example.tex

\begin{figure*}[t]
  \centering
 \begin{subfigure}[c]{0.28\textwidth}
   \centering
   \includegraphics[width = 1.00\textwidth]{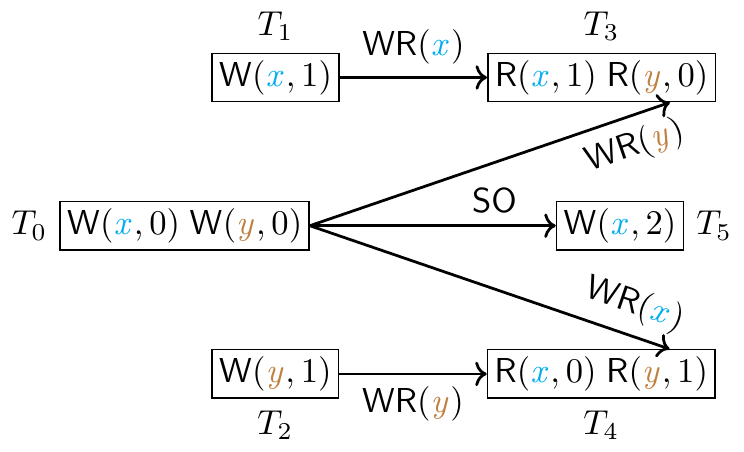}
   \caption{\small A ``long fork'' history with $\SO$ and $\WR$ edges.
    $T_{0}$ and $T_{5}$ are on the same session.}
   \label{fig:long-fork}
 \end{subfigure}
 \hfill
  \begin{subfigure}[c]{0.30\textwidth}
    \includegraphics[width = 1.00\textwidth]{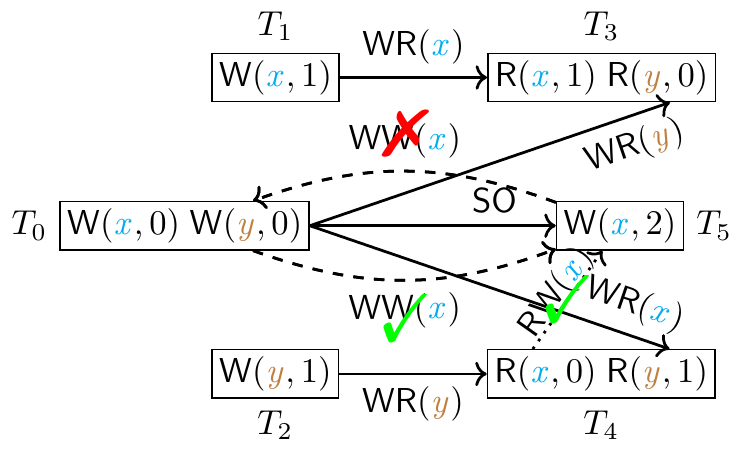}
    \caption{\small The $T_{5} \rel{\WW(\keyxvar)} T_{0}$ case is pruned
      due to the cycle $T_{0} \rel{\SO} T_{5} \rel{\WR(\keyxvar)} T_{0}$.}
    \label{fig:long-fork-2more-soww}
  \end{subfigure}
 \hfill
  \begin{subfigure}[c]{0.30\textwidth}
    \includegraphics[width = 1.00\textwidth]{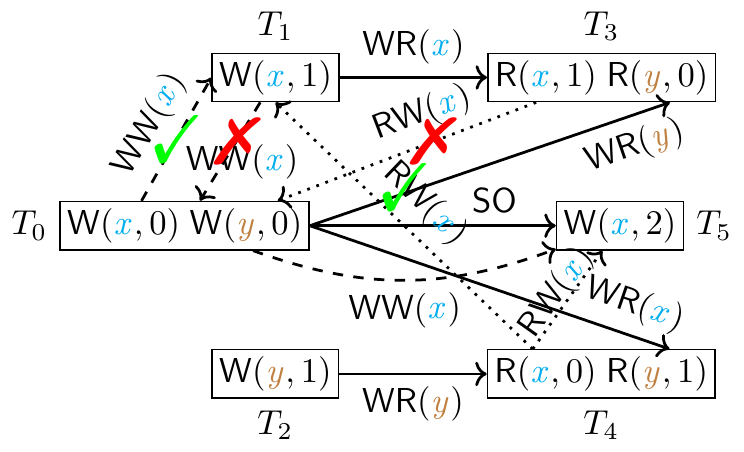}
    \caption{\small The $T_{1} \rel{\WW(\keyxvar)} T_{0}$ case is pruned
      due to the cycle $T_{3} \rel{\RW(\keyxvar)} T_{0} \rel{\WR(\keyyvar)} T_{3}$.}
    \label{fig:long-fork-keyx}
  \end{subfigure}

  \vspace{0.30cm}
  \begin{subfigure}[c]{0.30\textwidth}
    \includegraphics[width = 1.00\textwidth]{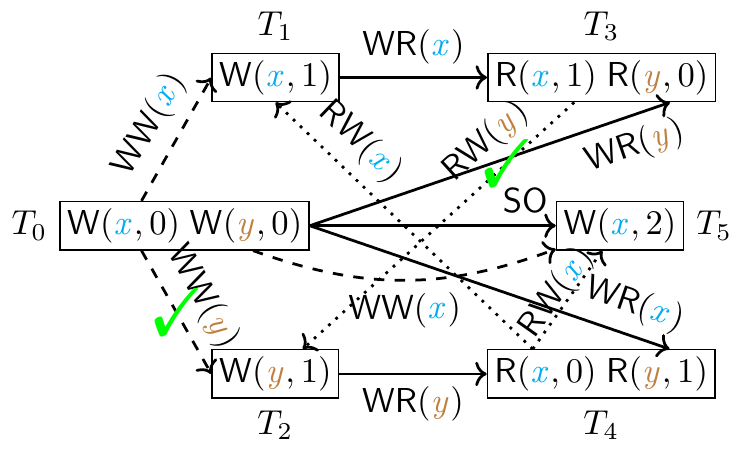}
    \caption{\small The $T_{2} \rel{\WW(\keyyvar)} T_{0}$ case
      is pruned and the $T_{0} \rel{\WW(\keyyvar)} T_{2}$ case
      become known.}
    \label{fig:long-fork-keyy}
  \end{subfigure}
  \hspace{100pt}
  \begin{subfigure}[c]{0.20\textwidth}
    \includegraphics[width = 1.00\textwidth]{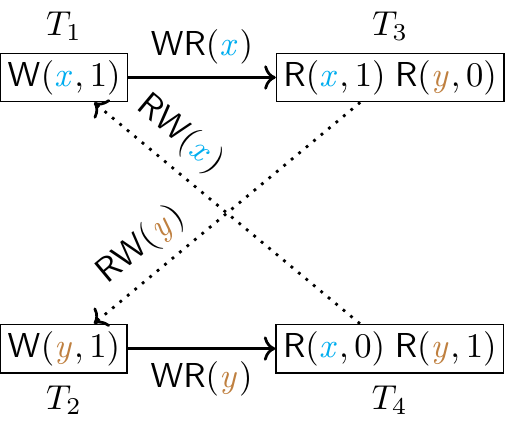}
    \caption{\small The undesired and  violating cycle found by MonoSAT.}
    \label{fig:long-fork-cycle}
  \end{subfigure}
  \caption{The ``long fork'' anomaly: an illustrating example of \name.
	 }
  \label{fig:long-fork-example}
\end{figure*}

%% file: algs/checksi.tex

\begin{algorithm}[t]
  \footnotesize
  \caption{The \name{} algorithm for checking SI}
  \label{alg:checksi}
  \begin{varwidth}[t]{0.49\textwidth}
  \begin{algorithmic}[1]

    \Procedure{\checksi}{$\H$}
      \label{line:func-checksi}
      \If{$\H \not\models \intaxiom \lor \abortedreads \lor \intermediatereads$}
        \label{line:checksi-intaxiom}
        \State \Return \false
      \EndIf

      \State \Call{\createknowngraph}{$\H$}
        \Comment{see Appendix~\ref{section:appendix-polysi}}
        \label{line:checksi-call-createknowngraph}
      \State \Call{\generateconstraints}{$\H$}
        \Comment{see Appendix~\ref{section:appendix-polysi}}
        \label{line:checksi-call-generateconstraints}
      \If{$\lnot \Call{\pruneconstraints}{\null}$}
        \label{line:checksi-call-pruneconstraints}
        \State \Return \false
          \label{line:checksi-return-false}
      \EndIf
      \State \Call{\encodeconstraints}{\null}
        \label{line:checksi-call-encodeconstraints}
      \State \Return \Call{\solveconstraints}{\null}
        \Comment{see Appendix~\ref{section:appendix-polysi}}
        \label{line:checksi-call-solveconstraints}
    \EndProcedure

    \Statex
    \Procedure{\pruneconstraints}{\null}
      \label{line:func-pruneconstraints}
      \Repeat
        \State $\graphA \gets \g|_{\SO_{\g} \cup \WR_{\g} \cup \WW_{\g}}$
          \label{line:pruneconstraints-A}
        \State $\graphB \gets \g|_{\RW_{\g}}$
          \label{line:pruneconstraints-B}
        \State $\knowninducedgraph \gets \graphA \cup (\graphA \comp \graphB)$
          \label{line:pruneconstraints-C}
        \State $\reachabilityvar \gets \Call{\reachability}{\knowninducedgraph}$
          \Comment{using Floyd-Warshall algorithm~\cite{CLRS2009}}. 
          \label{line:pruneconstraints-reachability}

        \hStatex
        \ForAll{$\consvar \gets \tuple{\eithervar, \orvar} \in \cons_{\g}$}
          \label{line:pruneconstraints-forall-cons}
          \ForAll{$(\fromvar, \tovar, \typevar) \in \eithervar$}
            \Comment{for the ``$\eithervar$'' possibility}
            \label{line:pruneconstraints-forall-edge-in-either}
            \If{$\typevar = \WW$}
              \label{line:pruneconstraints-ww-type}
              \If{$(\tovar, \fromvar) \in \reachabilityvar$}
                \label{line:pruneconstraints-ww-reachability}
                \State $\cons_{\g} \gets \cons_{\g} \setminus \set{\consvar}$
                  \label{line:pruneconstraints-ww-cons}
                \State $\edges_{\g} \gets \edges_{\g} \cup \orvar$
                  \label{line:pruneconstraints-ww-edges}
                \State {\bf break} the ``{\bf for all} $(\fromvar, \tovar, \typevar) \in \eithervar$'' loop
                  \label{line:pruneconstraints-ww-break}
              \EndIf
            \Else \Comment{$\typevar = \RW$}
              \ForAll{$\precvar \in \vertex_{\graphA}$ such that $(\precvar, \fromvar, \_) \in \edges_{\graphA}$}
                \label{line:pruneconstraints-rw-forall-pre-vertex}
                \If{$(\tovar, \precvar) \in \reachabilityvar$}
                  \label{line:pruneconstraints-rw-reachability}
                  \State $\cons_{\g} \gets \cons_{\g} \setminus \set{\consvar}$
                    \label{line:pruneconstraints-rw-cons}
                  \State $\edges_{\g} \gets \edges_{\g} \cup \orvar$
                    \label{line:pruneconstraints-rw-edges}
                  \State {\bf break} the ``{\bf for all} $(\fromvar, \tovar, \typevar) \in \eithervar$'' loop
                    \label{line:pruneconstraints-rw-break}
                \EndIf
              \EndFor
            \EndIf
          \EndFor

          \ForAll{$(\fromvar, \tovar, \typevar) \in \orvar$}
            \Comment{for the ``$\orvar$'' possibility}
            \label{line:pruneconstraints-forall-edge-in-or}
            \State the same with the ``$\eithervar$'' possibility except that
                it returns \False
            \Statex \hspace{40pt} if both $\eithervar$ and $\orvar$ possibilities of a constraint are pruned
          \EndFor
        \EndFor
      \Until{$\cons_{\g}$ remains unchanged}
        \label{line:pruneconstraints-until}
      \State \Return \True
        \label{line:pruneconstraints-return-true}
    \EndProcedure

    \Statex
    \Procedure{\encodeconstraints}{\null}
      \label{line:func-encodeconstraints}
      \ForAll{$\vvar_{i}, \vvar_{j} \in \vertex_{\g}$ such that $i \neq j$}
        \label{line:encodeconstraints-forall-vertex-pair}
        \State $\BV \gets \BV \cup \set{\BV_{i, j}, \BV^{\inducedgraph}_{i, j}}$
          \label{line:encodeconstraints-bv-add-bvij}
      \EndFor

      \ForAll{$(\vvar_{i}, \vvar_{j}) \in \edges_{\g}$} \Comment{encode the known graph of $\g$}
        \label{line:encodeconstraints-forall-edges}
        \State $\Clause \gets \Clause \cup \set{\BV_{i, j} = \True}$
          \label{line:encodeconstraints-set-bvij-true}
      \EndFor

      \ForAll{$\tuple{\eithervar, \orvar} \in \cons_{\g}$}  \Comment{encode the constraints of $\g$}
        \label{line:encodeconstraints-forall-cons}
        \State $\Clause \gets \Clause \;\cup\; \Bset{\bigl(
            \bigwedge\limits_{(\vvar_{i}, \vvar_{j}, \_) \in \eithervar}
              \!\!\!\!\BV_{i, j}
            \land \bigwedge\limits_{(\vvar_{i}, \vvar_{j}, \_) \in \orvar}
              \!\!\!\!\lnot\BV_{i, j}\bigr) \;\lor\;
            \bigl(
              \bigwedge\limits_{(\vvar_{i}, \vvar_{j}, \_) \in \orvar}
              \!\!\!\!\BV_{i, j}
              \land \bigwedge\limits_{(\vvar_{i}, \vvar_{j}, \_) \in \eithervar}
              \!\!\!\!\lnot\BV_{i, j}\bigr)}$
          \label{line:encodeconstraints-either-or}
      \EndFor

      \hStatex
      \State $\graphA \gets \g|_{\SO_{\g} \cup \WR_{\g} \cup \WW_{\g}}$
        \label{line:encodeconstraints-A}
      \State $\edges_{\graphA} \gets \edges_{\graphA} \cup \set{(\_, \_, \WW) \in \eithervar \cup \orvar \mid \tuple{\eithervar, \orvar} \in \cons_{\g}}$
        \label{line:encodeconstraints-add-ww-to-A}
      \State $\graphB \gets \g|_{\RW_{\g}}$
        \label{line:encodeconstraints-B}
      \State $\edges_{\graphB} \gets \edges_{\graphB} \cup \set{(\_, \_, \RW) \in \eithervar \cup \orvar \mid \tuple{\eithervar, \orvar} \in \cons_{\g}}$
        \label{line:encodeconstraints-add-rw-to-B}
      \State $\Clause \gets \Clause \cup \Bset{\BV^{\inducedgraph}_{i, j} =
        \big(\BV_{i, j} \land (\vvar_{i}, \vvar_{j}, \_) \in \edges_{\graphA}\big) \lor
        \big(\bigvee\limits_{\substack{(\vvar_{i}, \vvar_{k}, \_) \in \edges_{\graphA} \\
          (\vvar_{k}, \vvar_{j}, \_) \in \edges_{\graphB}}} \hspace{-1em} \BV_{i, k} \land \BV_{k, j}\big)
          \bigm| \vvar_{i}, \vvar_{j} \in \vertex_{\g}}$
        \Comment{encode the induced \si{} graph $\inducedgraph$ of $\g$}
        \label{line:encodeconstraints-add-clauses-for-I}
    \EndProcedure
  \end{algorithmic}
  \end{varwidth}
  \normalsize
\end{algorithm}

%% file: sections/constructing.tex

\subsection{Constructing the Generalized Polygraph} \label{ss:constructing-polygraph}

We construct the generalized polygraph $\g$ of the history $\H$ in two steps.
First, we create the known graph of $\g$ by adding the known edges of types $\SO$ and $\WR$ to $\edges_{\g}$.
Second, we generate the generalized constraints of $\g$ on possible dependencies between transactions.
Specifically, for each key $\keyxvar$ and each pair of transactions $T$ and $S$ that both write $\keyxvar$,
we generate a generalized constraint of the form $\tuple{\eithervar, \orvar}$
according to Definition~\ref{def:polygraph-generalized}.

%% file: sections/pruning.tex

\subsection{Pruning Constraints} \label{ss:prune-constraints}

\looseness=-1
To accelerate MonoSAT solving, we prune as many constraints as possible before encoding
(line~\linecode{\ref{alg:checksi}}{\ref{line:func-pruneconstraints}}).
A constraint can be pruned if either of its two possibilities,
represented by \emph{either} or \emph{or}, cannot happen,
i.e., adding the edges in one of the two possibilities would create a cycle in
the reduced \si{} graph.
If neither of the two possibilities in a constraint can happen,
\name{} immediately returns $\False$.
This process is repeated until no more constraints can be pruned
(line~\linecode{\ref{alg:checksi}}{\ref{line:pruneconstraints-until}}).


\input{figs/pruning-cases}

In each iteration, we first construct the \emph{currently known part} of the induced \si{} graph,
denoted $\knowninducedgraph$, of $\g$.
To do this, we define two auxiliary graphs, namely $\graphA \gets \g|_{\SO_{\g} \cup \WR_{\g} \cup \WW_{\g}}$
and $\graphB \gets \g|_{\RW_{\g}}$.
By Definition~\ref{def:compatible-graphs-with-an-induced-si-graph}, $\knowninducedgraph$ is $\graphA \cup (\graphA \comp \graphB)$
(line~\linecode{\ref{alg:checksi}}{\ref{line:pruneconstraints-C}}).
Then, we compute the reachability relation of $\knowninducedgraph$.
Next, for each constraint $\consvar$ of the form $\tuple{\eithervar, \orvar}$,
we check if $\eithervar$ or $\orvar$ would create cycles in $\knowninducedgraph$
(line~\linecode{\ref{alg:checksi}}{\ref{line:pruneconstraints-forall-cons}}).
Consider an edge $(\fromvar, \tovar, \typevar)$ in $\eithervar$ (line~\linecode{\ref{alg:checksi}}{\ref{line:pruneconstraints-forall-edge-in-either}}).
By construction, it must be of type $\WW$ or $\RW$.
Note that $\knowninducedgraph$ does not contain any $\RW$ edges by definition.
Therefore, an $\RW$ edge from $\fromvar$ to $\tovar$,
together with a path from $\tovar$ to $\fromvar$ in $\knowninducedgraph$,
does \emph{not} necessarily create a cycle in $\knowninducedgraph$.
This fails the simple reachability-based strategy used in Cobra~\cite{Cobra:OSDI2020}.

Suppose first that $(\fromvar, \tovar)$ is a $\WW$ edge;
see Figure~\ref{fig:pruning-ww-case}.
If there is already a path from $\tovar$ to $\fromvar$ in $\knowninducedgraph$
(line~\linecode{\ref{alg:checksi}}{\ref{line:pruneconstraints-ww-reachability}}),
adding the $\WW$ edge would create a cycle in $\knowninducedgraph$.
Thus, we can prune the constraint $\consvar$
and the edges in the other possibility $\orvar$ become known.

Now suppose that $(\fromvar, \tovar)$ is an $\RW$ edge;
see Figure~\ref{fig:pruning-rw-case}.
We check if there is a path in $\knowninducedgraph$ from $\tovar$
to any immediate predecessor $\precvar$ of $\fromvar$ in $\graphA$
(line~\linecode{\ref{alg:checksi}}{\ref{line:pruneconstraints-rw-forall-pre-vertex}}).
If there is a path, adding this $\RW$ edge would introduce,
via composition with the edge from $\precvar$ to $\fromvar$,
an edge from $\precvar$ to $\tovar$ in $\knowninducedgraph$
(the dashed arrow in Figure~\ref{fig:pruning-rw-case}).
Then, with the path from $\tovar$ to $\precvar$,
we obtain a cycle in $\knowninducedgraph$.

The pruning process of the $\orvar$ possibility is same with that for $\eithervar$,
except that
it returns \False{} if both $\eithervar$ and $\orvar$ possibilities of a constraint are pruned.

\input{sections/pruning-proof}

%% file: figs/pruning-cases.tex

\begin{figure}[t]
  \centering
  \begin{subfigure}[c]{0.2\textwidth}
    \centering
    \includegraphics[width = 0.80\textwidth]{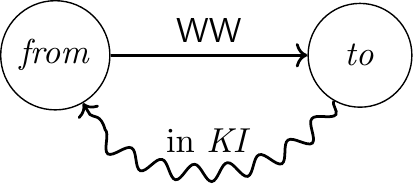}
    \caption{$(\fromvar, \tovar)$ is a $\WW$ edge.}
    \label{fig:pruning-ww-case}
  \end{subfigure}
  \hfill
  \begin{subfigure}[c]{0.2\textwidth}
    \centering
    \includegraphics[width = 0.80\textwidth]{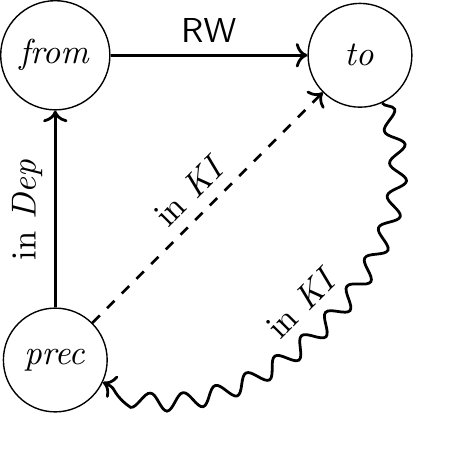}
    \caption{$(\fromvar, \tovar)$ is an $\RW$ edge.}
    \label{fig:pruning-rw-case}
  \end{subfigure}
  \caption{Two cases for pruning constraints.}
  \label{fig:pruning}
\end{figure}

%% file: sections/pruning-proof.tex
\looseness=-1
The following theorem states that $\pruneconstraints$ is correct
in that (1) it preserves the \si-(a)cyclicity of polygraphs;
and (2) it does not introduce new undesired cycles,
which ensures that any violation found in the pruned polygraph
using MonoSAT later also exists in the original polygraph.
This is crucial to the informativeness of \name.
The theorem's proof can be found in Appendix~\ref{section:appendix-proofs}.

\begin{theorem}[Correctness of \pruneconstraints] \label{thm:correctness-of-pruneconstraints}
  Let $G$ and $G_{p}$ be the generalized polygraphs
  before and after \emph{$\pruneconstraints$}, respectively. Then,
  \begin{enumerate}[(1)]
    \item $G$ is \si-acyclic if and only if \emph{$\pruneconstraints$} returns \emph{\True{}} and $G_{p}$ is \si-acyclic.
    \item Suppose that $G_{p}$ is not \si-acyclic.
      Let $\cycle$ be a cycle in a compatible graph with the induced \si{} graph of $G_{p}$.
      Then there is a compatible graph with the induced \si{} graph of $G$ that contains $\cycle$.
  \end{enumerate}
\end{theorem}

Combining Theorems~\ref{thm:polygraph-si} and~\ref{thm:correctness-of-pruneconstraints}, we prove \name{}'s soundness.
\begin{theorem}[Soundness of \name] \label{thm:soundness-of-polysi}
  \emph{\name{}} is sound, i.e., if \emph{\name{}} returns \emph{$\False$},
  then the input history indeed violates \si.
\end{theorem}

%% file: sections/sat-encoding.tex

\subsection{SAT Encoding} \label{ss:sat-encoding}

In this step we encode the induced \si{} graph, denoted $\inducedgraph$,
of the pruned polygraph $\g$ into an SAT formula
(line~\linecode{\ref{alg:checksi}}{\ref{line:func-encodeconstraints}}).
We use $\BV$ and $\Clause$ to denote the set of Boolean variables
and the set of clauses of the SAT formula, respectively.
For each pair of vertices $\vvar_{i}$ and $\vvar_{j}$,
we create two Boolean variables $\BV_{i, j}$ and $\BV^{\inducedgraph}_{i, j}$:
one for the polygraph $G$, and the other for its induced \si{} graph $\inducedgraph$.
An edge $(\vvar_{i}, \vvar_{j})$ is in the compatible graph with $\inducedgraph$ (resp., $\g$)
if and only if $\BV^{\inducedgraph}_{i, j}$ (resp., $\BV_{i, j}$) is assigned to \True{} by MonoSAT
in testing the acyclicity of $\inducedgraph$.

We first encode the polygraph $\g$.
For each edge $(\vvar_{i}, \vvar_{j})$ in the known graph of $\g$,
we add a clause $\BV_{i, j} = \True$.
For each constraint $\tuple{\eithervar, \orvar}$, the clause
$\bigl(\bigwedge\limits_{(\vvar_{i}, \vvar_{j}, \_) \in \eithervar}
  \!\!\!\!\BV_{i, j}
  \land \bigwedge\limits_{(\vvar_{i}, \vvar_{j}, \_) \in \orvar}
  \!\!\!\!\lnot\BV_{i, j}\bigr) \;\lor\;
\bigl(\bigwedge\limits_{(\vvar_{i}, \vvar_{j}, \_) \in \orvar}
  \!\!\!\!\BV_{i, j}
  \land \bigwedge\limits_{(\vvar_{i}, \vvar_{j}, \_) \in \eithervar}
  \!\!\!\!\lnot\BV_{i, j}\bigr)$
expresses that exactly one of $\eithervar$ or $\orvar$ happens.


Then we encode the induced \si{} graph $\inducedgraph$ of $\g$.
The auxiliary graph $\graphA$ contains all the known and potential $\SO$, $\WR$, and $\WW$ edges of $\g$
(lines~\linecode{\ref{alg:checksi}}{\ref{line:encodeconstraints-A}}
and~\linecode{\ref{alg:checksi}}{\ref{line:encodeconstraints-add-ww-to-A}}),
while $\graphB$ contains all the known and potential $\RW$ edges of $\g$
(lines~\linecode{\ref{alg:checksi}}{\ref{line:encodeconstraints-B}}
and~\linecode{\ref{alg:checksi}}{\ref{line:encodeconstraints-add-rw-to-B}}).
The clauses defined on $\BV^{\inducedgraph}$ at line~\linecode{\ref{alg:checksi}}{\ref{line:encodeconstraints-add-clauses-for-I}}
state that $\inducedgraph$ is the union of $\graphA$
and the composition of $\graphA$ with $\graphB$.

%% file: sections/monosat-solving.tex



%% file: sections/completing.tex

\subsection{Completing the SI Checking} \label{ss:completing-si-checking}

Theorem~\ref{thm:depgraph-si} assumes histories with only committed transactions
and considers the $\WR$, $\WW$, and $\RW$ relations over transactions
rather than read/write operations inside them.
This would miss  non-cycle anomalies.
Hence, for completeness, \name{} also checks whether a history
exhibits $\abortedreads$ or $\intermediatereads$ anomalies~\cite{Adya:PhDThesis1999,Elle:VLDB2020}
(line~\linecode{\ref{alg:checksi}}{\ref{line:checksi-intaxiom}}):

\begin{itemize}
  \item \emph{Aborted Reads}:
    a committed transaction cannot read a value from an aborted 
    transaction.
  \item \emph{Intermediate Reads}:
    a transaction cannot read a value that was overwritten by the transaction that wrote it.
\end{itemize}

Note that \name{}'s completeness relies on a common assumption  about  \emph{determinate} transactions \cite{Adya:PhDThesis1999,  Framework:CONCUR2015,ClientCentric:PODC2017,AnalysingSI:JACM2018,Elle:VLDB2020,Complexity:OOPSLA2019}, i.e., the status of each transaction, whether committed or aborted,  is legitimately decided.  
Indeterminate transactions are inherent to black-box testing:
it is difficult for a client to justify the status of a transaction due to the invisibility of system internals.
Together with the completeness of the dependency-graph-based characterization of \si{} in Theorem~\ref{thm:depgraph-si}, 
we prove \name{}'s completeness.

\begin{theorem}[Completeness of \name] \label{thm:completeness}
  \emph{\name{}} is complete with respect to a history that contains only determinate transactions,
  i.e., if such a history indeed violates \emph{\si},
  then \emph{\name{}} returns \false.
\end{theorem}

%% file: sections/experiments.tex

\section{Experiments}  \label{section:experiments}


We have presented our SI checking algorithm \name{} and established its \emph{soundness} and \emph{completeness}.
 In this section, we conduct a comprehensive assessment of  \name{}
 to answer the following questions with respect to the remaining criteria of  SIEGE+ (Section~\ref{section:intro}):


\vspace{1ex}
\noindent \textbf{(1) Effective:}
Can \name{} find SI violations in (production) databases?

\vspace{1ex}
\noindent \textbf{(2) Informative:}
  Can \name{} provide understandable counterexamples for SI violations?

\vspace{1ex}
\noindent \textbf{(3) Efficient:}
How efficient is \name{} (and its components)?
Can \name{} outperform the state of the art  under \emph{various} workloads and scale up to large-sized workloads?

\vspace{1ex}
Our answer to (1) is twofold (Section~\ref{ss:effectiveness}):
(i) \name{} successfully reproduces all of 2477 known SI anomalies in production databases;
and (ii) we use \name{} to detect novel SI violations in three cloud databases of different kinds:
 the graph database Dgraph~\cite{dgraph},  the relational database MariaDB-Galera~\cite{maria-galera}, and YugabyteDB~\cite{YugabyteDB} supporting multiple data models.
To answer (2) we provide an  algorithm that  recovers the violating scenario, highlighting  the cause of the violation found
 (Section \ref{ss:understanding-violations}).
Regarding  (3),  we (i) show that \name{} outperforms several competitive baselines including the most performant SI and serializability checkers to date; (ii) measure the contributions of its different components/optimizations to the overall performance
under both general and specific transaction workloads
 (Section \ref{ss:efficiency}); 
 and (iii) demonstrate its scalability for large-sized workloads with one billion keys and one million transactions.
Note that we demonstrate \name{}'s \textbf{generality} along with the answers to questions (1) and (3). 


\subsection{Workloads,  Benchmarks, and Setup}
\subsubsection{Workloads and Benchmarks}
To evaluate \name{} on
 \emph{general} read-only, write-only, and read-write transaction workloads, we have implemented a parametric
 workload generator.
 Its parameters are:  the number of client sessions (\#sess; 20 by default), the
number of transactions per session (\#txns/sess; 100 by default), the number of read/write operations per
transaction (\#ops/txn; 15 by default), the percentage of reads (\%reads; 50\% by default), the total number of keys (\#keys; 10k by default), and the
key-access distribution (dist) including uniform,  zipfian (by default),  and hotspot (80\% operations touching 20\% keys).
Note that the default 2k transactions with 30k operations issued by 20 sessions are sufficient to distinguish \name{} from  competing tools (see Section \ref{sss:comparison}).

%

Among such general workloads,
we also consider three representatives,  each with 10k transactions and 80k operations in total (\#sess=25,  \#txns/sess=400, and \#ops/txn=8),  in the comparison with Cobra and the decomposition and differential analysis of \name{}: (i) GeneralRH, read-heavy workload with 95\% reads; (ii) GeneralRW, medium workload with 50\% reads; and (iii) GeneralWH, write-heavy workloads with 30\% reads.




We also use three synthetic 
 benchmarks 
 with 
 only  serializable histories of 
 at least 10k transactions (which also satisfy SI):

\begin{itemize}
\item RUBiS~\cite{crubis}:   an eBay-like bidding system where users can, for example, register and bid for items.  The dataset archived by \cite{Cobra:OSDI2020} contains 20k
users and 200k items.

\item TPC-C~\cite{tpcc}: an open standard for benchmarking
 online transaction processing with a mix of five  different types of transactions (e.g.,  for orders and payment)
  portraying the activity of a wholesale supplier.  The dataset includes one warehouse, 10 districts,
and 30k customers.


\item C-Twitter~\cite{ctwitter}: a Twitter clone where users can,  for example, tweet and follow or unfollow other users (following the zipfian distribution).


\end{itemize}

To assess \name{}'s scalability, we also consider large-sized workloads with one billion keys and one million transactions (\#sess=20; \#txns/sess=50k). The workloads contain both short and long transactions;
the default sizes are 15 and 150, respectively.


%
%


\subsubsection{Setup.}
We use a PostgreSQL (v15 Beta 1) instance to produce \emph{valid} histories without isolation violations:
for the performance comparison with other SI checkers and the decomposition and differential analysis of \name{} itself,
 we set the isolation level to \emph{repeatable read} (implemented as SI in PostgreSQL \cite{postgresql-rr});
 for the runtime comparison with Cobra (Section \ref{sss:comparison}), we  use  the  \emph{serializability} isolation level
  to produce serializable histories.
We co-locate the client threads and PostgreSQL (or other databases  for testing; see Section \ref{sss:new-violation}) on a local machine.
 Each client thread issues a stream of transactions produced by our workload generator to the database and records the execution history.  All histories are   saved to a file  to benchmark each tool's performance.

We have implemented \name{} in 2.3k  lines of Java code,  and  
 the workload generator,  including the transformation from generated key-value operations to SQL queries (for the interactions with relational databases such as PostgreSQL), in 2.2k 
   lines of Rust code.  We ensure unique values written for each key using counters.
We use a simple database schema of  a two-column  table storing
 keys and values,  which is effective to find real violations in three production databases (see Section \ref{ss:effectiveness}).

We conducted all
 experiments 
 with a 4.5GHz
Intel Xeon E5-2620 (6-core) CPU, 48GB memory,  and an  NVIDIA  K620 GPU.

\input{sections/effectiveness}
\input{sections/understandability}
\input{sections/efficiency}



%% file: sections/effectiveness.tex

\subsection{Finding SI Violations} \label{ss:effectiveness}


\subsubsection{Reproducing Known \si{} Violations} \label{sss:reproduce}
 \name{} successfully reproduces \emph{all} known SI violations in an extensive collection of 2477 anomalous
  histories \cite{Complexity:OOPSLA2019,YugabyteDB-bug,CockroachDB-bug}. These histories  were obtained  from the earlier releases of three different production databases,
  i.e., CockroachDB, MySQL-Galera, and YugabyteDB; see Table \ref{benchmark} for details.  This set of experiments  
 provides supporting evidence for   
   \name{}'s \emph{soundness} and \emph{completeness}, established in Section \ref{section:si-sat}.
  
\begin{table}[t]
\centering
\caption{Summary of tested databases.   Multi-model refers to relational DBMS,  document store, and
wide-column store. }\label{benchmark} 
\small
\begin{tabular}{cccl}
  \toprule
  Database &	GitHub Stars &  Kind & Release \\
  \hline
  \textbf{New violations found:} &  &   &  \\
  Dgraph & 18.2k &  Graph & v21.12.0 \\
  MariaDB-Galera & 4.4k &  Relational & v10.7.3 \\
  YugabyteDB & 6.7k &  Multi-model & v2.11.1.0 \\
  \hline
  \textbf{Known bugs \cite{Complexity:OOPSLA2019,YugabyteDB-bug,CockroachDB-bug}:} &	 &   &  \\
  CockroachDB &	25.1k  &  Relational & v2.1.0\\ 
           & &   & v2.1.6 \\
  MySQL-Galera &	381  &  Relational & v25.3.26 \\
  YugabyteDB & 6.7k	 &  Multi-model & v1.1.10.0 \\

\bottomrule
\end{tabular}
\normalsize
\end{table}

\subsubsection{Detecting New  Violations.} \label{sss:new-violation}
We use \name{} to examine   recent releases of  three  well-known cloud databases (of different kinds) that  claim to provide SI:
Dgraph~\cite{dgraph}, MariaDB-Galera~\cite{maria-galera}, and YugabyteDB~\cite{YugabyteDB}.  See  Table \ref{benchmark} for details.  We have found and reported novel SI violations in all three databases which, as of the time of writing, are being investigated by the developers. 
In particular,   as communicated with the developers,  
(i)
our finding has helped the DGraph team confirm some of their  suspicions about their latest release; and (ii)
Galera has confirmed the incorrect claim on preventing lost updates for transactions issued on different cluster nodes and 
thereafter removed any claims on SI or ``partially supporting SI'' from the previous documentation.\footnote{\url{https://github.com/codership/documentation/commit/cc8d6125f1767493eb61e2cc82f5a365ecee6e7a} and \url{https://github.com/codership/documentation/commit/d87171b0d1b510fe59973cb7ce5892061ce67b80}}

%% file: sections/understandability.tex



\subsection{Understanding Violations} \label{ss:understanding-violations}
MonoSAT reports 
 cycles, constructed from its output logs, upon detecting an \si{} violation.
However,  such cycles
 are   \emph{uninformative} with respect to  understanding how the violation actually occurred.
For instance,  Figure~\ref{ce:galera}(a)   depicts the original cycle  returned by MonoSAT  for an SI violation found in MariaDB-Galera, where it is difficult to  identify the  cause of the violation.

Hence,
 we have designed an algorithm to interpret the returned cycles.
The key idea is to
(i) bring back any potentially
involved transactions and the associated dependencies,
(ii) restore the violating scenario by identifying the core participants and dependencies,
and (iii) remove the ``irrelevant'' dependencies to simplify the scenario.
We have integrated  into  \name{}  the algorithm written in 300 lines of C++ code.  The pseudocode is given in Appendix~\ref{section:appendix-interpretation}.  We have also integrated the Graphviz tool \cite{graphviz} into \name{} to visualize the final counterexamples (e.g., Figure~\ref{ce:galera}).

\vspace{1ex}
\noindent \textbf{Minimal Counterexample.} 
	A ``minimal'' counterexample would facilitate  understanding  how the violation actually occurred. 
	We define a minimal violation as a polygraph where no dependency can be removed; otherwise, the resulting polygraph would pass the verification of \name{}. 
 Given a polygraph $G$ (constructed from a collected history) and a cycle $C$ (returned by MonoSAT), 
there may however be more than one minimal violation with respect to $G$ and $C$ due to different interpretations of uncertain dependencies.
We call the one with  the least number of dependencies 
 \emph{the minimal counterexample with respect to $G$ and $C$}.

 \name{}  guarantees the minimality  of returned counterexample:

\begin{theorem}[Minimality]
	 \label{thm:minimal-counterexample-polysi}
	\name{} always returns a minimal counterexample with respect to $G$ and $C$, with
	$G$ the polygraph built from a history and  $C$ the cycle output by MonoSAT.
\end{theorem}

\looseness=-1
	We defer to Appendix~\ref{section:appendix-minimality} for the formal definitions of the minimal violation and counterexample and the proof of Theorem \ref{thm:minimal-counterexample-polysi}.

\begin{figure*}[t]
  \centering
    \begin{subfigure}[b]{0.2\textwidth}
        \centering
        \includegraphics[width = \textwidth]{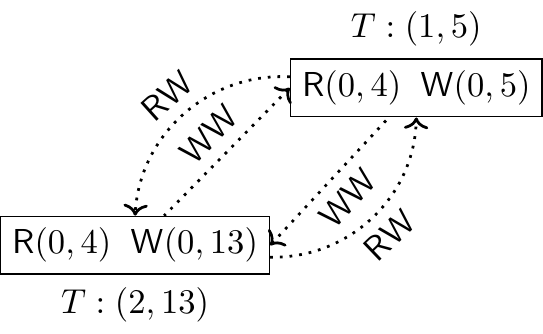}
        \caption{Original output}
    \end{subfigure}\hspace{3ex}
    \begin{subfigure}[b]{0.23\textwidth}
        \centering
        \includegraphics[width = \textwidth]{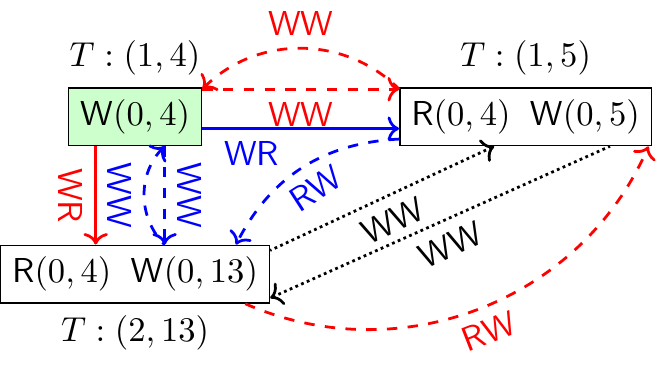}
        \caption{Missing  participants}
    \end{subfigure}\hspace{3ex}
    \begin{subfigure}[b]{0.23\textwidth}
        \centering
        \includegraphics[width = \textwidth]{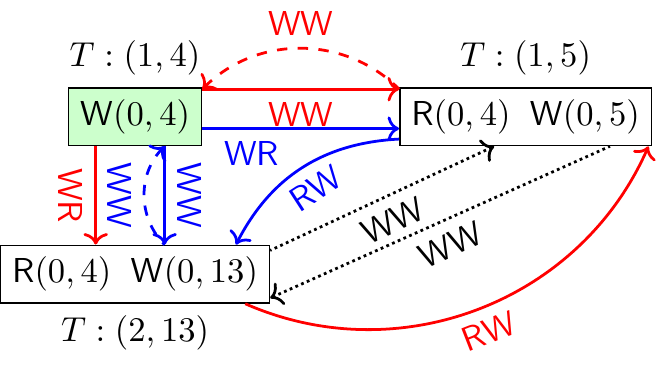}
        \caption{Recovered scenario}
    \end{subfigure}\hspace{3ex}
    \begin{subfigure}[b]{0.23\textwidth}
        \centering
        \includegraphics[width = \textwidth]{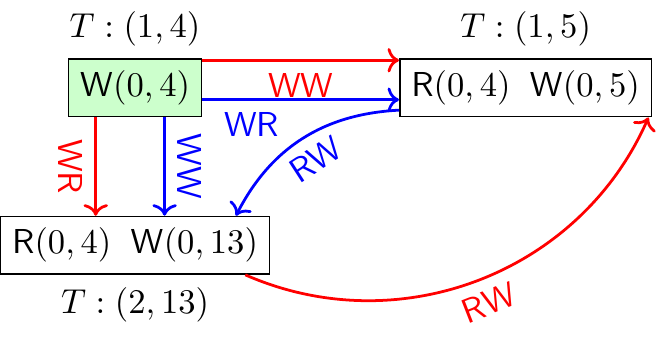}
        \caption{Finalized scenario}
    \end{subfigure}
    \caption{ \label{ce:galera}  Lost update: the SI violation  found in MariaDB-Galera.
    The original output dependencies  are represented by dotted black arrows.
    The recovered dependencies are colored in red/blue with dashed and solid arrows  representing uncertain  and certain dependencies,  respectively.  The  missing transaction is colored in green.
We omit  key 0, associated with all dependencies.    }
\end{figure*}

\vspace{1ex}
\noindent \textbf{Violation Found in MariaDB-Galera.}
We present an example violation detected in MariaDB-Galera. 
  In particular,
we illustrate how the interpretation algorithm 
helps us locate the violation cause: \emph{lost update}. 
We defer the Dgraph and YugabyteDB anomalies (causality violations) to Appendix~\ref{section:appendix-causality}.
In the following example, we use T:$(s, n)$ to
denote the $n$th transaction issued by session $s$.


 Given the original cycles returned by  MonoSAT in Figure \ref{ce:galera}(a),  \name{} first finds  the (only) ``missing'' transaction T:(1,4) (colored in green)
  and the associated  dependencies, as shown in Figure \ref{ce:galera}(b).  Note that some of the dependencies are uncertain at this moment, e.g.,  the $\WW$ dependency between T:(1,4) and T:(1,5) (in red).
\name{} then restores the violating scenario by resolving such uncertainties.  For example,  as depicted in Figure \ref{ce:galera}(c),  \name{} determines that
W(0,4) was actually installed first in the database,  i.e.,  T:(1,4)$\rel{\WW}$T:(1,5), because   there would otherwise  be an undesired cycle with the known dependencies,  i.e.,
T:(1,5)$\rel{\WW}$T:(1,4)$\rel{\WR}$T:(1,5).  The same reasoning applies to determine the $\WW$ dependency between T:(1,4) and T:(2,13) (in blue).
Finally,  \name{} finalizes the violating scenario by removing any remaining uncertainties including those dependencies not involved in the actual violation (the $\WW$ dependency between T:(1,5) and T:(2,13) in this case).

The violating scenario now becomes informative and explainable:
transaction T:(1,4)  writes value 4 on key 0,  which is read by  transactions T:(2,13) and T:(1,5).  Both transactions subsequently commit their writes on key 0 by W(0,13) and W(0,5), respectively,  which results in  a \emph{lost update} anomaly.

%% file: sections/efficiency.tex
\subsection{Performance Evaluation} \label{ss:efficiency}

In this section, we conduct an in-depth performance analysis of \name{}
and compare it to the following black-box checkers:

\begin{itemize}
  \item dbcop \cite{Complexity:OOPSLA2019} is,
  to the best of our knowledge,  the  most efficient black-box SI checker that does not use an off-the-shelf solver.
Note that, unlike our \name{} tool, dbcop does not check \emph{{aborted reads}} or \emph{intermediate reads} (see  Section \ref{ss:completing-si-checking}).

  \item Cobra \cite{Cobra:OSDI2020} is the state-of-the-art SER checker utilizing both  MonoSAT  and GPUs to  accelerate the checking procedure.
Cobra serves as a  
baseline because (i) checking SI is more complicated than checking SER in general  \cite{Complexity:OOPSLA2019},  and constraint pruning and the MonoSAT encoding for SI are  more challenging in particular due to 
more complex cycle patterns in dependency graphs (Theorem \ref{thm:depgraph-si}, Section \ref{sss:depgraph-si});
 and (ii) Cobra is the  most  performant  SER checker to date.


  \item CobraSI:  We implement the incremental algorithm \cite[Section 4.3]{Complexity:OOPSLA2019} for reducing checking SI to checking serializability (in polynomial time) to leverage  Cobra.
    We consider two variants: (i) CobraSI without GPU for a fair comparison with  \name{} and dbcop, which
    do not employ GPU or multithreading; and (ii)
     CobraSI with GPU as a strong competitor.


%
%
%

\end{itemize}


\input{tables/runtime}

\subsubsection{Performance Comparison with State of the Art.}
\label{sss:comparison}

\looseness=-1
Our first set of experiments compares \name{} with the competing SI checkers under a wide range of workloads.
The input histories  extracted from PostgreSQL (with the \emph{repeatable read} isolation level) are all valid with respect to SI.  The experimental results are shown in Figure~\ref{comparison}:
 \name{} significantly surpasses not only  the state-of-the-art SI checker dbcop but
 also CobraSI with GPU.
  In particular,  with more concurrency, such as  more sessions (a),  transactions per session (b), and operations per transaction (c),
CobraSI with GPU exhibits exponentially increasing checking time\footnote{Two major reasons are: (i) Cobra has already been shown to exhibit exponential verification time under general workloads \cite{Cobra:OSDI2020}; and (ii) the incremental algorithm for reducing checking SI to checking serializability typically doubles the number of transactions in a given history \cite{Complexity:OOPSLA2019}, rendering the checking even more expensive. } while \name{}  incurs only moderate overhead.  The result depicted in Figure \ref{comparison}(f) is also consistent: with the skewed key accesses representing high  concurrency as in the zipfian and hotspot distributions,  both dbcop and CobraSI without GPU acceleration time out.  Moreover,  even with the GPU acceleration,  CobraSI takes 6x more time than \name{}.
Finally,  unlike the other SI checkers,  \name{}'s performance is fairly stable with respect to  varying read/write proportions (d) and keys (e).


In Figure \ref{comparison-cobra}(a)  we compare \name{}
 with the baseline serializability checker Cobra.  We present
 the checking time on various benchmarks.  \name{} outperforms Cobra (with its GPU acceleration enabled) in five of the six benchmarks with up to 3x improvement (as for GeneralRH).
The only exception is TPC-C, where most of the transactions 
have the read-modify-write pattern,\footnote{In a read-modify-write transaction each read is followed by a write on the same key.} for which  Cobra
implements a specific optimization to efficiently infer dependencies before pruning and encoding.



	We also measure the memory usage for all the checkers under the same settings as in Figure~\ref{comparison} and Figure~\ref{comparison-cobra}(a).
As shown in Figure~\ref{comparison-memory}, \name{} consumes less memory (for storing both  generated graphs and constraints) than the competitors in general. Note that
 dbcop, the only checker that does not rely on solving and stores no constraints, is not competitive with \name{} for most of the cases. Regarding the comparison on specific benchmarks (Figure~\ref{comparison-cobra}(b)), \name{} and Cobra with GPU acceleration have similar overheads, while \name{} (resp. Cobra) requires less  memory for read-heavy workloads (resp. TPC-C).

\begin{table}[h]
	\centering
	\caption{Number of constraints and unknown dependencies before and after  pruning (P)  in the six benchmarks.}\label{benchmark-stat}
	\begin{tabular}{c|rr|rr}
		Benchmark  & \#cons.  & \#cons. & \#unk.  dep. & \#unk.  dep. \\
		& before P & after P & before P     & after P      \\
		\hline

		TPC-C      & 386k     & 0       & 3628k        & 0            \\
		GeneralRH  & 4k       & 29      & 39k          & 77           \\
		RUBiS      & 14k      & 149     & 171k         & 839          \\
		C-Twitter  & 59k      & 277     & 307k         & 776          \\

		GeneralRW  & 90k      & 2565    & 401k         & 5435         \\
		GeneralWH  & 167k     & 6962    & 468k         & 14376        \\


	\end{tabular}
	\vspace{-3ex}
\end{table}

\subsubsection{Decomposition Analysis of \name} \label{sss:decomposition}

We measure \name's checking time in terms of stages:
\emph{constructing},  which builds up a generalized polygraph from a given history;
\emph{pruning}, which prunes constraints in the generalized polygraph;
\emph{encoding}, which encodes the graph and the remaining constraints; and
\emph{solving}, which  runs the MonoSAT solver.

Figure \ref{fig:decomposition} depicts the results on six different datasets. 
Constructing a generalized polygraph is  relatively inexpensive.
The overhead of pruning is fairly constant, regardless of the workloads; \name{} can effectively prune (resp. resolve) a huge number of constraints (resp. unknown dependencies)  in this phase.
See Table~\ref{benchmark-stat} for details.
In particular, for TPC-C which contains only read-only and read-modify-write transactions,
\name{} is able to resolve all uncertainties on $\WW$ relations
and identify the unique version chain for each key.
 The encoding effort is moderate; TPC-C incurs more overhead as the number of operations in total is 5x more than the others.
 The  solving time  depends on the remaining constraints and unknown dependencies after pruning,  e.g.,  the left four datasets incur negligible overhead  (see  Table \ref{benchmark-stat}).


\subsubsection{Differential Analysis of \name{}} \label{sss:differential}


To investigate the contributions of \name{}'s two major optimizations,
we experiment with three variants:
(i) \name{} itself;  (ii) \name{} without pruning (P) constraints;
and (iii) \name{} without both  compacting (C) and pruning the constraints. 
Figure~\ref{fig:differential-analysis} demonstrates the acceleration produced by each optimization.  Note that the two variants without optimization exhibit (16GB) memory-exhausted  runs on TPC-C, which contain considerably more uncertain dependencies (3628k) and constraints (386k) without pruning than the other datasets (see Table \ref{benchmark-stat}).




\subsubsection{Scalability.}
	To assess \name{}'s scalability, we generate transaction workloads with one billion keys and one million transactions with hundreds of millions of operations.
	We experiment with varying read proportions and long transaction sizes (up to 450 operations per transaction).
As shown in Figure~\ref{large-sized-exp}, \name{} consumes less than 40GB memory in all cases and at  most 4 hours  for checking one million transactions.
We also observe that the time used increases linearly with larger-sized transactions while the memory overhead is fairly stable.
To conclude, large-sized workloads are
 quite
manageable for \name{} on modern hardware. Note that the competing checkers, as expected, fail to handle such workloads.

\pgfplotsset{every axis/.append style={font=\LARGE}}
\begin{figure}[t]
	\begin{scaletikzpicturetowidth}{0.20\textwidth}
		\begin{tikzpicture}[scale=\tikzscale]
			\begin{axis}[
				title={(a)},
				xlabel={read proportion \%},
				ylabel={Time (h)},
                ymin=0,
                ymax=7,
				cycle multiindex* list={
					color       \nextlist
					mark list*  \nextlist
				}
				]
				\addplot[color=brown,mark=o,mark size=3pt] table [x=readpct, y=time, col sep=comma] {tables/data/random-large-readpct.csv};
				\legend{\name}
			\end{axis}
		\end{tikzpicture}
		\hspace{3ex}
	\end{scaletikzpicturetowidth}
	\begin{scaletikzpicturetowidth}{0.20\textwidth}
	\begin{tikzpicture}[scale=\tikzscale]
		\begin{axis}[
			title={(b)},
			xlabel={read proportion \%},
			ylabel={Memory (GB)},
            ymin=30,
            ymax=45,
			cycle multiindex* list={
				color       \nextlist
				mark list*  \nextlist
			}
			]
			\addplot[color=brown,mark=o,mark size=3pt] table [x=readpct, y=memory, col sep=comma] {tables/data/random-large-readpct.csv};
			\legend{\name}
		\end{axis}
	\end{tikzpicture}
	\end{scaletikzpicturetowidth}


	\begin{scaletikzpicturetowidth}{0.20\textwidth}
		\begin{tikzpicture}[scale=\tikzscale]
			\begin{axis}[
				title={(c)},
				xlabel={\#ops per long transaction},
				ylabel={Time (h)},
                ymin=2,
                ymax=6,
                legend pos=north west,
				]
				\addplot[color=brown,mark=o,mark size=3pt] table [x=size, y=time, col sep=comma] {tables/data/random-large-longtxn-size.csv};
				\legend{\name}
			\end{axis}
		\end{tikzpicture}
	\end{scaletikzpicturetowidth}
	\hspace{3ex}
	\begin{scaletikzpicturetowidth}{0.20\textwidth}
		\begin{tikzpicture}[scale=\tikzscale]
			\begin{axis}[
				title={(d)},
				xlabel={\#ops per long transaction},
				ylabel={Memory (GB)},
                ymin=34,
                ymax=38,
                legend pos=north west,
				]
				\addplot[color=brown,mark=o,mark size=3pt] table [x=size, y=memory, col sep=comma] {tables/data/random-large-longtxn-size.csv};
				\legend{\name}
			\end{axis}
		\end{tikzpicture}
		\hspace{3ex}
	\end{scaletikzpicturetowidth}
	\caption{ \name{}'s overhead on large-sized workloads with one billion keys and one million transactions.  }
	\label{large-sized-exp}
\end{figure}

%% file: tables/runtime.tex
\pgfplotsset{height=140pt, width=200pt}
\begin{figure}[t]
    \begin{scaletikzpicturetowidth}{0.23\textwidth}
	\begin{tikzpicture}[scale=\tikzscale]
	  \begin{axis}[
		title={(a)},
		xlabel={\#sessions},
		ylabel={Time (s)},
		ymax=150,
		cycle multiindex* list={
            color       \nextlist
            mark list*  \nextlist
        }
		]
		\addplot[color=blue,mark=square,mark size=3pt] table [x=param, y=cobra(si), col sep=comma] {tables/data/sessions.csv};
		\addplot[color=red,mark=triangle,mark size=3pt] table [x=param, y=cobra(si-nogpu), col sep=comma] {tables/data/sessions.csv};
		\addplot[color=black,mark=x,mark size=3pt] table [x=param, y=oopsla, col sep=comma] {tables/data/sessions.csv};
		\addplot[color=brown,mark=o,mark size=3pt] table [x=param, y=si, col sep=comma] {tables/data/sessions.csv};
		\legend{CobraSI w/ GPU,CobraSI w/o GPU,dbcop,\name}
	  \end{axis}
	\end{tikzpicture}
		\hspace{1ex}
  \end{scaletikzpicturetowidth}           
  \begin{scaletikzpicturetowidth}{0.23\textwidth}
	\begin{tikzpicture}[scale=\tikzscale]
	  \begin{axis}[
		title={(b)},
		xlabel={\#txns/session},
		ylabel={Time (s)},
		ymax=100,
		legend pos=north west
		]
		\addplot[color=blue,mark=square,mark size=3pt] table [x=param, y=cobra(si), col sep=comma] {tables/data/txns.csv};
		\addplot[color=red,mark=triangle,mark size=3pt] table [x=param, y=cobra(si-nogpu), col sep=comma] {tables/data/txns.csv};
		\addplot[color=black,mark=x,mark size=3pt] table [x=param, y=oopsla, col sep=comma] {tables/data/txns.csv};
		\addplot[color=brown,mark=o,mark size=3pt] table [x=param, y=si, col sep=comma] {tables/data/txns.csv};
		\legend{CobraSI w/ GPU,CobraSI w/o GPU,dbcop,\name}
	  \end{axis}
	\end{tikzpicture}
  \end{scaletikzpicturetowidth}

  \begin{scaletikzpicturetowidth}{0.23\textwidth}
	\begin{tikzpicture}[scale=\tikzscale]
	  \begin{axis}[
		title={(c)},
		xlabel={\#ops/txn},
		ylabel={Time (s)},
		ymax=100,
		]
		\addplot[color=blue,mark=square,mark size=3pt] table [x=param, y=cobra(si), col sep=comma] {tables/data/nops.csv};
		\addplot[color=red,mark=triangle,mark size=3pt] table [x=param, y=cobra(si-nogpu), col sep=comma] {tables/data/nops.csv};
		\addplot[color=black,mark=x,mark size=3pt] table [x=param, y=oopsla, col sep=comma] {tables/data/nops.csv};
		\addplot[color=brown,mark=o,mark size=3pt] table [x=param, y=si, col sep=comma] {tables/data/nops.csv};
		\legend{CobraSI w/ GPU,CobraSI w/o GPU,dbcop,\name}
	  \end{axis}
	\end{tikzpicture}
  \end{scaletikzpicturetowidth}
	\hspace{1ex}
  \begin{scaletikzpicturetowidth}{0.23\textwidth}
	\begin{tikzpicture}[scale=\tikzscale]
	  \begin{axis}[
		title={(d)},
		xlabel={ read proportion \%},
		ylabel={Time (s)},
		ymax=179,
		legend pos=north east
		]
		\addplot[color=blue,mark=square,mark size=3pt] table [x=param, y=cobra(si), col sep=comma] {tables/data/readpct.csv};
		\addplot[color=red,mark=triangle,mark size=3pt] table [x=param, y=cobra(si-nogpu), col sep=comma] {tables/data/readpct.csv};
		\addplot[color=black,mark=x,mark size=3pt] table [x=param, y=oopsla, col sep=comma] {tables/data/readpct.csv};
		\addplot[color=brown,mark=o,mark size=3pt] table [x=param, y=si, col sep=comma] {tables/data/readpct.csv};
		\legend{CobraSI w/ GPU,CobraSI w/o GPU,dbcop,\name}
	  \end{axis}
	\end{tikzpicture}

  \end{scaletikzpicturetowidth}
    \begin{scaletikzpicturetowidth}{0.23\textwidth}
	\begin{tikzpicture}[scale=\tikzscale]
	  \begin{axis}[
		title={(e)},
		xlabel={\#keys},
		ylabel={Time (s)},
		ymax=100,
		]
		\addplot[color=blue,mark=square,mark size=3pt] table [x=param, y=cobra(si), col sep=comma] {tables/data/vars.csv};
		\addplot[color=red,mark=triangle,mark size=3pt] table [x=param, y=cobra(si-nogpu), col sep=comma] {tables/data/vars.csv};
		\addplot[color=black,mark=x,mark size=3pt] table [x=param, y=oopsla, col sep=comma] {tables/data/vars.csv};
		\addplot[color=brown,mark=o,mark size=3pt] table [x=param, y=si, col sep=comma] {tables/data/vars.csv};
		\legend{CobraSI w/ GPU,CobraSI w/o GPU,dbcop,\name}
	  \end{axis}
	\end{tikzpicture}
  \end{scaletikzpicturetowidth}
  	\hspace{1ex}
  \begin{scaletikzpicturetowidth}{0.23\textwidth}
	\begin{tikzpicture}[scale=\tikzscale]
	  \begin{axis}[
		title={(f)},
		x tick style={draw=none},
		ylabel={Time (s)},
		ymin=0,
		ymax=100,
		xmin=-0.5,
		xmax=2.5,
		ybar,
		area legend,
		xtick=data,
		xticklabels={uniform,zipfian,hotspot}
		]
		\addplot[color=blue, pattern color=blue, pattern=crosshatch] table [x expr=\coordindex, y=cobra(si), col sep=comma] {tables/data/distrib.csv};
		\addplot[color=red, pattern color=red, pattern=crosshatch dots] table [x expr=\coordindex, y=cobra(si-nogpu), col sep=comma] {tables/data/distrib.csv};
		\addplot[color=black, pattern color=black,pattern=north west lines] table [x expr=\coordindex, y=oopsla, col sep=comma] {tables/data/distrib.csv};
		\addplot[color=brown, pattern color=brown, fill=brown] table [x expr=\coordindex, y=si, col sep=comma] {tables/data/distrib.csv};
		\legend{CobraSI w/ GPU,CobraSI w/o GPU,dbcop,\name}
	  \end{axis}
	\end{tikzpicture}
  \end{scaletikzpicturetowidth}
 \caption{Performance comparison with the competing SI checkers under various workloads.   
 Experiments time out at 180s; data points are not plotted for timed-out experiments.}
 \label{comparison}
\end{figure}

\begin{figure}[t]
	\begin{scaletikzpicturetowidth}{0.23\textwidth}
		\begin{tikzpicture}[scale=\tikzscale]
			\begin{axis}[
				title={(a)},
				xlabel={\#sessions},
				ylabel={Memory (MB)},
				cycle multiindex* list={
					color       \nextlist
					mark list*  \nextlist
				},
				legend pos=north west
				]
				\addplot[color=blue,mark=square,mark size=3pt] table [x=param, y=cobra(si), col sep=comma] {tables/data/memory-sessions.csv};
				\addplot[color=red,mark=triangle,mark size=3pt] table [x=param, y=cobra(si-nogpu), col sep=comma] {tables/data/memory-sessions.csv};
				\addplot[color=black,mark=x,mark size=3pt] table [x=param, y=oopsla, col sep=comma] {tables/data/memory-sessions.csv};
				\addplot[color=brown,mark=o,mark size=3pt] table [x=param, y=si, col sep=comma] {tables/data/memory-sessions.csv};
				\legend{CobraSI w/ GPU, CobraSI w/o GPU,dbcop,\name}
			\end{axis}
		\end{tikzpicture}
		\hspace{1ex}
	\end{scaletikzpicturetowidth}
	\begin{scaletikzpicturetowidth}{0.23\textwidth}
		\begin{tikzpicture}[scale=\tikzscale]
			\begin{axis}[
				title={(b)},
				xlabel={\#txns/session},
				ylabel={Memory (MB)},
				legend pos=north west
				]
				\addplot[color=blue,mark=square,mark size=3pt] table [x=param, y=cobra(si), col sep=comma] {tables/data/memory-txns.csv};
				\addplot[color=red,mark=triangle,mark size=3pt] table [x=param, y=cobra(si-nogpu), col sep=comma] {tables/data/memory-txns.csv};
				\addplot[color=black,mark=x,mark size=3pt] table [x=param, y=oopsla, col sep=comma] {tables/data/memory-txns.csv};
				\addplot[color=brown,mark=o,mark size=3pt] table [x=param, y=si, col sep=comma] {tables/data/memory-txns.csv};
				\legend{CobraSI w/ GPU, CobraSI w/o GPU,dbcop,\name}
			\end{axis}
		\end{tikzpicture}
	\end{scaletikzpicturetowidth}

	\begin{scaletikzpicturetowidth}{0.23\textwidth}
		\begin{tikzpicture}[scale=\tikzscale]
			\begin{axis}[
				title={(c)},
				xlabel={\#ops/txn},
				ylabel={Memory (MB)},
				legend pos=north west
				]
				\addplot[color=blue,mark=square,mark size=3pt] table [x=param, y=cobra(si), col sep=comma] {tables/data/memory-nops.csv};
				\addplot[color=red,mark=triangle,mark size=3pt] table [x=param, y=cobra(si-nogpu), col sep=comma] {tables/data/memory-nops.csv};
				\addplot[color=black,mark=x,mark size=3pt] table [x=param, y=oopsla, col sep=comma] {tables/data/memory-nops.csv};
				\addplot[color=brown,mark=o,mark size=3pt] table [x=param, y=si, col sep=comma] {tables/data/memory-nops.csv};
				\legend{CobraSI w/ GPU, CobraSI w/o GPU,dbcop,\name}
			\end{axis}
		\end{tikzpicture}
	\end{scaletikzpicturetowidth}
	\hspace{1ex}
	\begin{scaletikzpicturetowidth}{0.23\textwidth}
		\begin{tikzpicture}[scale=\tikzscale]
			\begin{axis}[
				title={(d)},
				xlabel={ read proportion \%},
				ylabel={Memory (MB)},
				legend pos=north west
				]
				\addplot[color=blue,mark=square,mark size=3pt] table [x=param, y=cobra(si), col sep=comma] {tables/data/memory-readpct.csv};
				\addplot[color=red,mark=triangle,mark size=3pt] table [x=param, y=cobra(si-nogpu), col sep=comma] {tables/data/memory-readpct.csv};
				\addplot[color=black,mark=x,mark size=3pt] table [x=param, y=oopsla, col sep=comma] {tables/data/memory-readpct.csv};
				\addplot[color=brown,mark=o,mark size=3pt] table [x=param, y=si, col sep=comma] {tables/data/memory-readpct.csv};
				\legend{CobraSI w/ GPU, CobraSI w/o GPU,dbcop,\name}
			\end{axis}
		\end{tikzpicture}

	\end{scaletikzpicturetowidth}
	\begin{scaletikzpicturetowidth}{0.23\textwidth}
		\begin{tikzpicture}[scale=\tikzscale]
			\begin{axis}[
				title={(e)},
				xlabel={\#keys},
				ylabel={Memory (MB)},
				]
				\addplot[color=blue,mark=square,mark size=3pt] table [x=param, y=cobra(si), col sep=comma] {tables/data/memory-vars.csv};
				\addplot[color=red,mark=triangle,mark size=3pt] table [x=param, y=cobra(si-nogpu), col sep=comma] {tables/data/memory-vars.csv};
				\addplot[color=black,mark=x,mark size=3pt] table [x=param, y=oopsla, col sep=comma] {tables/data/memory-vars.csv};
				\addplot[color=brown,mark=o,mark size=3pt] table [x=param, y=si, col sep=comma] {tables/data/memory-vars.csv};
				\legend{CobraSI w/ GPU, CobraSI w/o GPU,dbcop,\name}
			\end{axis}
		\end{tikzpicture}
	\end{scaletikzpicturetowidth}
	\hspace{1ex}
	\begin{scaletikzpicturetowidth}{0.23\textwidth}
		\begin{tikzpicture}[scale=\tikzscale]
			\begin{axis}[
				title={(f)},
				x tick style={draw=none},
				ylabel={Memory (MB)},
				ymin=0,
				ymax=600,
				xmin=-0.5,
				xmax=2.5,
				ybar,
				area legend,
				xtick=data,
				xticklabels={uniform,zipfian,hotspot},
				legend pos=north west
				]
				\addplot[color=blue, pattern color=blue, pattern=crosshatch] table [x expr=\coordindex, y=cobra(si), col sep=comma] {tables/data/memory-distrib.csv};
				\addplot[color=red, pattern color=red, pattern=crosshatch dots] table [x expr=\coordindex, y=cobra(si-nogpu), col sep=comma] {tables/data/memory-distrib.csv};
				\addplot[color=black, pattern color=black,pattern=north west lines] table [x expr=\coordindex, y=oopsla, col sep=comma] {tables/data/memory-distrib.csv};
				\addplot[color=brown, pattern color=brown, fill=brown] table [x expr=\coordindex, y=si, col sep=comma] {tables/data/memory-distrib.csv};
				\legend{CobraSI w/ GPU, CobraSI w/o GPU,dbcop,\name}
			\end{axis}
		\end{tikzpicture}
	\end{scaletikzpicturetowidth}	
	\caption{ Comparison on memory overhead with  competing SI checkers under various workloads.  
	}
	\label{comparison-memory}
\end{figure}

\begin{figure}[t]
	\begin{scaletikzpicturetowidth}{0.23\textwidth}
				\begin{tikzpicture}[scale=\tikzscale]
						\begin{axis}[
							title={(a)},
								x tick style={draw=none},
								x tick label style={rotate=45,anchor=east},
								ylabel={Time (s)},
								ymin=0,
								ymax=80,
								ybar,
								bar width=7,
								area legend,
								legend pos=north west,
								xtick=data,
								xticklabels={RUBiS,TPC-C,C-Twitter,GeneralRH,GeneralRW,GeneralWH}
								]
								\addplot[fill=white] table [x expr=\coordindex, y=cobra, col sep=comma] {tables/data/cobra-10k-compare.csv};
								\addplot[pattern=north east lines] table [x expr=\coordindex, y=si, col sep=comma] {tables/data/cobra-10k-compare.csv};
								\legend{Cobra w/ GPU,\name}
							\end{axis}
					\end{tikzpicture}		
			\end{scaletikzpicturetowidth}
\hspace{1ex}
	\begin{scaletikzpicturetowidth}{0.23\textwidth}
	\begin{tikzpicture}[scale=\tikzscale]
		\begin{axis}[
			title={(b)},
			x tick style={draw=none},
			x tick label style={rotate=45,anchor=east},
			ylabel={Memory (MB)},
			ymin=0,
			ybar,
			bar width=7,
			area legend,
			legend pos=north east,
			xtick=data,
			xticklabels={RUBiS,TPC-C,C-Twitter,GeneralRH,GeneralRW,GeneralWH}
			]
			\addplot[fill=white] table [x expr=\coordindex, y=cobra, col sep=comma] {tables/data/memory-cobra-10k-compare.csv};
			\addplot[pattern=north east lines] table [x expr=\coordindex, y=si, col sep=comma] {tables/data/memory-cobra-10k-compare.csv};
			\legend{Cobra w/ GPU,\name}
		\end{axis}
	\end{tikzpicture}
\end{scaletikzpicturetowidth}
	\caption{Comparison on time and memory overhead with Cobra with GPU acceleration under representative workloads.}
\label{comparison-cobra}
\vspace{-2ex}
\end{figure}

\begin{figure}
	\begin{minipage}[b]{0.23\textwidth}
		\centering
		\begin{scaletikzpicturetowidth}{\textwidth}
			\begin{tikzpicture}[scale=\tikzscale]
				\begin{axis}[
					x tick style={draw=none},
					x tick label style={rotate=45,anchor=east},
					ylabel={Time (s)},
					ymin=0,
					ybar stacked,
					bar width=15,
					area legend,
					legend style={at={(0.35, 0.7)}, anchor=west},
					xtick=data,
					xticklabels={RUBiS,TPC-C,C-Twitter,GeneralRH,GeneralRW,GeneralWH}
					]
					\addplot[fill=white] table [x expr=\coordindex, y=solve, col sep=comma] {tables/data/cobra-10k.csv};
					\addplot[pattern=crosshatch dots] table [x expr=\coordindex, y=encode, col sep=comma] {tables/data/cobra-10k.csv};
					\addplot[pattern=north east lines] table [x expr=\coordindex, y=prune, col sep=comma] {tables/data/cobra-10k.csv};
					\addplot[fill=black] table [x expr=\coordindex, y=construct, col sep=comma] {tables/data/cobra-10k.csv};
					\legend{Solving,Encoding,Pruning,Constructing}
				\end{axis}
			\end{tikzpicture}
		\end{scaletikzpicturetowidth}
		\caption{\label{fig:decomposition} Decomposing \name's checking time into  stages.}
	\end{minipage}
  \hspace{1ex}
	\begin{minipage}[b]{0.23\textwidth}
		\centering
		\begin{scaletikzpicturetowidth}{\textwidth}
			\begin{tikzpicture}[scale=\tikzscale]
				\begin{axis}[
				    legend style={nodes={scale=0.8, transform shape}},
					x tick style={draw=none},
					x tick label style={rotate=45,anchor=east},
					scale only axis,
					height=0.6\textwidth,
					width=\textwidth,
					ylabel={Time (s) in log scale },
					ymin=0,
					ymax=50000,
					ybar,
					ymode=log,
					bar width=4,
					area legend,
					legend pos=north west,
					xtick=data,
					xticklabels={RUBiS,TPC-C,C-Twitter,GeneralRH,GeneralRW,GeneralWH}
					]
					\addplot[fill=white] table [x expr=\coordindex, y=si, col sep=comma] {tables/data/optimisation.csv};
					\addplot[pattern=north east lines] table [x expr=\coordindex, y=si(no-pruning), col sep=comma] {tables/data/optimisation.csv};
					\addplot[fill=black] table [x expr=\coordindex, y=si(no-pruning-coalescing), col sep=comma] {tables/data/optimisation.csv};
					\draw[color=red,pattern=north east lines,pattern color=red] (90, 0) rectangle (110, 4.21);
					\draw[color=red,fill=red] (119, 0) rectangle (142, 4.12);
					\legend{\name,\name~w/o P,\name~w/o CP}
				\end{axis}
			\end{tikzpicture}
		\end{scaletikzpicturetowidth}
		\caption{\label{fig:differential-analysis} Diff. analysis of \name{}.   Memory-exhausted runs are colored in red.}
	\end{minipage}
\end{figure}

%% file: sections/discussion.tex

\section{Discussion} \label{section:discussion}




\noindent
\textbf{Fault Injection.} 
We have found SI violations in three production databases without injecting faults, such as network partition and clock drift.  
 Since \name{} is an off-the-shelf checker,  
it is straightforward to  
integrate it into existing testing frameworks with fault injection 
such as Jepsen \cite{jepsen}  and CoFI \cite{cofi}; both have been demonstrated to effectively trigger  bugs in distributed systems.

\vspace{1ex}  
\noindent 
\textbf{Database Schema.}
In our testing of production databases, 
  we used a simple,  yet effective,  database schema adopted by most of black-box checkers~\cite{ConsAD:VLDB2014, Complexity:OOPSLA2019, DBLP:conf/netys/ZennouBBEE19,Cobra:OSDI2020}: a  two-column table storing key-value pairs. 
Extending it to multi-columns or even the column-family data model could be done by: (i) representing each cell in a table as a \emph{compound key},  i.e.,   ``TableName:PrimaryKey:ColumnName'',  
and a single value, i.e.,  the content of the cell \cite{Eiger:NSDI2013,MonkeyDB:OOPSLA2021}; and (ii)
 utilizing the compiler in \cite{MonkeyDB:OOPSLA2021} to rewrite (more complex)
SQL queries to key-value read/write operations.

\vspace{1ex}  
\noindent 
\textbf{Predicates.}
To the best of our knowledge, 
 none of the state-of-the-art black-box checkers \cite{ConsAD:VLDB2014, Complexity:OOPSLA2019, DBLP:conf/netys/ZennouBBEE19,Cobra:OSDI2020, Elle:VLDB2020,MonkeyDB:OOPSLA2021} considers predicates nor can they detect predicate-specific anomalies.  Given the non-predicate violations found by \name{} (as well as dbcop~\cite{Complexity:OOPSLA2019} and Elle~\cite{Elle:VLDB2020}),  we conjecture that more anomalies would arise with predicates.
It is therefore interesting future work to
extend our SI characterization  to represent predicates and to explore optimizations with respect to encoding and pruning.






\vspace{1ex}  
\noindent 
\textbf{Unique Value.} 
	As demonstrated in our experiments, guaranteeing ``unique value'' is a  pragmatic, purely  black-box  technique, and effective in detecting anomalies. 
	When this assumption is broken, 
	the complexity of the checking problem would become higher due to inferring uncertain  
	$\WR$ dependencies (a single read may be related to multiple ``false'' writes). Accordingly, we could add the encoding in \name{} for  unique existence of $\WR$ dependency among all uncertainties prior to SAT solving.

\vspace{1ex}  
\noindent 
\textbf{Optimization for Long Histories.}
 \name{}'s overhead when checking one million transactions with 450 operations per long transaction 
  is manageable for modern hardware. Still, optimizing \name{} for long transaction histories would help to reduce checking overhead, especially for \emph{online}
	transactional processing workloads. 
We could consider  periodically taking snapshots (via read-only transactions) across all sessions in a history using an additional client session. Such snapshots
carry  the summary of   write dependencies thus far, which discards prior transactions in the history. As a result, at any point of time, one only needs to consider a segment of the history consisting of the latest snapshot and its subsequent transactions.


%% file: sections/related-work-nobi.tex
\section{Related Work}
\label{section:related-work}



\textbf{Characterizing Snapshot Isolation.}
Many frameworks and  formalisms have been developed to characterize  SI and its variants. 
Berenson et al. ~\cite{CritiqueANSI:SIGMOD1995} considers \si{} as
a multi-version concurrency control mechanism (described also in Section~\ref{ss:si-informal}).
Adya~\cite{Adya:PhDThesis1999} 
presents the first formal definition
of  \si{} using dependency graphs, which,
  as pointed out by \cite{AnalysingSI:JACM2018},  still relies on low-level implementation choices such as how to order  start and commit events in transactions.
Cerone et al.  \cite{Framework:CONCUR2015} proposes an axiomatic framework to declaratively define \si{}  with
the dual notions of visibility (what transactions can observe) and arbitration (the order of installed versions/values).
The follow-up work~\cite{AnalysingSI:JACM2018}
 characterizes \si{}  solely in terms of Adya's  dependency graphs,
 requiring no additional information about transactions.
 Crooks et al.  \cite{ClientCentric:PODC2017}  introduces an alternative implementation-agnostic formalization of \si{} and its variants based on
client-observed values read/written.

 Driven by black-box testing of SI,
 we base
our GP-based characterization  on Cerone and Gotsman's formal specification \cite{AnalysingSI:JACM2018}.  In particular,   our new characterization:

\vspace{1ex}
\noindent
 (i) targets the prevalent \emph{strong session} variant of SI \cite{AnalysingSI:JACM2018, LazyReplSI:VLDB2006}, where \emph{sessions},  advocated by Terry et al.  \cite{TerryDPSTW94},  have been adopted by many production databases in practice (e.g., DGraph \cite{dgraph}, Galera \cite{maria-galera},  and CockroachDB \cite{cockroach});

\vspace{1ex}
\noindent
(ii)
does not rely on implementation details such as concurrency control mechanism as in \cite{CritiqueANSI:SIGMOD1995} and  timestamps as in  \cite{Adya:PhDThesis1999},  and the operational semantics of the underlying database as in~\cite{CentralisedSemantics:ECOOP2020},  which are usually invisible to the outsiders; and

\vspace{1ex}
\noindent
(iii)
naturally models uncertain dependencies inherent to black-box testing
using generalized constraints (Section \ref{section:polygraph-si}) and  enables the acceleration of SMT solving by compacting constraints (Section~\ref{ss:efficiency}).

%
%
%
%
%
%



\vspace{1ex}
Regarding the comparison with  \cite{ClientCentric:PODC2017},
despite its promising characterization of \si{} suitable for black-box testing,  we are
unaware of any checking algorithm  based on it.  A straightforward
(suboptimal) implementation  would require  enumerating  all permutations of the transactions in a history,  e.g.,  10k transactions in our experiment would require checking 10k-factorial permutations.


\vspace{1ex}
\noindent
\textbf{Dynamic Checking of SI.}
This technique
determines whether a collected history from dynamically executing a database satisfies SI.
We are unaware of any black-box SI checker that satisfies SIEGE+.

  dbcop   \cite{Complexity:OOPSLA2019} is the most efficient black-box SI checker to date.
The underlying checking algorithm runs in time
 $O(n^{c})$,  with $n$ and $c$ the number of transactions and clients involved in a  single history,  respectively.  The authors
devise both a polynomial-time algorithm for checking serializabilty (also with a fixed number of client sessions) and a polynomial-time algorithm for reducing checking SI  to checking serializabilty.  However,  as demonstrated in Section \ref{ss:efficiency},  dbcop is practically not as  efficient as our \name{} tool under various workloads.
Moreover,     dbcop is incomplete as it does not check
non-cycle anomalies such as \emph{aborted reads} and \emph{intermediate reads} (Section~\ref{ss:completing-si-checking}).
Finally, dbcop  provides no details upon a violation; only a ``false'' answer is returned.

Elle~\cite{Elle:VLDB2020}
is a state-of-the-art
checker  for  a variety of isolation levels, including strong session SI,\footnote{
  Despite the claim to support checking strong session SI,
  we have confirmed with the developer that Elle does not fulfill this functionality in its latest release \cite{elle-bug}.
  The developer has fixed the issue by adding the checking of ``g-nonadjacent-process''.}
  which is part of the Jepsen \cite{jepsen} testing framework.
Elle requires specific data models like lists  in workloads
to infer the $\WW$ dependencies and specific APIs to perform  list-specific  operations such as ``append''.
In contrast, \name{} is compatible with
general and production workloads
and uses standard key-value and SQL  APIs.
 Elle builds upon Adya's formalization of SI~\cite{Adya:PhDThesis1999}, thus relying on
the start and commit timestamps of transactions for completeness.
Such information may not always be available, e.g., MongoDB~\cite{MongoDB} and TiDB~\cite{TiDB} have no timestamps in their logs for read-only transactions.
Nonetheless,  
the underlying SI characterization for \name{} does not rely on any implementation details.
Finally, Elle's actual implementation is incomplete\footnote{
This was ``unsound'' in the previous version
(also in \cite[Section 7]{PolySI:VLDB2023}).
As clarified by the developer, it is actually ``incomplete''.}
for efficiency reasons and there are anomalies it cannot detect.
We have confirmed this with the developer~\cite{elle-bug}.


ConsAD \cite{ConsAD:VLDB2014}  is a checker tailored to application servers as opposed to black-box databases in our setting.
 Its SI checking algorithm is also based on dependency graphs.
 To
 determine the $\WW$ dependencies,
ConsAD enforces the commit order of
 update transactions using, e.g.,  artificial SQL queries, to acquire exclusive locks on the database records,  resulting in additional overhead \cite{ConsAD:VLDB2014,  RushMon:SIGMOD2018}.
Moreover,
ConsAD is 
incapable of  detecting non-cycle anomalies.





CAT \cite{Maude:Liu2019} is
a dynamic  white-box checker for SI (and several other isolation levels).
The current release is restricted to  distributed databases implemented in the Maude formal language~\cite{AllAboutMaude:Book2007}.
CAT must
capture the internal transaction information, e.g., start/commit times, during a system run.


%% file: sections/conclusion.tex

\section{Conclusion}  \label{section:conclusion}

We have presented the design of \name{},
along with a novel characterization of SI using generalized polygraphs.
We have established the soundness and completeness of our new characterization and \name{}'s checking algorithm.
Moreover, we have demonstrated \name{}'s effectiveness by reproducing all of 2477 known  anomalies
and by finding new violations in three popular production  databases,
its efficiency by experimentally showing that it outperforms the state-of-the-art tools and can scale up to large-sized workloads,
and its generality, operating over a wide range of workloads and databases of different kinds.
Finally, we have leveraged \name{}'s interpretation algorithm to identify the causes of the violations.

\name{} is the first black-box SI checker that satisfies the SIEGE+ principle. 
The obvious next step is to 
 apply SMT solving to build SIEGE+ black-box checkers for other data consistency properties such as 
 transactional causal consistency \cite{Eiger:NSDI2013,DBLP:journals/pvldb/DidonaGWZ18} 
 and the recently proposed regular sequential consistency \cite{DBLP:conf/sosp/Helt0LL21}. 
 Moreover,  we will pursue
 the  research directions discussed in Section \ref{section:discussion}.

%% file: sections/ack.tex

\section*{Acknowledgments} \label{section:ack}

We would like to thank the anonymous reviewers
for their helpful feedback.
This work was supported by the CCF-Tencent Open Fund
(Tencent RAGR20200201).

%% file: sections/appendix.tex

\input{sections/appendix-polysi-alg}
\input{sections/appendix-proof}
\input{sections/appendix-interpretation}
\input{sections/appendix-yugabytedb}
\input{tables/memory}
\input{tables/large}
\input{tables/random-large}
\input{sections/appendix-minimal}
\input{tables/runtime-elle}

%% file: sections/appendix-polysi-alg.tex

\section{The Checking Algorithm for SI}
\label{section:appendix-polysi}

\input{algs/checksi-full}

The full version of the checking algorithm is given in Algorithm~\ref{alg:checksi-full}.
In the following, we reference pseudocode lines using the format algorithm\#:line\#.

%% file: algs/checksi-full.tex

\begin{algorithm*}[t]
  \footnotesize
  \caption{The \name{} algorithm for checking SI (the full version)}
  \label{alg:checksi-full}
  \begin{varwidth}[t]{0.49\textwidth}
  \begin{algorithmic}[1]

    \Procedure{\checksi}{$\H$}
      \label{line:full-func-checksi}
      \If{$\H \not\models \intaxiom \lor \abortedreads \lor \intermediatereads$}
        \label{line:full-checksi-intaxiom}
        \State \Return \false
      \EndIf

      \hStatex
      \State \Call{\createknowngraph}{$\H$}
        \label{line:full-checksi-call-createknowngraph}
      \State \Call{\generateconstraints}{$\H$}
        \label{line:full-checksi-call-generateconstraints}
      \If{$\lnot \Call{\pruneconstraints}{\null}$}
        \label{line:full-checksi-call-pruneconstraints}
        \State \Return \false
          \label{line:full-checksi-return-false}
      \EndIf
      \State \Call{\encodeconstraints}{\null}
        \label{line:full-checksi-call-encodeconstraints}
      \State \Return \Call{\solveconstraints}{\null}
        \label{line:full-checksi-call-solveconstraints}
    \EndProcedure

    \Statex
    \Procedure{\createknowngraph}{$\H$}
      \label{line:full-func-createknowngraph}
      \ForAll{$T, S \in \T$ such that $T \rel{\SO} S$}
        \label{line:full-createknowngraph-so-relation}
        \State $\edges_{\g} \gets \edges_{\g} \cup (T, S, \SO)$
          \label{line:full-createknowngraph-so-edges}
      \EndFor
      \ForAll{$T, S \in \T$ such that $T \rel{\WR} S$}
        \label{line:full-createknowngraph-wr-relation}
        \State $\edges_{\g} \gets \edges_{\g} \cup (T, S, \WR)$
          \label{line:full-createknowngraph-wr-edges}
      \EndFor
    \EndProcedure

    \Statex
    \Procedure{\generateconstraints}{$\H$}
      \label{line:full-func-generateconstraints}
      \ForAll{$\keyxvar \in \Key$}
        \label{line:full-generateconstraints-forall-key}
        \ForAll{$T, S \in \WriteTx_{\keyxvar}$ such that $T \neq S$}
          \label{line:full-generateconstraints-forall-ts}
          \State $\eithervar \gets \set{(T, S, \WW)} \cup \bigcup\limits_{T' \in \WR(\keyxvar)(T)} \!\!\set{(T', S, \RW)}$
            \label{line:full-generateconstraints-either}
          \State $\orvar \gets \set{(S, T, \WW)} \cup \bigcup\limits_{S' \in \WR(\keyxvar)(S)} \!\!\set{(S', T, \RW)}$
            \label{line:full-generateconstraints-or}
          \State $\cons_{\g} \gets \cons_{\g} \cup \set{\tuple{\eithervar, \orvar}}$
            \label{line:full-generateconstraints-either-or}
        \EndFor
      \EndFor
    \EndProcedure

    \Statex
    \Procedure{\encodeconstraints}{\null}
      \label{line:full-func-encodeconstraints}
      \ForAll{$\vvar_{i}, \vvar_{j} \in \vertex_{\g}$ such that $i \neq j$}
        \label{line:full-encodeconstraints-forall-vertex-pair}
        \State Create two Boolean variables $\BV_{i, j}$ and $\BV^{\inducedgraph}_{i, j}$
          \label{line:full-encodeconstraints-create-bv-ij}
        \State $\BV \gets \BV \cup \set{\BV_{i, j}, \BV^{\inducedgraph}_{i, j}}$
          \label{line:full-encodeconstraints-bv-add-bvij}
      \EndFor

      \ForAll{$(\vvar_{i}, \vvar_{j}) \in \edges_{\g}$} \Comment{encode the known graph $\g$}
        \label{line:full-encodeconstraints-forall-edges}
        \State $\Clause \gets \Clause \cup \set{\BV_{i, j} = \True}$
          \label{line:full-encodeconstraints-set-bvij-true}
      \EndFor

      \ForAll{$\tuple{\eithervar, \orvar} \in \cons_{\g}$}  \Comment{encode the constraints $\g$}
        \label{line:full-encodeconstraints-forall-cons}
        \State $\Clause \gets \Clause \;\cup\; \bbset{\bigl(
            \bigwedge\limits_{(\vvar_{i}, \vvar_{j}, \_) \in \eithervar} \BV_{i, j}
            \land \bigwedge\limits_{(\vvar_{i}, \vvar_{j}, \_) \in \orvar} \lnot\BV_{i, j}\bigr) \;\lor\;
            \bigl(
              \bigwedge\limits_{(\vvar_{i}, \vvar_{j}, \_) \in \orvar} \BV_{i, j}
              \land \bigwedge\limits_{(\vvar_{i}, \vvar_{j}, \_) \in \eithervar} \lnot\BV_{i, j}\bigr)}$
          \label{line:full-encodeconstraints-either-or}
      \EndFor

      \hStatex
      \State $\graphA \gets \g|_{\SO_{\g} \cup \WR_{\g} \cup \WW_{\g}}$
        \label{line:full-encodeconstraints-A}
      \State $\edges_{\graphA} \gets \edges_{\graphA} \cup \set{(\_, \_, \WW) \in \eithervar \cup \orvar \mid \tuple{\eithervar, \orvar} \in \cons_{\g}}$
        \label{line:full-encodeconstraints-add-ww-to-A}
      \State $\graphB \gets \g|_{\RW_{\g}}$
        \label{line:full-encodeconstraints-B}
      \State $\edges_{\graphB} \gets \edges_{\graphB} \cup \set{(\_, \_, \RW) \in \eithervar \cup \orvar \mid \tuple{\eithervar, \orvar} \in \cons_{\g}}$
        \label{line:full-encodeconstraints-add-rw-to-B}
      \State $\Clause \gets \Clause \cup \bbset{\BV^{\inducedgraph}_{i, j} =
        \big(\BV_{i, j} \land (\vvar_{i}, \vvar_{j}, \_) \in \edges_{\graphA}\big) \lor
        \big(\bigvee\limits_{\substack{(\vvar_{i}, \vvar_{k}, \_) \in \edges_{\graphA} \\
          (\vvar_{k}, \vvar_{j}, \_) \in \edges_{\graphB}}} \hspace{-1em} \BV_{i, k} \land \BV_{k, j}\big)
          \bigm| \vvar_{i}, \vvar_{j} \in \vertex_{\g}}$
        \Comment{encode the induced \si{} graph $\inducedgraph$ of $\g$}
        \label{line:full-encodeconstraints-add-clauses-for-I}
    \EndProcedure
    \algstore{checksi}
  \end{algorithmic}
  \end{varwidth}\qquad
  \begin{varwidth}[t]{0.49\textwidth}
  \begin{algorithmic}[1]
    \algrestore{checksi}
    \Procedure{\pruneconstraints}{\null}
      \label{line:full-func-pruneconstraints}
      \Repeat
        \State $\graphA \gets \g|_{\SO_{\g} \cup \WR_{\g} \cup \WW_{\g}}$
          \label{line:full-pruneconstraints-A}
        \State $\graphB \gets \g|_{\RW_{\g}}$
          \label{line:full-pruneconstraints-B}
        \State $\knowninducedgraph \gets \graphA \cup (\graphA \comp \graphB)$
          \label{line:full-pruneconstraints-C}
        \State $\reachabilityvar \gets \Call{\reachability}{\knowninducedgraph}$
          \label{line:full-pruneconstraints-reachability}

        \hStatex
        \ForAll{$\consvar \gets \tuple{\eithervar, \orvar} \in \cons_{\g}$}
          \label{line:full-pruneconstraints-forall-cons}
          \ForAll{$(\fromvar, \tovar, \typevar) \in \eithervar$}
            \Comment{for the ``$\eithervar$'' possibility}
            \label{line:full-pruneconstraints-forall-edge-in-either}
            \If{$\typevar = \WW$}
              \label{line:full-pruneconstraints-ww-type}
              \If{$(\tovar, \fromvar) \in \reachabilityvar$}
                \label{line:full-pruneconstraints-ww-reachability}
                \State $\cons_{\g} \gets \cons_{\g} \setminus \set{\consvar}$
                  \label{line:full-pruneconstraints-ww-cons}
                \State $\edges_{\g} \gets \edges_{\g} \cup \orvar$
                  \label{line:full-pruneconstraints-ww-edges}
                \State {\bf break} the ``{\bf for all} $(\fromvar, \tovar, \typevar) \in \eithervar$'' loop
                  \label{line:full-pruneconstraints-ww-break}
              \EndIf
            \Else \Comment{$\typevar = \RW$}
              \ForAll{$\precvar \in \vertex_{\graphA}$ such that $(\precvar, \fromvar, \_) \in \edges_{\graphA}$}
                \label{line:full-pruneconstraints-rw-forall-pre-vertex}
                \If{$(\tovar, \precvar) \in \reachabilityvar$}
                  \label{line:full-pruneconstraints-rw-reachability}
                  \State $\cons_{\g} \gets \cons_{\g} \setminus \set{\consvar}$
                    \label{line:full-pruneconstraints-rw-cons}
                  \State $\edges_{\g} \gets \edges_{\g} \cup \orvar$
                    \label{line:full-pruneconstraints-rw-edges}
                  \State {\bf break} the ``{\bf for all} $(\fromvar, \tovar, \typevar) \in \eithervar$'' loop
                    \label{line:full-pruneconstraints-rw-break}
                \EndIf
              \EndFor
            \EndIf
          \EndFor

          \ForAll{$(\fromvar, \tovar, \typevar) \in \orvar$}
            \Comment{for the ``$\orvar$'' possibility}
            \label{line:full-pruneconstraints-forall-edge-in-or}
            \If{$\typevar = \WW$}
              \label{line:full-pruneconstraints-ww-type-in-or}
              \If{$(\tovar, \fromvar) \in \reachabilityvar$}
                \label{line:full-pruneconstraints-ww-reachability-in-or}
                \If{$\consvar \notin \cons_{\g}$}
                  \Comment{neither ``$\eithervar$'' nor ``$\orvar$'' is possible}
                  \label{line:full-pruneconstraints-ww-consvar-not-in-cons}
                  \State \Return \False
                    \label{line:full-pruneconstraints-ww-return-false}
                \EndIf
                \State $\cons_{\g} \gets \cons_{\g} \setminus \set{\consvar}$
                  \label{line:full-pruneconstraints-ww-cons-in-or}
                \State $\edges_{\g} \gets \edges_{\g} \cup \eithervar$
                  \label{line:full-pruneconstraints-ww-edges-in-or}
                \State {\bf break} the ``{\bf for all} $(\fromvar, \tovar, \typevar) \in \orvar$'' loop
                  \label{line:full-pruneconstraints-ww-break-in-or}
              \EndIf
            \Else \Comment{$\typevar = \RW$}
              \ForAll{$\precvar \in \vertex_{\graphA}$ such that $(\precvar, \fromvar, \_) \in \edges_{\graphA}$}
                \label{line:full-pruneconstraints-rw-forall-pre-vertex-in-or}
                \If{$(\tovar, \precvar) \in \reachabilityvar$}
                  \label{line:full-pruneconstraints-rw-reachability-in-or}
                  \If{$\consvar \notin \cons_{\g}$}
                    \Comment{neither ``$\eithervar$'' nor ``$\orvar$'' is possible}
                    \label{line:full-pruneconstraints-rw-consvar-not-in-cons}
                    \State \Return \False
                      \label{line:full-pruneconstraints-rw-return-false}
                  \EndIf
                  \State $\cons_{\g} \gets \cons_{\g} \setminus \set{\consvar}$
                    \label{line:full-pruneconstraints-rw-cons-in-or}
                  \State $\edges_{\g} \gets \edges_{\g} \cup \eithervar$
                    \label{line:full-pruneconstraints-rw-edges-in-or}
                  \State {\bf break} the ``{\bf for all} $(\fromvar, \tovar, \typevar) \in \orvar$'' loop
                    \label{line:full-pruneconstraints-rw-break-in-or}
                \EndIf
              \EndFor
            \EndIf
          \EndFor
        \EndFor
      \Until{$\cons_{\g}$ remains unchanged}
        \label{line:full-pruneconstraints-until}
      \State \Return \True
        \label{line:full-pruneconstraints-return-true}
    \EndProcedure

    \Statex
    \Procedure{\solveconstraints}{\null}
      \label{line:full-func-solveconstraints}
      \State $\solver \gets \Call{\instance}{\BV, \Clause}$
        \label{line:full-solveconstraints-solver}
      \State \Return $\Call{\solve}{\solver, \inducedgraph \text{ is acyclic}}$
        \label{line:full-solveconstraints-call-solve}
    \EndProcedure
  \end{algorithmic}
  \end{varwidth}
  \normalsize
\end{algorithm*}

%% file: sections/appendix-proof.tex

\section{Proofs} \label{section:appendix-proofs}

\subsection{Proof of Theorem~\ref{thm:polygraph-si}} \label{ss:proof-polygraph-si}

\begin{proof} \label{proof:polygraph-si}
  The proof proceeds in two directions.

  (``$\implies$'') Suppose that $\H$ satisfies \si.
  By Theorem~\ref{thm:depgraph-si}, $\H$ satisfies $\intaxiom$
  and there exist $\WR$, $\WW$, and $\RW$ relations with which $\H$ can be extended to a dependency graph $G$
  such that $(\SO_{\g} \cup \WR_{\g} \cup \WW_{\g}) \comp \RW_{\g}?$ is acyclic.
  We show that the generalized polygraph $G'$ of $\H$ is \si-acyclic
  by constructing a compatible graph $G''$ with $G'$ such that $G''|_{\inducedrule}$
  is acyclic when the edge types are ignored:
  Consider a constraint $\tuple{\eithervar, \orvar}$ in $G'$ and any $\WW$ edge $T \rel{\WW} S$ in $G$.
  If $(T, S, \WW) \in \eithervar$, add all the edges in $\eithervar$ into $G''$.
  Otherwise, add all the edges in $\orvar$ into $G''$.

  (``$\impliedby$'') Suppose that $\H \models \intaxiom$ and the generalized polygraph $G'$ of $\H$ is \si-acyclic.
  By Definition~\ref{def:si-acyclicity}, there exists a compatible graph $G''$ with $G'$
  such that $G''|_{\inducedrule}$ is acyclic when the edge types are ignored.
  We show that $\H$ satisfies \si{} by constructing suitable $\WR$, $\WW$, and $\RW$ relations
  with which $\H$ can be extended to a dependency graph $G$
  such that $(\SO_{\g} \cup \WR_{\g} \cup \WW_{\g}) \comp \RW_{\g}?$ is acyclic:
  We simply take $G$ to be $G''$ by defining, e.g., $\WW = \set{(a, b) \mid (a, b, \WW) \in E_{G''}}$.
\end{proof}

\subsection{Proof of Theorem~\ref{thm:correctness-of-pruneconstraints}} \label{ss:correctness-of-pruneconstraints}

\begin{proof} \label{proof:correctness-of-pruneconstraints}
  We prove (1) by induction on the number of iterations of pruning.
  Consider an arbitrary iteration of pruning $\pruning$
  (lines~\code{\ref{alg:checksi-full}}{\ref{line:full-pruneconstraints-A}}--\code{\ref{alg:checksi-full}}{\ref{line:full-pruneconstraints-rw-break-in-or}})
  and denote the generalized polygraphs just before and after $\pruning$ by $G_{1}$ and $G_{2}$, respectively.
  We should show that $G_{1}$ is \si-acyclic if and only if
  $\pruning$ does not return $\False$ from line~\code{\ref{alg:checksi-full}}{\ref{line:full-pruneconstraints-ww-return-false}}
  or line~\code{\ref{alg:checksi-full}}{\ref{line:full-pruneconstraints-rw-return-false}} and $G_{2}$ is \si-acyclic.
  The following proof proceeds in two directions.

  (``$\implies$'') Suppose that $G_{1}$ is \si-acyclic. We then proceed by contradiction.
  \begin{itemize}
    \item Suppose that $\pruning$ returns $\False$ from
      line~\code{\ref{alg:checksi-full}}{\ref{line:full-pruneconstraints-ww-return-false}} or line~\code{\ref{alg:checksi-full}}{\ref{line:full-pruneconstraints-rw-return-false}}.
      This happens when some constraint in $G_{1}$ constructed based on two write transactions, say $T$ and $S$,
      on the same key cannot be resolved appropriately:
      every compatible graph with the induced \si{} graph of $G_{1}$ contains a cycle
      without adjacent $\RW$ edges, no matter whether $T \rel{\WW} S$ or $S \rel{\WW} T$ is in it.
      That is, $G_{1}$ is not \si-acyclic. Contradiction.
    \item Suppose that $G_{2}$ is not \si-acyclic.
      $G_{1}$ is not \si-acyclic because
      the set of constraints in $G_{2}$ is a subset of that in $G_{1}$
      and $\pruning$ prunes a constraint only if one of its two possibilities cannot happen.
      Contradiction.
  \end{itemize}

  (``$\impliedby$'') Suppose that $\pruning$ does not return $\False$ from line~\code{\ref{alg:checksi-full}}{\ref{line:full-pruneconstraints-ww-return-false}}
  or line~\code{\ref{alg:checksi-full}}{\ref{line:full-pruneconstraints-rw-return-false}} and $G_{2}$ is \si-acyclic.
  Therefore, there exists an acyclic compatible graph with the induced \si{} graph of $G_{2}$.
  This is also an acyclic compatible graph with the induced \si{} graph of $G_{1}$.
  Hence, $G_{1}$ is \si-acyclic.

  The second part of the theorem holds because any compatible graph with the induced \si{} graph of $G_{p}$
  is also a compatible graph with the induced \si{} graph of $G$.
\end{proof}

%% file: sections/appendix-interpretation.tex

\section{The Interpretation Algorithm} \label{section:appendix-interpretation}

In this section, we describe the interpretation algorithm, called \interpret{},
that helps locate the causes of violations found by \name; see Algorithm~\ref{alg:interpret}.
The algorithm takes as input the undesired cycles constructed from the log generated by MonoSAT,
and outputs a dependency graph demonstrating the cause of the violation.

\input{algs/interpret}

It is often difficult to locate the cause of a violation based solely on the original cycles
because the cycles may miss some crucial information such as
the core participating transactions and the dependencies between transactions.
Therefore, \interpret{} first restores such information from the generalized polygraph of the history
to reproduce the whole violating scenario
(line~\code{\ref{alg:interpret}}{\ref{line:interpret-call-restore}}).
This may bring uncertain dependencies to the resulting polygraph, which are then resolved via pruning
(line~\code{\ref{alg:interpret}}{\ref{line:interpret-call-resolve}}).
Finally, \interpret{} finalizes the violating scenario by removing any remaining uncertain dependencies,
as they are the ``effect'' of the violation instead of the ``cause''
(line~\code{\ref{alg:interpret}}{\ref{line:interpret-call-finalize}}).

\input{sections/appendix-restore}
\input{sections/appendix-resolve}
\input{sections/appendix-finalize}

%% file: algs/interpret.tex

\begin{algorithm*}[t]
  \footnotesize
	\caption{The interpretation algorithm.}
	\label{alg:interpret}
    \begin{varwidth}[t]{0.49\textwidth}
	\begin{algorithmic}[1]
		\Statex $\H = (\T, \SO)$: the original history
		\Statex $G = (\vertex\_1, \edges\_1, \cons\_1)$: the polygraph
		\Statex $\cycle = (\vertex\_2, \edges\_2)$: The cycle found by \name

		\Statex
        \Procedure{\interpret}{$\H, G, \cycle$}
          \label{procedure:interpret}
            \State $\recoveredgraph$ $\gets$ $\Call{\Restore}{\H, G, \cycle}$
                \label{line:interpret-call-restore}
            \State $\taggedgraph$ $\gets$ $\Call{\Resolve}{\recoveredgraph,\H, G, \cycle}$
                \label{line:interpret-call-resolve}
            \State $\finalizedgraph$ $\gets$ $\Call{\Finalize}{\taggedgraph}$
                \label{line:interpret-call-finalize}
            \State \Return $\recoveredgraph,\taggedgraph,\finalizedgraph$
              \label{line:interpret-return}
        \EndProcedure

		\Statex
        \Procedure{\Restore}{$\H, G, \cycle$} \label{procedure:restore}
            \State $\recoveredgraph \gets \Call{\Findacs}{\cycle}$
            
                \ForAll{$(u\rel{\RW}v) \in \edges\_2$}
                \label{recover-rw-edges}
                    \ForAll{$c \in \cons\_1$}
                        \If{$((u\rel{\RW}v) \in \ct) \wedge (\exists w, (w\rel{\WW}v) \in \ct)$}
                            \State
                            $\recoveredgraph$ $\gets$ $\recoveredgraph
                                \cup \{w\rel{\WW}v\} \cup \{w\rel{\WR}u\}$
                        \EndIf
                    \EndFor
                \EndFor
                \label{recover-rw-edges-finish}

            \State
            \Return $\recoveredgraph$
        \EndProcedure     
        
        \Statex
        \Procedure{\Findacs}{$\acs,G$} \label{procedure:restore acs}
            \State
                $\graphsize$ $\gets$ number of edges in $\acs$ 
            \State
                $\minimalExtendDeps$ $\gets$ $+\infty$
            \State
                $\minimalAcs$ $\gets$ empty
            \ForAll{$\dependency \in (\edges\_2 \cap \ct.either)$, where $\ct \in \cons\_1$}
                \label{recover-edge-in-cons-1}
                    \If{($\exists \dependency' \in \ct.or, \dependency' \in \recoveredgraph$)}
                        continue
                    \EndIf
                    \label{opposite-dp-has-contained-1}

                    \ForAll{$\dependency' \in \ct.or$}
                    \label{find-cycle-contain-opposite-dp}
                        \ForAll{( cycle $(\vertex',\edges')\subseteq G$) $\wedge$  ($\dependency' \in \edges'$ )}
                            \State
                                $\testgraph$ $\gets$ $\recoveredgraph \cup (\vertex',\edges')$
                            \State
                                $(\extendDeps, \extendDeps)$ $\gets$ $\Call{\Findacs}{\testgraph,G}$
                                \label{find-cycle-recursively1}
                            \If{$\extendDeps < \minimalExtendDeps$}
                                \State
                                    $\minimalExtendDeps$ $\gets$ $\extendDeps$
                                \State
                                    $\minimalAcs$ $\gets$ $\extendgraph$
                            \EndIf
                        \EndFor
                    \EndFor
                \EndFor
                \label{recover-cons-finish-1}
                
                \ForAll{$\dependency \in (\edges\_2 \cap \ct.or)$, where $\ct \in \cons\_1$}
                \label{recover-edge-in-cons-2}
                    \If{($\exists \dependency' \in \ct.either, \dependency' \in \recoveredgraph$)}
                        continue
                    \EndIf
                    \label{opposite-dp-has-contained-2}

                    \ForAll{$\dependency' \in \ct.either$}
                    \label{find-cycle-contain-opposite-dp}
                        \ForAll{( cycle $(\vertex',\edges')\subseteq G$) $\wedge$  ($\dependency' \in \edges'$ )}
                            \State
                                $\testgraph$ $\gets$ $\recoveredgraph \cup (\vertex',\edges')$
                            \State
                                $(\extendDeps, \extendDeps)$ $\gets$ $\Call{\Findacs}{\testgraph,G}$
                                \label{find-cycle-recursively2}
                            \If{$\extendDeps < \minimalExtendDeps$}
                                \State
                                    $\minimalExtendDeps$ $\gets$ $\extendDeps$
                                \State
                                    $\minimalAcs$ $\gets$ $\extendgraph$
                            \EndIf
                        \EndFor
                    \EndFor
                \EndFor
                \label{recover-cons-finish-2}
            \Return $(\minimalExtendDeps, \minimalAcs+\graphsize)$
        \EndProcedure
        \algstore{interpret}
      \end{algorithmic}
      \end{varwidth}\qquad
      \begin{varwidth}[t]{0.49\textwidth}
        \begin{algorithmic}[1]
        \algrestore{interpret}
        \Procedure{\Resolve}{$\recoveredgraph,\H, G, \cycle$} \label{procedure:resolve}
            \State $\taggedgraph \gets \recoveredgraph$

            \ForAll{$\dependency \in \taggedgraph$}
            \label{set-known-certain}
                \If{$\dependency \in \edges\_1$}
                    \State
                        $\dependency.tag$ $\gets$ `certain'
                \Else
                    \State
                        $\dependency.tag$ $\gets$ `uncertain'
                \EndIf
            \EndFor

            \While{$\taggedgraph$ was changed in the last loop}
            \label{resolve-part}
                \ForAll{$\dependency \in \taggedgraph$}
                \label{tag-cons}
                    \If{($\dependency.tag$ = `uncertain') $\wedge$ ($\dependency$ in a cycle $(\vertex',\edges')$)}
                        \If{$\forall \dependency' \in \edges' \wedge \dependency'\neq \dependency$, $\dependency'.tag$ = certain }
                            \State
                                $\dependency.tag$ $\gets$ `uncertain'
                            \ForAll{$\dependency_{opposite}$, where $\{ \dependency / \dependency_{opposite} \} \in \cons\_1$}
                                \State
                                    $\dependency_{opposite}.tag$ $\gets$ `certain'
                            \EndFor
                        \EndIf
                    \EndIf
                \EndFor
            \EndWhile
            \label{resolve-part-finish}

            \Return $\taggedgraph$
        \EndProcedure

        \Statex
        \Procedure{\Finalize}{$\taggedgraph$} \label{procedure:finalize}
            \State
                $\finalizedgraph$ $\gets$ $\taggedgraph$
            \ForAll{$\dependency \in \finalizedgraph$}
            \label{delete-uncertain-dps}
                \If{$\dependency.tag = $ `uncertain'}
                    \State
                    $\finalizedgraph$ $\gets$ $\finalizedgraph \setminus \{ \dependency \}$
                \EndIf
            \EndFor
            \State
                \Return $\finalizedgraph$
        \EndProcedure
	\end{algorithmic}
  \end{varwidth}
  \normalsize
\end{algorithm*}

%% file: sections/appendix-restore.tex

\paragraph*{Restore Transactions and Dependencies.}
The procedure \Restore{} restores dependencies and the transactions involved in these dependencies in two steps
(line~\code{\ref{alg:interpret}}{\ref{procedure:restore}}).
First, it checks each $\RW$ dependency in the input undesired cycles
and restores its associated $\WW$ and $\WR$ dependencies if they are missing
(lines~\code{\ref{alg:interpret}}{\ref{recover-rw-edges}} -- \code{\ref{alg:interpret}}{\ref{recover-rw-edges-finish}}).
Then, for each $\WW$ dependency in an undesired cycle, 
it restores the other direction of this $\WW$ dependency (if missing)
which may be involved in another undesired cycle
(line~\code{\ref{alg:interpret}}{\ref{recover-edge-in-cons-1}} 
- line~\code{\ref{alg:interpret}}{\ref{recover-cons-finish-2}}).
Now it has recovered the minimal counter example.
To help understand the violation more clearly, it checks each $\RW$ dependency
in the input undesired cycles and restore its associated $\WW$ and
$\WR$ dependencies if they are missing ((line~\code{\ref{alg:interpret}}{\ref{recover-rw-edges}} -- \code{\ref{alg:interpret}}{\ref{recover-rw-edges-finish}}).

%% file: sections/appendix-resolve.tex

\paragraph*{Resolve Uncertain Dependencies.}
The restored violating scenario may contain $\WW$ and $\RW$ dependencies
that are still uncertain at this moment.
For an uncertain dependency to be part of the cause of a violation, it must be resolved.
The procedure \Resolve{} resolves as many uncertain dependencies as possible
using the same idea of pruning (Section~\ref{ss:prune-constraints}).
Specifically, if an uncertain dependency from $\eithervar$ (resp. $\orvar$)
in a constraint would create an undesired cycle with other certain dependencies,
it, along with its associated dependencies in $\eithervar$ (resp. $\orvar$) is removed
and the dependencies in $\orvar$ (resp. $\eithervar$) become certain 
(lines~\code{\ref{alg:interpret}}{\ref{resolve-part}} 
- ~\code{\ref{alg:interpret}}{\ref{resolve-part-finish}}).

%% file: sections/appendix-finalize.tex

\paragraph*{Finalize the Violating Scenario.}

The procedure \Finalize{} finalizes the violating scenario
by removing all the remaining uncertain dependencies,
as they are the ``effect'' of the violation instead of the ``cause''
(line~\code{\ref{alg:interpret}}{\ref{procedure:finalize}}).

%% file: sections/appendix-yugabytedb.tex

\section{Causality Violations Found}
\label{section:appendix-causality}

\subsection{A Causality Violation Found in Dgraph}

\begin{figure*}[t]
	\centering
	\begin{subfigure}[b]{0.22\textwidth}
		\centering
		\includegraphics[width = \textwidth]{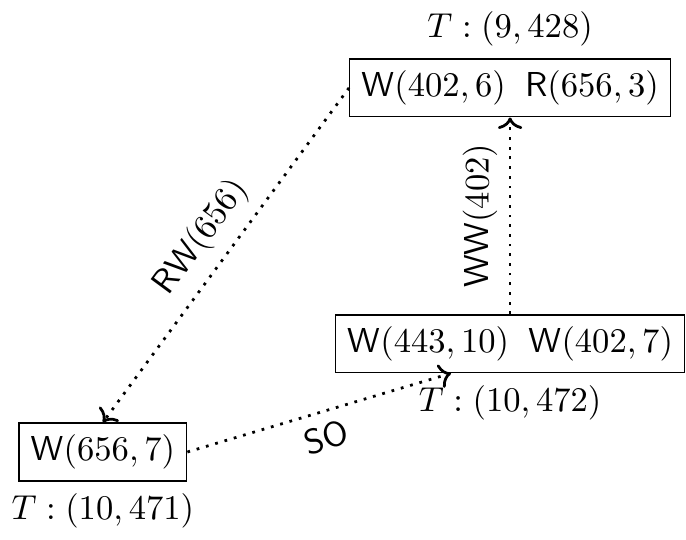}
		\caption{Original output}
	\end{subfigure}\hspace{22ex}
	\begin{subfigure}[b]{0.38\textwidth}
		\centering
		\includegraphics[width = 1.25\textwidth]{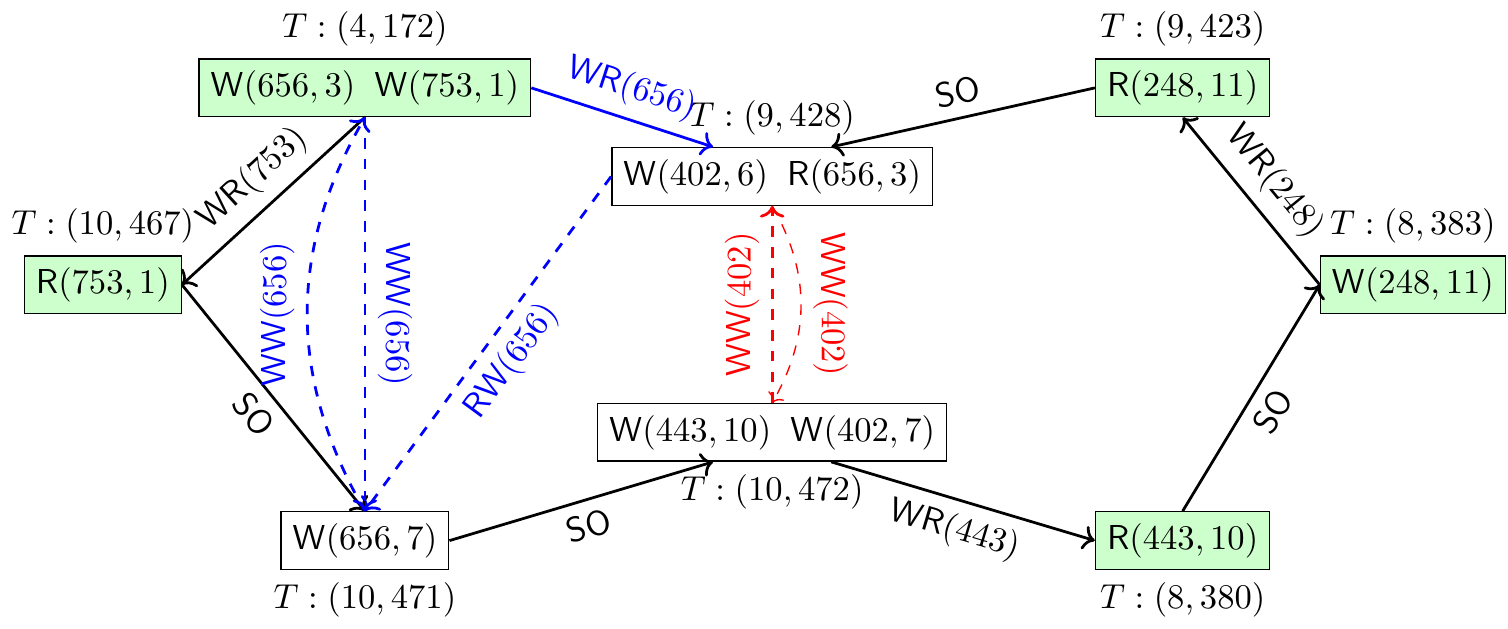}
		\caption{Missing participants}
	\end{subfigure}\vspace{3ex}
	\begin{subfigure}[b]{0.38\textwidth}
		\centering
		\includegraphics[width = 1.25\textwidth]{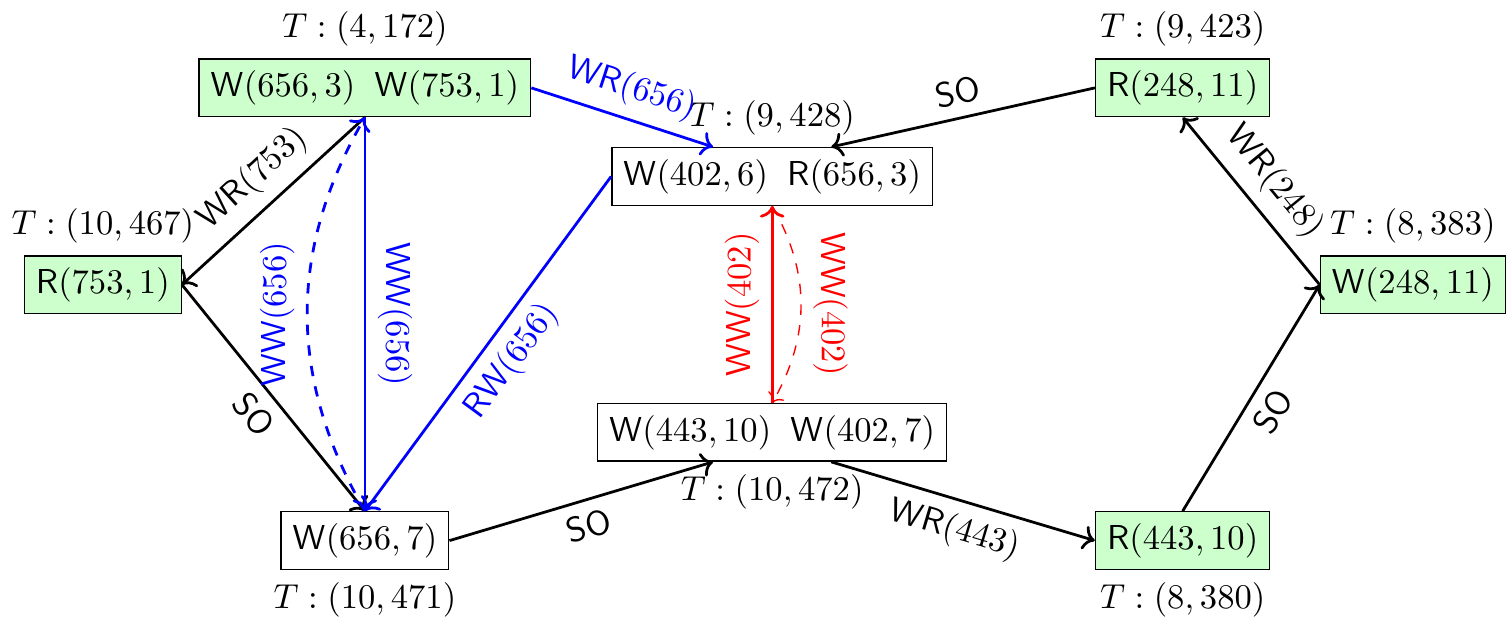}
		\caption{Recovered scenario}
	\end{subfigure}\hspace{15ex}
	\begin{subfigure}[b]{0.38\textwidth}
		\centering
		\includegraphics[width = 1.25\textwidth]{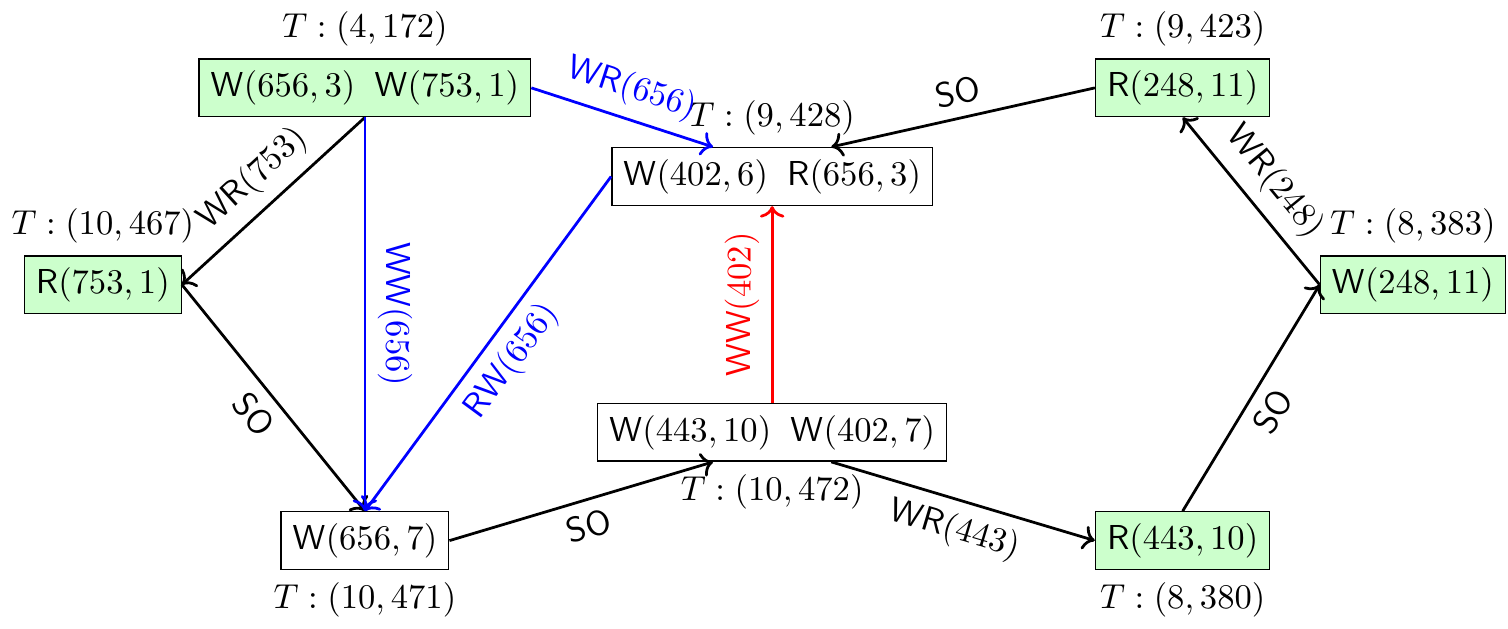}
		\caption{Finalized scenario}
	\end{subfigure}\hspace{15ex}
	\caption{\label{ce:dgraph}  Causality violation: the SI anomaly
		found in Dgraph.
		Dashed and solid  arrows represent uncertain and  certain dependencies, respectively.
		Recovered transactions are colored in green.
		The core dependencies involved in the two sub-scenarios are colored in red and blue, respectively.}
\end{figure*}

There might be multiple ``missing'' transactions that together contribute to a violation.
As depicted in Figure~\ref{ce:dgraph}(b),   \name{} restores the potentially involved five transactions (in green) and the associated dependencies from the original output cycle in Figure~\ref{ce:dgraph}(a).
In particular,  two sub-scenarios are involved: the left subgraph  concerns the $\RW(656)$ dependency on key 656 (in blue) while
the  right subgraph concerns the dependency $\WW(402)$ on key 402 (in red).
Following the same procedure as in the MariaDB-Galera example,
\name{}  resolves the outstanding dependencies if no  undesired cycles arise; see Figure \ref{ce:dgraph}(c).
For example,  we have T:(4,172)$\rel{\WW}$T:(10,471) (in blue) as there would otherwise be a cycle T:(10,471)$\rel{\WW}$T:(4,172)$\rel{\WR}$T:(10,467)$\rel{\SO}$T:(10,471).
The final violating scenario is shown in Figure \ref{ce:dgraph}(d) where two  impossible (dashed) dependencies in Figure \ref{ce:dgraph}(c) have been eliminated.


This violation occurs as the causality order is violated, which is a happens-before relationship between any two transactions in a given history~\cite{DBLP:journals/cacm/Lamport78,Eiger:NSDI2013}.\footnote{Intuitively,
	transaction $T$ causally depends
	on  transaction $S$ if any of the following conditions holds:
	(i) $T$ and $S$ are issued in the same session and $S$ is executed before $T$;
	(ii) $T$ reads the value written by $S$; and
	(iii) there exists another transaction $R$ such that $T$ causally depends on $R$ which in turn causally depends on $S$. 
}
More specifically,
transaction T:(9,428) causally depends on   transaction T:(10,471) (via the  counterclockwise path in the right subgraph) and transaction T:(10,471) causally depends on transaction T:(4,172) (via the clockwise path in the left subgraph).  Hence,
transaction T:(9,428) should have fetched the value 7 of key 656 written by transaction T:(10,471),  instead of the value 3 of transaction T:(4,172), to respect the causality order.

\subsection{A Causality Violation Found in YugabyteDB} 

\input{figs/yugabytedb-ce}

Figure~\ref{fig:yugabytedb-ce} shows how Algorithm~\ref{alg:interpret}
interprets a causality violation found in YugabyteDB.
MonoSAT reports an undesired cycle $T: (0, 7) \rel{\WW(10)} T: (1,15) \rel{\WR(13)} T: (0, 6) \rel{\SO} T: (0, 7)$; see Figure~\ref{fig:yugabytedb-ce-original}.
In this scenario, no transactions  observe values that have been causally overwritten;
consider the only transaction $T: (0, 6)$ in Figure~\ref{fig:yugabytedb-ce-original} that reads.
To locate the cause of the causality violation,
\name{} first finds the (only) ``missing'' transaction $T: (0, 9)$ (colored in green)
and the associated dependencies, as shown in Figure~\ref{fig:yugabytedb-ce-participants}.
The $\WW$ and $\RW$ dependencies are uncertain at this moment (represented by dashed arrows),
while the $\WR$ dependency is certain (represented by solid arrows, colored in blue).
\name{} then restores the violating scenario by resolving such uncertainties.
Specifically, as shown in Figure~\ref{fig:yugabytedb-ce-recovered},
\name{} determines that $W(10, 3)$ of transaction $T: (0, 7)$ was actually
installed before $W(10, 26)$ of transaction $T: (1, 15)$, i.e., $T: (0, 7) \rel{\WW(10)} T: (1, 15)$.
Otherwise, the other direction of dependency $T: (1, 15) \rel{\WW(10)} T: (0, 7)$
would enforce the dependency $T: (0, 9) \rel{\RW(10)} T: (0, 7)$,
which would create an undesired cycle with the known dependency $T: (0, 7) \rel{\SO} T: (0, 9)$.
Finally, \name{} finalizes the violating scenario by removing the remaining uncertainties
from Figure~\ref{fig:yugabytedb-ce-recovered} to obtain Figure~\ref{fig:yugabytedb-ce-finalized}.

Thanks to the participation of transaction $T: (0, 9)$,
the cause of the causality violation becomes clear:
Transaction $T: (0, 7)$ causally depends on transaction $T: (1, 15)$,
via $T: (1, 15) \rel{\WR(13)} T: (0, 6) \rel{\SO} T: (0, 7)$.
However, transaction $T: (0, 9)$ following transaction $T: (0,7)$ on the same session
reads the value $26$ of key $10$ from transaction $T: (1, 15)$,
which should have been overwritten by transaction $T: (0, 7)$.

%% file: figs/yugabytedb-ce.tex

\begin{figure*}[t]
  \centering
    \begin{subfigure}[b]{0.2\textwidth}
        \centering
        \includegraphics[width = \textwidth]{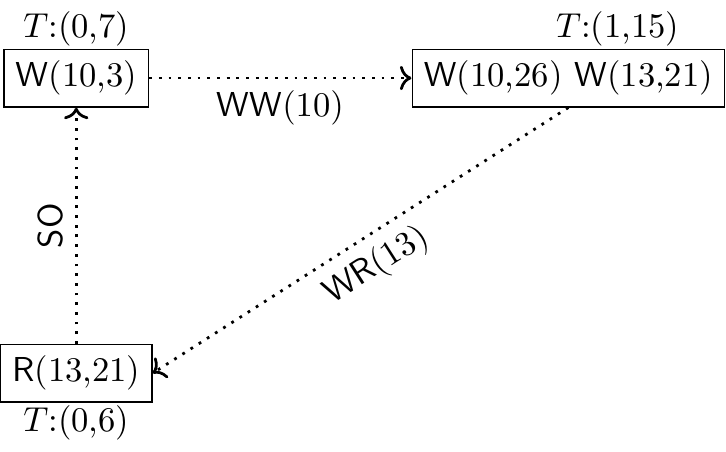}
        \caption{Original output}
        \label{fig:yugabytedb-ce-original}
    \end{subfigure}\hspace{6ex}
    \begin{subfigure}[b]{0.17\textwidth}
        \centering
        \includegraphics[width = 1.25\textwidth]{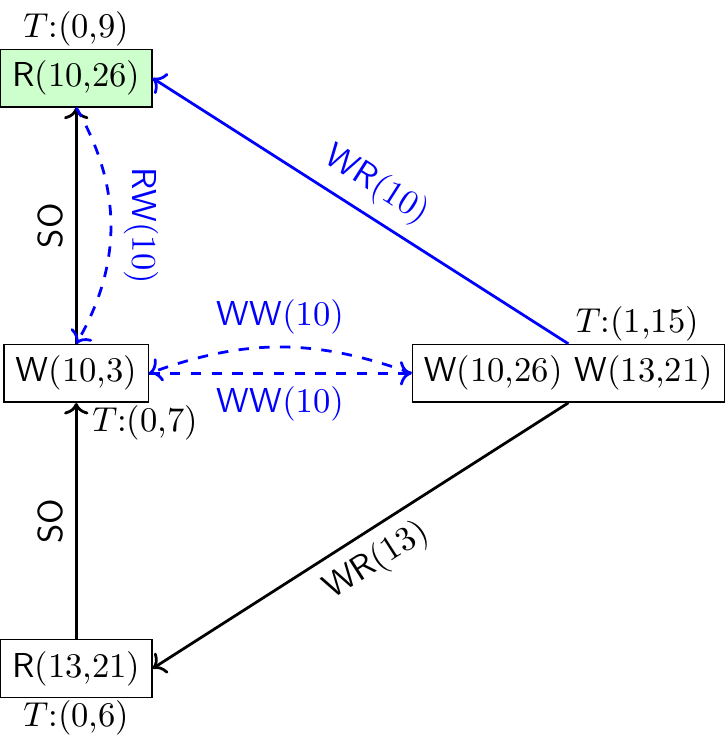}
        \caption{Missing participants}
		\label{fig:yugabytedb-ce-participants}
    \end{subfigure}\hspace{10ex}
    \begin{subfigure}[b]{0.17\textwidth}
        \centering
        \includegraphics[width = 1.25\textwidth]{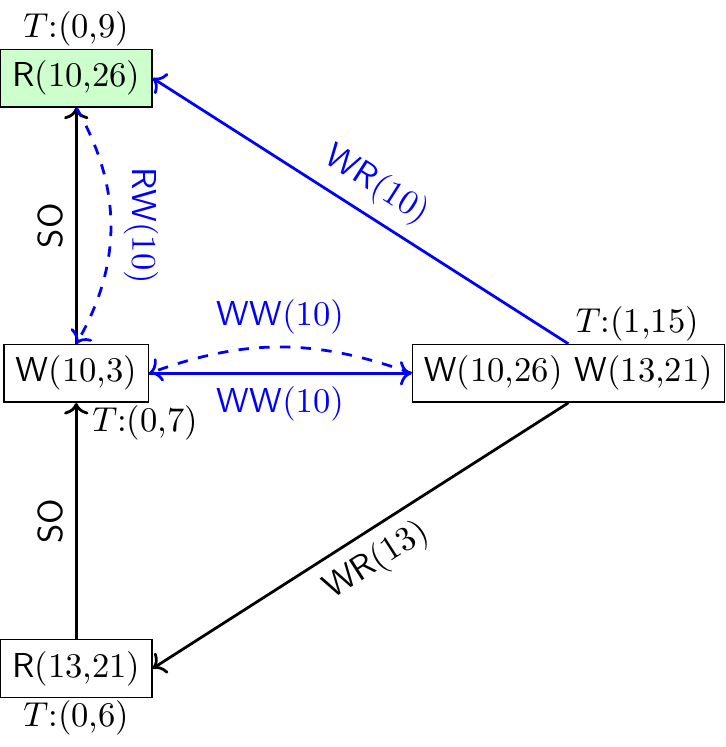}
        \caption{Recovered scenario}
		\label{fig:yugabytedb-ce-recovered}
    \end{subfigure}\hspace{10ex}
    \begin{subfigure}[b]{0.17\textwidth}
        \centering
        \includegraphics[width = 1.25\textwidth]{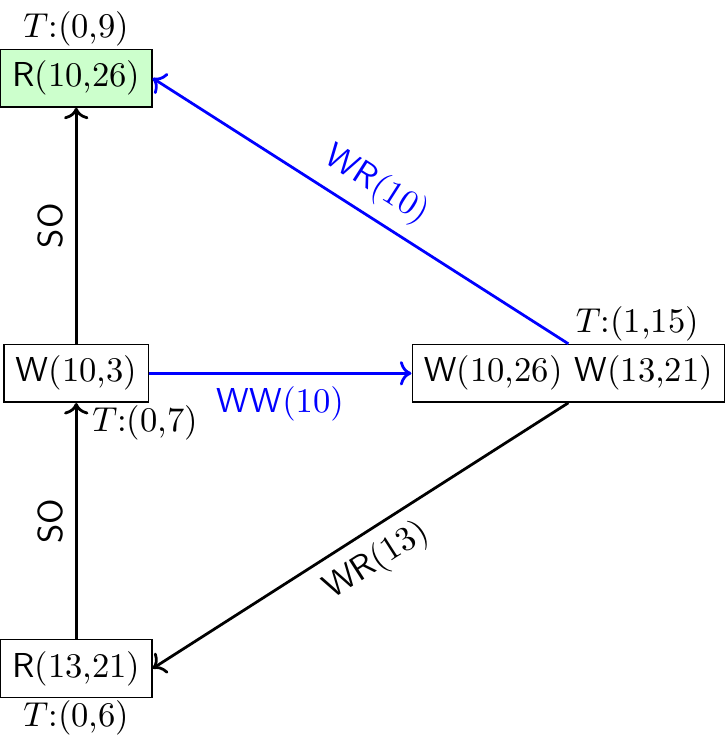}
        \caption{Finalized scenario}
		\label{fig:yugabytedb-ce-finalized}
    \end{subfigure}
    \caption{Causality violation: the SI anomaly found in YugabyteDB.
    The recovered dependencies are colored in blue with dashed and solid arrows indicating uncertain  and certain dependencies,  respectively.  The (only) missing transaction is colored in green.}
	\label{fig:yugabytedb-ce}
\end{figure*}

%% file: tables/memory.tex


%% file: tables/large.tex


%% file: tables/random-large.tex


%% file: sections/appendix-minimal.tex


\section{Minimality of Counterexamples}
\label{section:appendix-minimality}

In this section we show that our interpretation algorithm
actually returns a \emph{minimal} counterexample
which contains the cycle found by MonoSAT and
is just informative to help understand the violation.

\subsection{Definitions} \label{ss:defs}


\begin{definition} \label{def:vio}
  A \bfit{violation} is a polygraph that fails the checking of PolySI.
\end{definition}

Let $C$ be the cycle found by MonoSAT in polygraph $G = \{V,E,Cons\}$.
Consider the violation $vio = \{V',E',Cons'\}$ returned by the interpretation algorithm.
It satisfies the following properties:

\begin{itemize}
  \item $vio \subseteq G$, which means $V' \subseteq V, E' \subseteq E, Cons' \subseteq Cons$
  \item $C \subseteq vio$. 
  \item $vio$ fails the checking of PolySI.
\end{itemize}

We call $vio$ a violation with respect to $\set{G,C}$.

\begin{definition} \label{def:minimal-violation}
    A violation $vio$ is \bfit{minimal},
    if removing any dependency edge from $vio$ makes it no longer a violation.
\end{definition}

There may be more than one minimal violation with respect to $\{G,C\}$.
And we name the one who has the least number of dependencies
\bfit{the minimal counter example with respect to $\{G,C\}$},
which we expect to get.

\begin{definition} \label{def:minimal counterexample}
    Given a polygraph $G$ and a cycle $C$,
    there may be more than one minimal violations based on $\{G,C\}$.
    We name the one which has the least number of dependencies
    \bfit{the minimal counterexample with respect to $\{G,C\}$}.
\end{definition}


\subsection{Patterns of Minimal Counterexamples}

A minimal counter example is a special polygraph,
which consists of several cycles.
We define a data structure called "adjoining cycle set",
and prove that a minimal counterexample is exactly 
a minimal complete adjoining cycle set.

\begin{definition} \label{def:acs}
    A set of cycles $acs$ is called \bfit{adjoining cycle set}, 
    if $\forall$ cycle $C_1 \in acs$, 
    $\exists$ a constraint $cons = \{dep_1,.../dep_2,...\}$ and a cycle $C_2 \in acs$, s.t. $dep_1 \in C_1$ and $dep_2 \in C_2$.
    \begin{itemize}
        \item $C_1$ is the adjoining cycle of $C_2$, while $C_2$ is the adjoining cycle of $C_1$.
        \item If $\forall$ constraint $cons$, the two choice of $cons$ both have dependencies in $acs$ or neither have dependencies in $acs$, then $acs$ is \bfit{complete}.
    \end{itemize}
\end{definition}

Definition \ref{def:acs} describes the data structure called \bfit{adjoining cycle set}.
Figure \ref{def_complete_acs} gives an example of adjoining cycle set. 
When we say a adjoining cycle set is complete, 
it means no more cycles can be added to the set 
and each constraint included in the set has two choices.
$\{C1,C3\}$ is an adjoining cycle set, 
but it is incomplete because there exists a constraint (the green one) 
that has only one choice in $\{C1,C3\}$.
$\{C1,C2,C3\}$ is a complete adjoining cycle set, 
since each constraint has two choices in $\{C1,C2.C3\}$.

\begin{figure}[ht] 
    \centering
    \includegraphics[scale = 0.5]{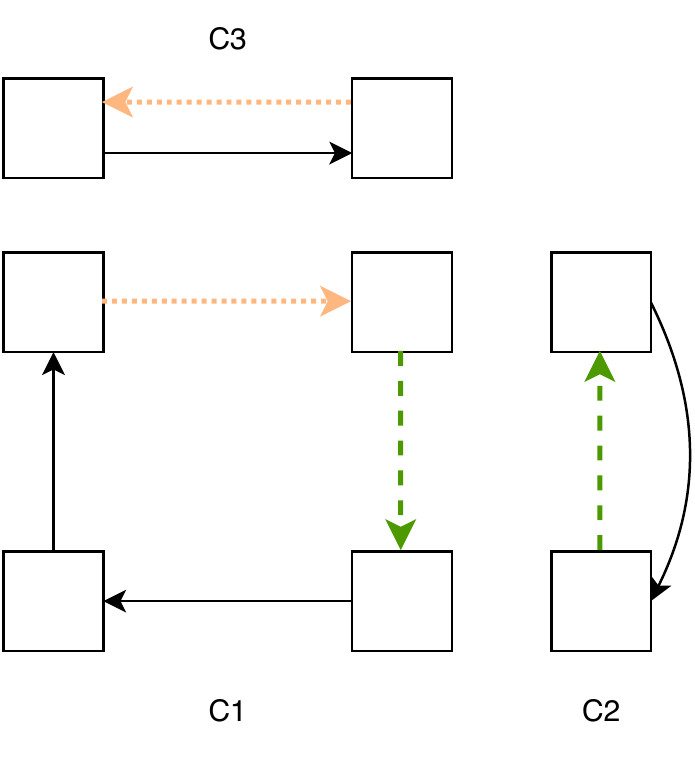}
    \caption{Complete and incomplete adjoining cycle sets. $\{C1,C2,C3\}$ is a complete adjoining cycle set. $\{C1,C2\}$ and $\{C1,C3\}$ are incomplete cycle sets.}
    \label{def_complete_acs}
\end{figure}

\begin{lemma} \label{lemma:one-acs-in-minimal}
    A minimal violation is exactly a complete adjoining cycle set.
\end{lemma}

\begin{proof}
    If there exists a violation $vio$, which do not contain any complete adjoining cycle set.
    It contain at least one incomplete cycle set $acs_1$,
    because a violation must have at least one cycle.
    So there exist a cycle $C_1 \in acs_1$ and a constraint $cons = \{dep_1...,/...\}$, where $dep_1 \in acs_1$ and all dependencies in the other choice of $cons$ do not contained by $acs_1$.
    Then we let $cons$ choose the other choice, and then $dep_1$ will disappear and $C_1$ will not be a cycle any more.
    Next we focus on $acs_1-\{C_1\}$, (if $acs_1$ has more than one cycle),
    we will get another incomplete cycle group $acs_2 = acs_1 - \{C_1\}$.
    Keep doing this recurisely, all cycles in $acs_1$ will be broken.
    And this polygraph will not be a violation.
    
    If there exists two complete adjoining cycle group in $vio$,
    then we can delete one complete adjoining cycle group and 
    the left polygraph is still a violation. So it is not a minimal violation.

    In conclusion, a minimal violation $vio$ contains exactly one complete adjoining cycle set $acs$. 
    Since both choices of each constraint are in one cycle, 
    $acs$ is a violation that cannot pass the verification of PolySI.
    $vio$ is a minimal violation, $acs \subseteq vio$ and $acs$ is a violation, 
    then it is trivial that $vio = acs$.
\end{proof}

\begin{definition} \label{def:smallest complete acs}
    A complete adjoining cycle set $acs$ is called \bfit{minimal complete adjoining cycle set containing cycle $C$} if 
    \begin{itemize}
        \item $C \in acs$
        \item $\forall$ complete adjoining cycle set $acs'$ containing $C$, the number of dependencies in $acs'$ $\geq$ that in $acs$
    \end{itemize}
\end{definition}

\begin{theorem}[Minimal counterexample] \label{thm:minimal-counterexample}
    A minimal complete adjoining cycle set containing cycle $C$ is exactly
    the minimal counter example with respect to $\{G,C\}$.
\end{theorem}

\begin{proof}
    A minimal complete adjoining cycle set containing $C$ is 
    one of complete adjoining cycle sets containing $C$ 
    which has the least number of dependencies.
    By lemma \ref{lemma:one-acs-in-minimal}, it is the smallest
    minimal violation containing cycle $C$.
    So it is one of minimal violations with respect to $\{G,C\}$, 
    and has the least number of dependencies.
    By definition \ref{def:minimal counterexample}, it is exactly the
    minimal counter example with respect to $\{G,C\}$.
\end{proof}


\subsection{Restoring Minimal Patterns}

The core idea of our interpretation algorithm is to restore the 
\bfit{minimal complete adjoining cycle set containing $C$}.
Function $Find\_ACS$ is designed to do this task.
(line~\code{\ref{alg:interpret}}{\ref{procedure:restore acs}})

For each $\WW$ or $\RW$ dependency in an undesired cycle, 
it checks each dependency in the other direction (if missing),
and restores the undesired cycles containing these dependencies
(line~\code{\ref{alg:interpret}}{\ref{recover-edge-in-cons-1}} 
- line~\code{\ref{alg:interpret}}{\ref{recover-cons-finish-2}}).
Doing this recursively (line~\code{\ref{alg:interpret}}{\ref{find-cycle-recursively1}} and line~\code{\ref{alg:interpret}}{\ref{find-cycle-recursively2}}), 
the function will detect all of the adjoining cycle sets that contain $\acs$.
The minimal one will be returned.

The algorithm uses the brute force and sounds inefficient. 
However, it works well in practical experiments.
From our experience of verification, 
when it tries to detect a cycle, 
there usually exists a small cycle with only one $\WW$ or $\RW$ dependency.
For example, in the violation we found in figure~\ref{ce:dgraph} a, 
we first need to find a adjoining cycle from the dependency 
$T:(9,428)\rel{}T:{10,471}$. 
There exists a simple cycle $T:(10,471)\rel{}T:(4,172)\rel{}T:(10,467)\rel{}T:(10,471)$.
This cycle has a small size and only has one $\WW$ dependency, 
which means the function does not need to detect more cycles recursively 
in order to restore a complete adjoining cycle set, 
and can pay no attention to other possible adjoining cycle sets 
that contain more than 3 dependencies.

However, there still exist the worst cases that the interpretation algorithm cannot
return the minimal counter example quickly. 
But in this case, we can be more patient to wait for it.
Since we have already known the result that there exists a violation,
and we just need to wait for the minimal counter example.
And if it still costs too much time, we can interrupt it and output 
a minimal violation we have found instead of the smallest one,
after the interpretation algorithm has run enough time.

\begin{theorem}[Minimality]
	 \label{thm:minimal-counterexample-polysi}
	\name{} always returns a minimal counterexample with respect to $G$ and $C$, with
	$G$ the polygraph built from a history and  $C$ the cycle output by MonoSAT.
\end{theorem}

\begin{proof}
    Our interpretation algorithm will return a `minimal complete adjoining cycle set' containing cycle $C$,
    which is exactly the minimal counter example.
\end{proof}

Now we have restored the minimal counter example. 
However, to help understand the violation,
we still restore some more dependencies.
For each $rw$ dependency appears in violations, 
it is caused by a $wr$ dependency and a $ww$ dependency.
If we delete the $wr$ dependency and $ww$ dependency, 
it is difficult to understand the $rw$ dependency.
So for each $rw$ dependency in a minimal counterexample, 
we recover the related $wr$ and $ww$ dependencies to the violation
(Alg\ref{alg:interpret}:line\ref{recover-rw-edges}-\ref{recover-rw-edges-finish}).
This is very useful in our experiments.
For example, in the violation we found in figure~\ref{ce:galera},
the transaction $T:(1,4)$ is recovered from two $rw$ dependencies.
Without the recovered transaction, 
it is difficult to understand the violation.

%% file: tables/runtime-elle.tex
\newpage
\section{Performance Evaluation of PolySI-List}
\label{app:polysi-list}

Figure~\ref{comparison-polysi-list} shows the performance evaluation of PolySI-List, the extension to \name{} for handling Elle-like logs/histories with the ``list'' data structure. 

\pgfplotsset{height=140pt, width=200pt}
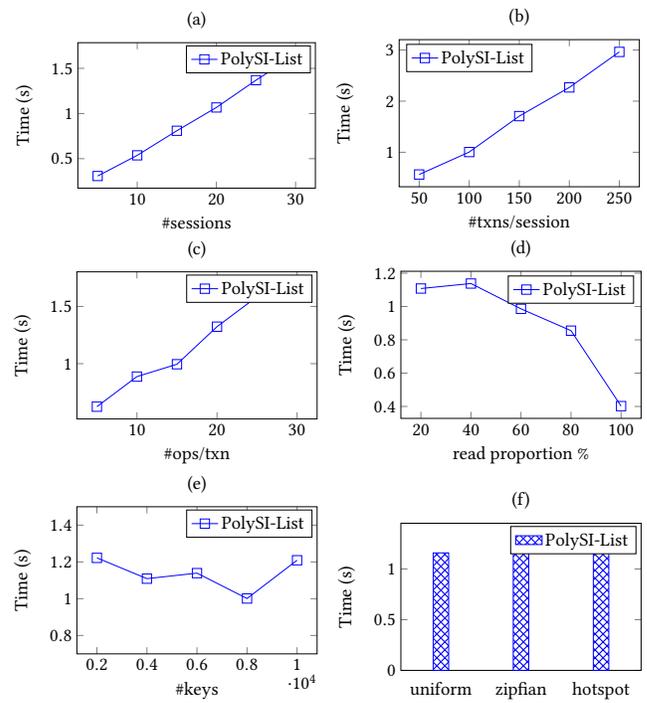
\begin{figure}[h]
    \begin{scaletikzpicturetowidth}{0.23\textwidth}
	\begin{tikzpicture}[scale=\tikzscale]
	  \begin{axis}[
		title={(a)},
		xlabel={\#sessions},
		ylabel={Time (s)},
		cycle multiindex* list={
            color       \nextlist
            mark list*  \nextlist
        }
		]
		\addplot[color=blue,mark=square,mark size=3pt] table [x=param, y=si(list), col sep=comma] {tables/data/elle-sessions.csv};
		\legend{\name-List}
	  \end{axis}
	\end{tikzpicture}
		\hspace{1ex}
  \end{scaletikzpicturetowidth}           
  \begin{scaletikzpicturetowidth}{0.23\textwidth}
	\begin{tikzpicture}[scale=\tikzscale]
	  \begin{axis}[
		title={(b)},
		xlabel={\#txns/session},
		ylabel={Time (s)},
		legend pos=north west
		]
		\addplot[color=blue,mark=square,mark size=3pt] table [x=param, y=si(list), col sep=comma] {tables/data/elle-txns.csv};
		\legend{\name-List}
	  \end{axis}
	\end{tikzpicture}
  \end{scaletikzpicturetowidth}

  \begin{scaletikzpicturetowidth}{0.23\textwidth}
	\begin{tikzpicture}[scale=\tikzscale]
	  \begin{axis}[
		title={(c)},
		xlabel={\#ops/txn},
		ylabel={Time (s)},
		]
		\addplot[color=blue,mark=square,mark size=3pt] table [x=param, y=si(list), col sep=comma] {tables/data/elle-nops.csv};
		\legend{\name-List}
	  \end{axis}
	\end{tikzpicture}
  \end{scaletikzpicturetowidth}
	\hspace{1ex}
  \begin{scaletikzpicturetowidth}{0.23\textwidth}
	\begin{tikzpicture}[scale=\tikzscale]
	  \begin{axis}[
		title={(d)},
		xlabel={ read proportion \%},
		ylabel={Time (s)},
		legend pos=north east
		]
		\addplot[color=blue,mark=square,mark size=3pt] table [x=param, y=si(list), col sep=comma] {tables/data/elle-readpct.csv};
		\legend{\name-List}
	  \end{axis}
	\end{tikzpicture}

  \end{scaletikzpicturetowidth}
    \begin{scaletikzpicturetowidth}{0.23\textwidth}
	\begin{tikzpicture}[scale=\tikzscale]
	  \begin{axis}[
		title={(e)},
		ymin=0.7,
		ymax=1.5,
		xlabel={\#keys},
		ylabel={Time (s)},
		]
		\addplot[color=blue,mark=square,mark size=3pt] table [x=param, y=si(list), col sep=comma] {tables/data/elle-vars.csv};
		\legend{\name-List}
	  \end{axis}
	\end{tikzpicture}
  \end{scaletikzpicturetowidth}
  	\hspace{1ex}
  \begin{scaletikzpicturetowidth}{0.23\textwidth}
	\begin{tikzpicture}[scale=\tikzscale]
	  \begin{axis}[
		title={(f)},
		x tick style={draw=none},
		ylabel={Time (s)},
		ymin=0,
		xmin=-0.5,
		xmax=2.5,
		ybar,
		area legend,
		xtick=data,
		xticklabels={uniform,zipfian,hotspot}
		]
		\addplot[color=blue, pattern color=blue, pattern=crosshatch] table [x expr=\coordindex, y=si(list), col sep=comma] {tables/data/elle-distrib.csv};
		\legend{\name-List}
	  \end{axis}
	\end{tikzpicture}
  \end{scaletikzpicturetowidth}
 \caption{Performance evaluation of PolySI-List.}
 \label{comparison-polysi-list}
\end{figure}

%% file: si-sat.bbl

\begin{thebibliography}{54}


\ifx \showCODEN    \undefined \def \showCODEN     #1{\unskip}     \fi
\ifx \showDOI      \undefined \def \showDOI       #1{#1}\fi
\ifx \showISBNx    \undefined \def \showISBNx     #1{\unskip}     \fi
\ifx \showISBNxiii \undefined \def \showISBNxiii  #1{\unskip}     \fi
\ifx \showISSN     \undefined \def \showISSN      #1{\unskip}     \fi
\ifx \showLCCN     \undefined \def \showLCCN      #1{\unskip}     \fi
\ifx \shownote     \undefined \def \shownote      #1{#1}          \fi
\ifx \showarticletitle \undefined \def \showarticletitle #1{#1}   \fi
\ifx \showURL      \undefined \def \showURL       {\relax}        \fi
\providecommand\bibfield[2]{#2}
\providecommand\bibinfo[2]{#2}
\providecommand\natexlab[1]{#1}
\providecommand\showeprint[2][]{arXiv:#2}

\bibitem[\protect\citeauthoryear{Adya}{Adya}{1999}]%
        {Adya:PhDThesis1999}
\bibfield{author}{\bibinfo{person}{Atul Adya}.}
  \bibinfo{year}{1999}\natexlab{}.
\newblock \emph{\bibinfo{title}{Weak Consistency: A Generalized Theory and
  Optimistic Implementations for Distributed Transactions}}.
\newblock \bibinfo{thesistype}{Ph.D. Dissertation}. \bibinfo{address}{USA}.
\newblock


\bibitem[\protect\citeauthoryear{Bailis, Davidson, Fekete, Ghodsi, Hellerstein,
  and Stoica}{Bailis et~al\mbox{.}}{2013}]%
        {HAT:VLDB2013}
\bibfield{author}{\bibinfo{person}{Peter Bailis}, \bibinfo{person}{Aaron
  Davidson}, \bibinfo{person}{Alan Fekete}, \bibinfo{person}{Ali Ghodsi},
  \bibinfo{person}{Joseph~M. Hellerstein}, {and} \bibinfo{person}{Ion Stoica}.}
  \bibinfo{year}{2013}\natexlab{}.
\newblock \showarticletitle{Highly Available Transactions: Virtues and
  Limitations}.
\newblock \bibinfo{journal}{\emph{Proc. VLDB Endow.}} \bibinfo{volume}{7},
  \bibinfo{number}{3} (\bibinfo{date}{nov} \bibinfo{year}{2013}),
  \bibinfo{pages}{181--192}.
\newblock
\showISSN{2150-8097}
\urldef\tempurl%
\url{https://doi.org/10.14778/2732232.2732237}
\showDOI{\tempurl}


\bibitem[\protect\citeauthoryear{Bailis, Fekete, Ghodsi, Hellerstein, and
  Stoica}{Bailis et~al\mbox{.}}{2016}]%
        {RAMP:TODS2016}
\bibfield{author}{\bibinfo{person}{Peter Bailis}, \bibinfo{person}{Alan
  Fekete}, \bibinfo{person}{Ali Ghodsi}, \bibinfo{person}{Joseph~M.
  Hellerstein}, {and} \bibinfo{person}{Ion Stoica}.}
  \bibinfo{year}{2016}\natexlab{}.
\newblock \showarticletitle{Scalable Atomic Visibility with RAMP Transactions}.
\newblock \bibinfo{journal}{\emph{ACM Trans. Database Syst.}}
  \bibinfo{volume}{41}, \bibinfo{number}{3}, Article \bibinfo{articleno}{15}
  (\bibinfo{date}{jul} \bibinfo{year}{2016}), \bibinfo{numpages}{45}~pages.
\newblock
\showISSN{0362-5915}
\urldef\tempurl%
\url{https://doi.org/10.1145/2909870}
\showDOI{\tempurl}


\bibitem[\protect\citeauthoryear{Bayless, Bayless, Hoos, and Hu}{Bayless
  et~al\mbox{.}}{2015}]%
        {MonoSAT:AAAI2015}
\bibfield{author}{\bibinfo{person}{Sam Bayless}, \bibinfo{person}{Noah
  Bayless}, \bibinfo{person}{Holger~H. Hoos}, {and} \bibinfo{person}{Alan~J.
  Hu}.} \bibinfo{year}{2015}\natexlab{}.
\newblock \showarticletitle{SAT modulo Monotonic Theories}. In
  \bibinfo{booktitle}{\emph{Proceedings of the Twenty-Ninth AAAI Conference on
  Artificial Intelligence}} \emph{(\bibinfo{series}{AAAI'15})}.
  \bibinfo{publisher}{AAAI Press}, \bibinfo{pages}{3702--3709}.
\newblock
\showISBNx{0262511290}


\bibitem[\protect\citeauthoryear{Berenson, Bernstein, Gray, Melton, O'Neil, and
  O'Neil}{Berenson et~al\mbox{.}}{1995}]%
        {CritiqueANSI:SIGMOD1995}
\bibfield{author}{\bibinfo{person}{Hal Berenson}, \bibinfo{person}{Phil
  Bernstein}, \bibinfo{person}{Jim Gray}, \bibinfo{person}{Jim Melton},
  \bibinfo{person}{Elizabeth O'Neil}, {and} \bibinfo{person}{Patrick O'Neil}.}
  \bibinfo{year}{1995}\natexlab{}.
\newblock \showarticletitle{A Critique of ANSI SQL Isolation Levels}. In
  \bibinfo{booktitle}{\emph{SIGMOD '95}}. \bibinfo{publisher}{ACM},
  \bibinfo{pages}{1--10}.
\newblock
\showISBNx{0897917316}
\urldef\tempurl%
\url{https://doi.org/10.1145/223784.223785}
\showDOI{\tempurl}


\bibitem[\protect\citeauthoryear{Bernstein, Hadzilacos, and Goodman}{Bernstein
  et~al\mbox{.}}{1986}]%
        {Bernstein:Book1986}
\bibfield{author}{\bibinfo{person}{Philip~A Bernstein}, \bibinfo{person}{Vassos
  Hadzilacos}, {and} \bibinfo{person}{Nathan Goodman}.}
  \bibinfo{year}{1986}\natexlab{}.
\newblock \bibinfo{booktitle}{\emph{Concurrency Control and Recovery in
  Database Systems}}.
\newblock \bibinfo{publisher}{Addison-Wesley Longman Publishing Co., Inc.},
  \bibinfo{address}{USA}.
\newblock
\showISBNx{0201107155}


\bibitem[\protect\citeauthoryear{Biswas and Enea}{Biswas and Enea}{2019}]%
        {Complexity:OOPSLA2019}
\bibfield{author}{\bibinfo{person}{Ranadeep Biswas} {and}
  \bibinfo{person}{Constantin Enea}.} \bibinfo{year}{2019}\natexlab{}.
\newblock \showarticletitle{On the Complexity of Checking Transactional
  Consistency}.
\newblock \bibinfo{journal}{\emph{Proc. ACM Program. Lang.}}
  \bibinfo{volume}{3}, \bibinfo{number}{OOPSLA}, Article
  \bibinfo{articleno}{165} (\bibinfo{date}{Oct.} \bibinfo{year}{2019}),
  \bibinfo{numpages}{28}~pages.
\newblock
\urldef\tempurl%
\url{https://doi.org/10.1145/3360591}
\showDOI{\tempurl}


\bibitem[\protect\citeauthoryear{Biswas, Kakwani, Vedurada, Enea, and
  Lal}{Biswas et~al\mbox{.}}{2021}]%
        {MonkeyDB:OOPSLA2021}
\bibfield{author}{\bibinfo{person}{Ranadeep Biswas}, \bibinfo{person}{Diptanshu
  Kakwani}, \bibinfo{person}{Jyothi Vedurada}, \bibinfo{person}{Constantin
  Enea}, {and} \bibinfo{person}{Akash Lal}.} \bibinfo{year}{2021}\natexlab{}.
\newblock \showarticletitle{MonkeyDB: Effectively Testing Correctness under
  Weak Isolation Levels}.
\newblock \bibinfo{journal}{\emph{Proc. ACM Program. Lang.}}
  \bibinfo{volume}{5}, \bibinfo{number}{OOPSLA}, Article
  \bibinfo{articleno}{132} (\bibinfo{date}{oct} \bibinfo{year}{2021}),
  \bibinfo{numpages}{27}~pages.
\newblock
\urldef\tempurl%
\url{https://doi.org/10.1145/3485546}
\showDOI{\tempurl}


\bibitem[\protect\citeauthoryear{Bouajjani, Enea, Guerraoui, and
  Hamza}{Bouajjani et~al\mbox{.}}{2017}]%
        {VCC:POPL2017}
\bibfield{author}{\bibinfo{person}{Ahmed Bouajjani},
  \bibinfo{person}{Constantin Enea}, \bibinfo{person}{Rachid Guerraoui}, {and}
  \bibinfo{person}{Jad Hamza}.} \bibinfo{year}{2017}\natexlab{}.
\newblock \showarticletitle{On verifying causal consistency}. In
  \bibinfo{booktitle}{\emph{POPL'17}}. \bibinfo{publisher}{{ACM}},
  \bibinfo{pages}{626--638}.
\newblock


\bibitem[\protect\citeauthoryear{Cerone, Bernardi, and Gotsman}{Cerone
  et~al\mbox{.}}{2015}]%
        {Framework:CONCUR2015}
\bibfield{author}{\bibinfo{person}{Andrea Cerone}, \bibinfo{person}{Giovanni
  Bernardi}, {and} \bibinfo{person}{Alexey Gotsman}.}
  \bibinfo{year}{2015}\natexlab{}.
\newblock \showarticletitle{{A Framework for Transactional Consistency Models
  with Atomic Visibility}}. In \bibinfo{booktitle}{\emph{CONCUR'15}}
  \emph{(\bibinfo{series}{LIPIcs})}, Vol.~\bibinfo{volume}{42}.
  \bibinfo{publisher}{Schloss Dagstuhl--Leibniz-Zentrum fuer Informatik},
  \bibinfo{pages}{58--71}.
\newblock
\showISBNx{978-3-939897-91-0}
\showISSN{1868-8969}


\bibitem[\protect\citeauthoryear{Cerone and Gotsman}{Cerone and
  Gotsman}{2018}]%
        {AnalysingSI:JACM2018}
\bibfield{author}{\bibinfo{person}{Andrea Cerone} {and} \bibinfo{person}{Alexey
  Gotsman}.} \bibinfo{year}{2018}\natexlab{}.
\newblock \showarticletitle{Analysing Snapshot Isolation}.
\newblock \bibinfo{journal}{\emph{J. ACM}} \bibinfo{volume}{65},
  \bibinfo{number}{2}, Article \bibinfo{articleno}{11} (\bibinfo{date}{Jan.}
  \bibinfo{year}{2018}), \bibinfo{numpages}{41}~pages.
\newblock
\showISSN{0004-5411}
\urldef\tempurl%
\url{https://doi.org/10.1145/3152396}
\showDOI{\tempurl}


\bibitem[\protect\citeauthoryear{Chen, Dou, Wang, and Qin}{Chen
  et~al\mbox{.}}{2020}]%
        {cofi}
\bibfield{author}{\bibinfo{person}{Haicheng Chen}, \bibinfo{person}{Wensheng
  Dou}, \bibinfo{person}{Dong Wang}, {and} \bibinfo{person}{Feng Qin}.}
  \bibinfo{year}{2020}\natexlab{}.
\newblock \showarticletitle{CoFI: Consistency-Guided Fault Injection for Cloud
  Systems}. In \bibinfo{booktitle}{\emph{{ASE} 2020}}.
  \bibinfo{publisher}{{IEEE}}.
\newblock
\urldef\tempurl%
\url{https://doi.org/10.1145/3324884.3416548}
\showDOI{\tempurl}


\bibitem[\protect\citeauthoryear{Clavel, Dur\'{a}n, Eker, Lincoln,
  Mart\'{\i}-Oliet, Meseguer, and Talcott}{Clavel et~al\mbox{.}}{2007}]%
        {AllAboutMaude:Book2007}
\bibfield{author}{\bibinfo{person}{Manuel Clavel}, \bibinfo{person}{Francisco
  Dur\'{a}n}, \bibinfo{person}{Steven Eker}, \bibinfo{person}{Patrick Lincoln},
  \bibinfo{person}{Narciso Mart\'{\i}-Oliet}, \bibinfo{person}{Jos\'{e}
  Meseguer}, {and} \bibinfo{person}{Carolyn Talcott}.}
  \bibinfo{year}{2007}\natexlab{}.
\newblock \bibinfo{booktitle}{\emph{All about Maude - a High-Performance
  Logical Framework: How to Specify, Program and Verify Systems in Rewriting
  Logic}}.
\newblock \bibinfo{publisher}{Springer-Verlag}, \bibinfo{address}{Berlin,
  Heidelberg}.
\newblock
\showISBNx{3540719407}


\bibitem[\protect\citeauthoryear{Cluster}{Cluster}{2022}]%
        {maria-galera}
\bibfield{author}{\bibinfo{person}{MariaDB~Galera Cluster}.}
  \bibinfo{year}{Accessed August, 2022}\natexlab{}.
\newblock
\newblock
\newblock
\shownote{\url{https://mariadb.com/kb/en/what-is-mariadb-galera-cluster/}.}


\bibitem[\protect\citeauthoryear{CockroachDB}{CockroachDB}{2022}]%
        {cockroach}
\bibfield{author}{\bibinfo{person}{CockroachDB}.} \bibinfo{year}{Accessed
  August, 2022}\natexlab{}.
\newblock
\newblock
\newblock
\shownote{\url{https://www.cockroachlabs.com/}.}


\bibitem[\protect\citeauthoryear{Cormen, Leiserson, Rivest, and Stein}{Cormen
  et~al\mbox{.}}{2009}]%
        {CLRS2009}
\bibfield{author}{\bibinfo{person}{Thomas~H. Cormen},
  \bibinfo{person}{Charles~E. Leiserson}, \bibinfo{person}{Ronald~L. Rivest},
  {and} \bibinfo{person}{Clifford Stein}.} \bibinfo{year}{2009}\natexlab{}.
\newblock \bibinfo{booktitle}{\emph{Introduction to Algorithms, Third Edition}
  (\bibinfo{edition}{3rd} ed.)}.
\newblock \bibinfo{publisher}{The MIT Press}.
\newblock
\showISBNx{0262033844}


\bibitem[\protect\citeauthoryear{Crooks, Pu, Alvisi, and Clement}{Crooks
  et~al\mbox{.}}{2017}]%
        {ClientCentric:PODC2017}
\bibfield{author}{\bibinfo{person}{Natacha Crooks}, \bibinfo{person}{Youer Pu},
  \bibinfo{person}{Lorenzo Alvisi}, {and} \bibinfo{person}{Allen Clement}.}
  \bibinfo{year}{2017}\natexlab{}.
\newblock \showarticletitle{Seeing is Believing: A Client-Centric Specification
  of Database Isolation}. In \bibinfo{booktitle}{\emph{PODC '17}}.
  \bibinfo{publisher}{ACM}, \bibinfo{pages}{73--82}.
\newblock
\showISBNx{9781450349925}
\urldef\tempurl%
\url{https://doi.org/10.1145/3087801.3087802}
\showDOI{\tempurl}


\bibitem[\protect\citeauthoryear{Darnell}{Darnell}{2022}]%
        {CockroachDB-bug}
\bibfield{author}{\bibinfo{person}{Ben Darnell}.} \bibinfo{year}{Accessed
  August, 2022}\natexlab{}.
\newblock \bibinfo{title}{Lessons Learned from 2+ Years of Nightly Jepsen
  Tests}.
\newblock
\newblock
\newblock
\shownote{\url{https://www.cockroachlabs.com/blog/jepsen-tests-lessons/}.}


\bibitem[\protect\citeauthoryear{Database}{Database}{2022}]%
        {Oracle}
\bibfield{author}{\bibinfo{person}{Oracle Database}.} \bibinfo{year}{Accessed
  August, 2022}\natexlab{}.
\newblock
\newblock
\newblock
\shownote{\url{https://www.oracle.com/database/}.}


\bibitem[\protect\citeauthoryear{Daudjee and Salem}{Daudjee and Salem}{2006}]%
        {LazyReplSI:VLDB2006}
\bibfield{author}{\bibinfo{person}{Khuzaima Daudjee} {and}
  \bibinfo{person}{Kenneth Salem}.} \bibinfo{year}{2006}\natexlab{}.
\newblock \showarticletitle{Lazy Database Replication with Snapshot Isolation}.
  In \bibinfo{booktitle}{\emph{VLDB'06}}. \bibinfo{publisher}{VLDB Endowment},
  \bibinfo{pages}{715--726}.
\newblock


\bibitem[\protect\citeauthoryear{Dgraph}{Dgraph}{2022}]%
        {dgraph}
\bibfield{author}{\bibinfo{person}{Dgraph}.} \bibinfo{year}{Accessed August,
  2022}\natexlab{}.
\newblock
\newblock
\newblock
\shownote{\url{https://dgraph.io/}.}


\bibitem[\protect\citeauthoryear{Didona, Guerraoui, Wang, and
  Zwaenepoel}{Didona et~al\mbox{.}}{2018}]%
        {DBLP:journals/pvldb/DidonaGWZ18}
\bibfield{author}{\bibinfo{person}{Diego Didona}, \bibinfo{person}{Rachid
  Guerraoui}, \bibinfo{person}{Jingjing Wang}, {and} \bibinfo{person}{Willy
  Zwaenepoel}.} \bibinfo{year}{2018}\natexlab{}.
\newblock \showarticletitle{Causal Consistency and Latency Optimality: Friend
  or Foe?}
\newblock \bibinfo{journal}{\emph{Proc. {VLDB} Endow.}} \bibinfo{volume}{11},
  \bibinfo{number}{11} (\bibinfo{year}{2018}), \bibinfo{pages}{1618--1632}.
\newblock


\bibitem[\protect\citeauthoryear{Gan, Ren, Ripberger, Blanas, and Wang}{Gan
  et~al\mbox{.}}{2020}]%
        {IsoDiff:VLDB2020}
\bibfield{author}{\bibinfo{person}{Yifan Gan}, \bibinfo{person}{Xueyuan Ren},
  \bibinfo{person}{Drew Ripberger}, \bibinfo{person}{Spyros Blanas}, {and}
  \bibinfo{person}{Yang Wang}.} \bibinfo{year}{2020}\natexlab{}.
\newblock \showarticletitle{IsoDiff: Debugging Anomalies Caused by Weak
  Isolation}.
\newblock \bibinfo{journal}{\emph{Proc. VLDB Endow.}} \bibinfo{volume}{13},
  \bibinfo{number}{12} (\bibinfo{date}{July} \bibinfo{year}{2020}),
  \bibinfo{pages}{2773--2786}.
\newblock
\showISSN{2150-8097}
\urldef\tempurl%
\url{https://doi.org/10.14778/3407790.3407860}
\showDOI{\tempurl}


\bibitem[\protect\citeauthoryear{Graphviz}{Graphviz}{2022}]%
        {graphviz}
\bibfield{author}{\bibinfo{person}{Graphviz}.} \bibinfo{year}{Accessed
  December, 2022}\natexlab{}.
\newblock \bibinfo{title}{Open source graph visualization software}.
\newblock
\newblock
\newblock
\shownote{\url{https://graphviz.org/}.}


\bibitem[\protect\citeauthoryear{Helt, Burke, Levy, and Lloyd}{Helt
  et~al\mbox{.}}{2021}]%
        {DBLP:conf/sosp/Helt0LL21}
\bibfield{author}{\bibinfo{person}{Jeffrey Helt}, \bibinfo{person}{Matthew
  Burke}, \bibinfo{person}{Amit Levy}, {and} \bibinfo{person}{Wyatt Lloyd}.}
  \bibinfo{year}{2021}\natexlab{}.
\newblock \showarticletitle{Regular Sequential Serializability and Regular
  Sequential Consistency}. In \bibinfo{booktitle}{\emph{{SOSP}'21}}.
  \bibinfo{publisher}{{ACM}}, \bibinfo{pages}{163--179}.
\newblock


\bibitem[\protect\citeauthoryear{Huang, Liu, Chen, Wei, Basin, Li, and
  Pan}{Huang et~al\mbox{.}}{2023}]%
        {PolySI:VLDB2023}
\bibfield{author}{\bibinfo{person}{Kaile Huang}, \bibinfo{person}{Si Liu},
  \bibinfo{person}{Zhenge Chen}, \bibinfo{person}{Hengfeng Wei},
  \bibinfo{person}{David Basin}, \bibinfo{person}{Haixiang Li}, {and}
  \bibinfo{person}{Anqun Pan}.} \bibinfo{year}{2023}\natexlab{}.
\newblock \showarticletitle{Efficient Black-Box Checking of Snapshot Isolation
  in Databases}.
\newblock \bibinfo{journal}{\emph{Proc. VLDB Endow.}} \bibinfo{volume}{16},
  \bibinfo{number}{6} (\bibinfo{date}{apr} \bibinfo{year}{2023}),
  \bibinfo{pages}{1264–1276}.
\newblock
\showISSN{2150-8097}
\urldef\tempurl%
\url{https://doi.org/10.14778/3583140.3583145}
\showDOI{\tempurl}


\bibitem[\protect\citeauthoryear{Huang, Liu, Chen, Wei, Basin, Li, and
  Pan}{Huang et~al\mbox{.}}{2022}]%
        {elle-bug}
\bibfield{author}{\bibinfo{person}{Kaile Huang}, \bibinfo{person}{Si Liu},
  \bibinfo{person}{Zhenge Chen}, \bibinfo{person}{Hengfeng Wei},
  \bibinfo{person}{David Basin}, \bibinfo{person}{Haixiang Li}, {and}
  \bibinfo{person}{Anqun Pan}.} \bibinfo{year}{Accessed December,
  2022}\natexlab{}.
\newblock \bibinfo{title}{Issue \#17}.
\newblock
\newblock
\newblock
\shownote{\url{https://github.com/jepsen-io/elle/issues/17}.}


\bibitem[\protect\citeauthoryear{Jepsen}{Jepsen}{2022a}]%
        {jepsen}
\bibfield{author}{\bibinfo{person}{Jepsen}.} \bibinfo{year}{Accessed August,
  2022}\natexlab{a}.
\newblock
\newblock
\newblock
\shownote{\url{https://jepsen.io}.}


\bibitem[\protect\citeauthoryear{Jepsen}{Jepsen}{2022b}]%
        {YugabyteDB-bug}
\bibfield{author}{\bibinfo{person}{Jepsen}.} \bibinfo{year}{Accessed August,
  2022}\natexlab{b}.
\newblock \bibinfo{title}{Issue \#824}.
\newblock
\newblock
\newblock
\shownote{\url{https://github.com/YugaByte/yugabyte-db/issues/824}.}


\bibitem[\protect\citeauthoryear{Kallen}{Kallen}{2022}]%
        {ctwitter}
\bibfield{author}{\bibinfo{person}{Nick Kallen}.} \bibinfo{year}{Accessed
  August, 2022}\natexlab{}.
\newblock \bibinfo{title}{Big Data in Real Time at Twitter}.
\newblock
\newblock
\newblock
\shownote{\url{https://www.infoq.com/presentations/Big-Data-in-Real-Time-at-Twitter/}.}


\bibitem[\protect\citeauthoryear{Kingsbury and Alvaro}{Kingsbury and
  Alvaro}{2020}]%
        {Elle:VLDB2020}
\bibfield{author}{\bibinfo{person}{Kyle Kingsbury} {and} \bibinfo{person}{Peter
  Alvaro}.} \bibinfo{year}{2020}\natexlab{}.
\newblock \showarticletitle{Elle: Inferring Isolation Anomalies from
  Experimental Observations}.
\newblock \bibinfo{journal}{\emph{Proc. VLDB Endow.}} \bibinfo{volume}{14},
  \bibinfo{number}{3} (\bibinfo{date}{Nov.} \bibinfo{year}{2020}),
  \bibinfo{pages}{268--280}.
\newblock
\showISSN{2150-8097}


\bibitem[\protect\citeauthoryear{Lamport}{Lamport}{1978}]%
        {DBLP:journals/cacm/Lamport78}
\bibfield{author}{\bibinfo{person}{Leslie Lamport}.}
  \bibinfo{year}{1978}\natexlab{}.
\newblock \showarticletitle{Time, Clocks, and the Ordering of Events in a
  Distributed System}.
\newblock \bibinfo{journal}{\emph{Commun. {ACM}}} \bibinfo{volume}{21},
  \bibinfo{number}{7} (\bibinfo{year}{1978}), \bibinfo{pages}{558--565}.
\newblock


\bibitem[\protect\citeauthoryear{Liu, {\"{O}}lveczky, Zhang, Wang, and
  Meseguer}{Liu et~al\mbox{.}}{2019}]%
        {Maude:Liu2019}
\bibfield{author}{\bibinfo{person}{Si Liu}, \bibinfo{person}{Peter~Csaba
  {\"{O}}lveczky}, \bibinfo{person}{Min Zhang}, \bibinfo{person}{Qi Wang},
  {and} \bibinfo{person}{Jos{\'{e}} Meseguer}.}
  \bibinfo{year}{2019}\natexlab{}.
\newblock \showarticletitle{Automatic Analysis of Consistency Properties of
  Distributed Transaction Systems in {M}aude}. In
  \bibinfo{booktitle}{\emph{{TACAS} 2019}} \emph{(\bibinfo{series}{LNCS})},
  Vol.~\bibinfo{volume}{11428}. \bibinfo{publisher}{Springer},
  \bibinfo{pages}{40--57}.
\newblock


\bibitem[\protect\citeauthoryear{Lloyd, Freedman, Kaminsky, and Andersen}{Lloyd
  et~al\mbox{.}}{2013}]%
        {Eiger:NSDI2013}
\bibfield{author}{\bibinfo{person}{Wyatt Lloyd}, \bibinfo{person}{Michael~J.
  Freedman}, \bibinfo{person}{Michael Kaminsky}, {and}
  \bibinfo{person}{David~G. Andersen}.} \bibinfo{year}{2013}\natexlab{}.
\newblock \showarticletitle{Stronger semantics for low-latency geo-replicated
  storage}. In \bibinfo{booktitle}{\emph{NSDI' 13}}. \bibinfo{publisher}{USENIX
  Association}, \bibinfo{pages}{313--328}.
\newblock
\showISBNx{978-1-931971-00-3}


\bibitem[\protect\citeauthoryear{Lu, Sen, and Lloyd}{Lu et~al\mbox{.}}{2020}]%
        {DBLP:conf/osdi/LuSL20}
\bibfield{author}{\bibinfo{person}{Haonan Lu}, \bibinfo{person}{Siddhartha
  Sen}, {and} \bibinfo{person}{Wyatt Lloyd}.} \bibinfo{year}{2020}\natexlab{}.
\newblock \showarticletitle{Performance-Optimal Read-Only Transactions}. In
  \bibinfo{booktitle}{\emph{{OSDI} 2020}}. \bibinfo{publisher}{{USENIX}
  Association}, \bibinfo{pages}{333--349}.
\newblock


\bibitem[\protect\citeauthoryear{MongoDB}{MongoDB}{2022}]%
        {MongoDB}
\bibfield{author}{\bibinfo{person}{MongoDB}.} \bibinfo{year}{Accessed August,
  2022}\natexlab{}.
\newblock
\newblock
\newblock
\shownote{\url{https://www.mongodb.com/}.}


\bibitem[\protect\citeauthoryear{Papadimitriou}{Papadimitriou}{1979}]%
        {SER:JACM1979}
\bibfield{author}{\bibinfo{person}{Christos~H. Papadimitriou}.}
  \bibinfo{year}{1979}\natexlab{}.
\newblock \showarticletitle{The Serializability of Concurrent Database
  Updates}.
\newblock \bibinfo{journal}{\emph{J. ACM}} \bibinfo{volume}{26},
  \bibinfo{number}{4} (\bibinfo{date}{oct} \bibinfo{year}{1979}),
  \bibinfo{pages}{631--653}.
\newblock
\showISSN{0004-5411}
\urldef\tempurl%
\url{https://doi.org/10.1145/322154.322158}
\showDOI{\tempurl}


\bibitem[\protect\citeauthoryear{Peng and Dabek}{Peng and Dabek}{2010}]%
        {Percolator:OSDI2010}
\bibfield{author}{\bibinfo{person}{Daniel Peng} {and} \bibinfo{person}{Frank
  Dabek}.} \bibinfo{year}{2010}\natexlab{}.
\newblock \showarticletitle{Large-Scale Incremental Processing Using
  Distributed Transactions and Notifications}. In
  \bibinfo{booktitle}{\emph{OSDI'10}}. \bibinfo{publisher}{USENIX Association},
  \bibinfo{address}{USA}, \bibinfo{pages}{251--264}.
\newblock


\bibitem[\protect\citeauthoryear{PostgreSQL}{PostgreSQL}{2022}]%
        {postgresql-rr}
\bibfield{author}{\bibinfo{person}{PostgreSQL}.} \bibinfo{year}{Accessed
  August, 2022}\natexlab{}.
\newblock \bibinfo{title}{Transaction Isolation}.
\newblock
\newblock
\newblock
\shownote{\url{https://www.postgresql.org/docs/current/transaction-iso.html}.}


\bibitem[\protect\citeauthoryear{RUBiS}{RUBiS}{2022}]%
        {crubis}
\bibfield{author}{\bibinfo{person}{RUBiS}.} \bibinfo{year}{Accessed August,
  2022}\natexlab{}.
\newblock \bibinfo{title}{Auction Site for e-Commerce Technologies
  Benchmarking}.
\newblock
\newblock
\newblock
\shownote{\url{https://projects.ow2.org/view/rubis/}.}


\bibitem[\protect\citeauthoryear{Server}{Server}{2022}]%
        {SQLServer}
\bibfield{author}{\bibinfo{person}{Microsoft~SQL Server}.}
  \bibinfo{year}{Accessed August, 2022}\natexlab{}.
\newblock
\newblock
\newblock
\shownote{\url{https://www.microsoft.com/en-us/sql-server/}.}


\bibitem[\protect\citeauthoryear{Shang, Yu, and Elmore}{Shang
  et~al\mbox{.}}{2018}]%
        {RushMon:SIGMOD2018}
\bibfield{author}{\bibinfo{person}{Zechao Shang}, \bibinfo{person}{Jeffrey~Xu
  Yu}, {and} \bibinfo{person}{Aaron~J. Elmore}.}
  \bibinfo{year}{2018}\natexlab{}.
\newblock \showarticletitle{RushMon: Real-Time Isolation Anomalies Monitoring}.
  In \bibinfo{booktitle}{\emph{SIGMOD '18}}. \bibinfo{publisher}{ACM},
  \bibinfo{pages}{647--662}.
\newblock
\showISBNx{9781450347037}
\urldef\tempurl%
\url{https://doi.org/10.1145/3183713.3196932}
\showDOI{\tempurl}


\bibitem[\protect\citeauthoryear{Sovran, Power, Aguilera, and Li}{Sovran
  et~al\mbox{.}}{2011}]%
        {PSI:SOSP2011}
\bibfield{author}{\bibinfo{person}{Yair Sovran}, \bibinfo{person}{Russell
  Power}, \bibinfo{person}{Marcos~K. Aguilera}, {and} \bibinfo{person}{Jinyang
  Li}.} \bibinfo{year}{2011}\natexlab{}.
\newblock \showarticletitle{Transactional Storage for Geo-Replicated Systems}.
  In \bibinfo{booktitle}{\emph{SOSP '11}}. \bibinfo{publisher}{ACM},
  \bibinfo{pages}{385--400}.
\newblock
\showISBNx{9781450309776}
\urldef\tempurl%
\url{https://doi.org/10.1145/2043556.2043592}
\showDOI{\tempurl}


\bibitem[\protect\citeauthoryear{Tan, Zhao, Mu, and Walfish}{Tan
  et~al\mbox{.}}{2020}]%
        {Cobra:OSDI2020}
\bibfield{author}{\bibinfo{person}{Cheng Tan}, \bibinfo{person}{Changgeng
  Zhao}, \bibinfo{person}{Shuai Mu}, {and} \bibinfo{person}{Michael Walfish}.}
  \bibinfo{year}{2020}\natexlab{}.
\newblock \showarticletitle{COBRA: Making Transactional Key-Value Stores
  Verifiably Serializable}. In \bibinfo{booktitle}{\emph{OSDI'20}}. Article
  \bibinfo{articleno}{4}, \bibinfo{numpages}{18}~pages.
\newblock
\showISBNx{978-1-939133-19-9}


\bibitem[\protect\citeauthoryear{Terry, Demers, Petersen, Spreitzer, Theimer,
  and Welch}{Terry et~al\mbox{.}}{1994}]%
        {TerryDPSTW94}
\bibfield{author}{\bibinfo{person}{Douglas~B. Terry}, \bibinfo{person}{Alan~J.
  Demers}, \bibinfo{person}{Karin Petersen}, \bibinfo{person}{Mike Spreitzer},
  \bibinfo{person}{Marvin Theimer}, {and} \bibinfo{person}{Brent~B. Welch}.}
  \bibinfo{year}{1994}\natexlab{}.
\newblock \showarticletitle{Session Guarantees for Weakly Consistent Replicated
  Data}. In \bibinfo{booktitle}{\emph{{PDIS}}}. \bibinfo{publisher}{{IEEE}
  Computer Society}, \bibinfo{pages}{140--149}.
\newblock


\bibitem[\protect\citeauthoryear{testing~of MongoDB~4.2.6}{testing~of
  MongoDB~4.2.6}{2022}]%
        {MongoDB-Jepsen}
\bibfield{author}{\bibinfo{person}{Jepsen testing~of MongoDB~4.2.6}.}
  \bibinfo{year}{Accessed August, 2022}\natexlab{}.
\newblock
\newblock
\newblock
\shownote{\url{http://jepsen.io/analyses/mongodb-4.2.6}.}


\bibitem[\protect\citeauthoryear{testing~of TiDB~2.1.7}{testing~of
  TiDB~2.1.7}{2022}]%
        {TiDB-Jepsen}
\bibfield{author}{\bibinfo{person}{Jepsen testing~of TiDB~2.1.7}.}
  \bibinfo{year}{Accessed August, 2022}\natexlab{}.
\newblock
\newblock
\newblock
\shownote{\url{https://jepsen.io/analyses/tidb-2.1.7}.}


\bibitem[\protect\citeauthoryear{TiDB}{TiDB}{2022}]%
        {TiDB}
\bibfield{author}{\bibinfo{person}{TiDB}.} \bibinfo{year}{Accessed August,
  2022}\natexlab{}.
\newblock
\newblock
\newblock
\shownote{\url{https://en.pingcap.com/tidb/}.}


\bibitem[\protect\citeauthoryear{TPC}{TPC}{2022}]%
        {tpcc}
\bibfield{author}{\bibinfo{person}{TPC}.} \bibinfo{year}{Accessed August,
  2022}\natexlab{}.
\newblock \bibinfo{title}{{TPC-C}: On-Line Transaction Processing Benchmark}.
\newblock
\newblock
\newblock
\shownote{\url{https://www.tpc.org/tpcc/}.}


\bibitem[\protect\citeauthoryear{Warszawski and Bailis}{Warszawski and
  Bailis}{2017}]%
        {ACIDRain:SIGMOD2017}
\bibfield{author}{\bibinfo{person}{Todd Warszawski} {and}
  \bibinfo{person}{Peter Bailis}.} \bibinfo{year}{2017}\natexlab{}.
\newblock \showarticletitle{ACIDRain: Concurrency-Related Attacks on
  Database-Backed Web Applications}. In \bibinfo{booktitle}{\emph{{SIGMOD}
  2017}}. \bibinfo{publisher}{{ACM}}, \bibinfo{pages}{5--20}.
\newblock
\urldef\tempurl%
\url{https://doi.org/10.1145/3035918.3064037}
\showDOI{\tempurl}


\bibitem[\protect\citeauthoryear{Xiong, Cerone, Raad, and Gardner}{Xiong
  et~al\mbox{.}}{2020}]%
        {CentralisedSemantics:ECOOP2020}
\bibfield{author}{\bibinfo{person}{Shale Xiong}, \bibinfo{person}{Andrea
  Cerone}, \bibinfo{person}{Azalea Raad}, {and} \bibinfo{person}{Philippa
  Gardner}.} \bibinfo{year}{2020}\natexlab{}.
\newblock \showarticletitle{{Data Consistency in Transactional Storage Systems:
  A Centralised Semantics}}. In \bibinfo{booktitle}{\emph{ECOOP'20}},
  Vol.~\bibinfo{volume}{166}. \bibinfo{pages}{21:1--21:31}.
\newblock
\urldef\tempurl%
\url{https://doi.org/10.4230/LIPIcs.ECOOP.2020.21}
\showDOI{\tempurl}


\bibitem[\protect\citeauthoryear{YugabyteDB}{YugabyteDB}{2022}]%
        {YugabyteDB}
\bibfield{author}{\bibinfo{person}{YugabyteDB}.} \bibinfo{year}{Accessed
  August, 2022}\natexlab{}.
\newblock
\newblock
\newblock
\shownote{\url{https://www.yugabyte.com/}.}


\bibitem[\protect\citeauthoryear{Zellag and Kemme}{Zellag and Kemme}{2014}]%
        {ConsAD:VLDB2014}
\bibfield{author}{\bibinfo{person}{Kamal Zellag} {and} \bibinfo{person}{Bettina
  Kemme}.} \bibinfo{year}{2014}\natexlab{}.
\newblock \showarticletitle{Consistency anomalies in multi-tier architectures:
  automatic detection and prevention}.
\newblock \bibinfo{journal}{\emph{{VLDB} J.}} \bibinfo{volume}{23},
  \bibinfo{number}{1} (\bibinfo{year}{2014}), \bibinfo{pages}{147--172}.
\newblock
\urldef\tempurl%
\url{https://doi.org/10.1007/s00778-013-0318-x}
\showDOI{\tempurl}


\bibitem[\protect\citeauthoryear{Zennou, Biswas, Bouajjani, Enea, and
  Erradi}{Zennou et~al\mbox{.}}{2019}]%
        {DBLP:conf/netys/ZennouBBEE19}
\bibfield{author}{\bibinfo{person}{Rachid Zennou}, \bibinfo{person}{Ranadeep
  Biswas}, \bibinfo{person}{Ahmed Bouajjani}, \bibinfo{person}{Constantin
  Enea}, {and} \bibinfo{person}{Mohammed Erradi}.}
  \bibinfo{year}{2019}\natexlab{}.
\newblock \showarticletitle{Checking Causal Consistency of Distributed
  Databases}. In \bibinfo{booktitle}{\emph{{NETYS} 2019}}
  \emph{(\bibinfo{series}{LNCS})}, Vol.~\bibinfo{volume}{11704}.
  \bibinfo{publisher}{Springer}, \bibinfo{pages}{35--51}.
\newblock
\urldef\tempurl%
\url{https://doi.org/10.1007/978-3-030-31277-0\_3}
\showURL{%
\tempurl}


\end{thebibliography}
